\documentclass[a4paper, oneside, english,11pt]{report}
\usepackage[utf8]{inputenc}
\usepackage[T1]{fontenc}
\usepackage{graphicx}
\usepackage{hyperref}
\usepackage{booktabs,caption}
\usepackage[normalem]{ulem}
\usepackage{placeins}
\usepackage{amsmath}
\usepackage{amssymb}

\usepackage[left=2.cm, right=2.cm, top=2.5cm, bottom=2.cm, head=15pt]{geometry} 
\usepackage{setspace}
\usepackage{pdfpages}
\usepackage{fancyhdr}
\usepackage[numbib,notlof,notlot,nottoc]{tocbibind}

\usepackage{mathptmx}
\usepackage{subcaption}
\usepackage[backend=bibtex,style=ieee,natbib=true,isbn=false,doi=false]{biblatex}
\usepackage{adjustbox}
\usepackage[autostyle=true]{csquotes} 

\DeclareFieldFormat{pages}{#1}

\usepackage{floatrow}
\usepackage{amsmath}	
\usepackage{amssymb}	
\usepackage{multicol}        
\usepackage{bm}		
\usepackage{pdflscape}	
\usepackage{multirow}
\usepackage{verbatim}
\usepackage[outdir=./]{epstopdf}
\usepackage{braket}

\usepackage{xcolor}
\definecolor{blue(ryb)}{rgb}{0.01, 0.28, 1.0}
\definecolor{cobalt}{rgb}{0.0, 0.28, 0.67}
\definecolor{brickred}{rgb}{0.8, 0.25, 0.33}
\usepackage{ae,aecompl}
\hypersetup{pdfpagemode={UseOutlines},
bookmarksopen=true,
bookmarksopenlevel=0,
hypertexnames=true,
colorlinks=true,
citecolor=magenta,
pdfstartview={FitV},
unicode,
breaklinks=true,
filecolor=cobalt,
linkcolor=blue(ryb),
urlcolor=brickred
}
\renewbibmacro{in:}{%
  \ifboolexpr{%
     test {\ifentrytype{article}}%
     or
     test {\ifentrytype{inproceedings}}%
  }{}{\printtext{\bibstring{in}\intitlepunct}}%
}

\usepackage{amsthm}

\newtheorem{theorem}{Theorem}

\newtheorem{lemma}[theorem]{Lemma}
\newtheorem{definition}[theorem]{Definition}
\newtheorem{remark}[theorem]{Remark}
\newtheorem{corollary}[theorem]{Corollary}
\newtheorem{conjecture}[theorem]{Conjecture}
\newtheorem{proposition}[theorem]{Proposition}

\usepackage{pifont}
\def\club{\ding{168}}
\def\diamond{\color{red}\ding{169}}
\def\heart{\color{red}\ding{170}}
\def\spade{\ding{171}}
\def\flower{\color{blue}\ding{95}}
\def\star{\color{blue}\ding{87}}

\usepackage{tikz}

\usepackage{mathtools}

\usepackage{array}
\newcommand{\PreserveBackslash}[1]{\let\temp=\\#1\let\\=\temp}
\newcolumntype{C}[1]{>{\PreserveBackslash\centering}p{#1}}
\newcolumntype{R}[1]{>{\PreserveBackslash\raggedleft}p{#1}}
\newcolumntype{L}[1]{>{\PreserveBackslash\raggedright}p{#1}}

\usepackage{makecell}

\usepackage[toc,page]{appendix}

\usepackage{pbox}

\usepackage{diagbox}

\renewcommand\bra[1]{{\langle{#1}|}}
\renewcommand\ket[1]{{|{#1}\rangle}}

\usetikzlibrary{quotes,angles}

\usepackage{CJKutf8}
\newenvironment{Japanese}{%
  \CJKfamily{min}%
  \CJKtilde
  \CJKnospace}{}
  
\usepackage{centernot}
\usepackage{mathtools}
\usepackage{ stmaryrd }

\usetikzlibrary{arrows}

\def\blankpage{%
      \clearpage%
      \thispagestyle{empty}%
      \addtocounter{page}{-1}%
      \null%
      \clearpage}
\usepackage[autostyle=true]{csquotes} 

\title{Quantum mappings and designs}

\date{September, 2020}
\newcommand{\thesisauthor}{Grzegorz Rajchel-Mieldzio{\'c}
} 
\newcommand{\supervisor}{Prof. Karol {\.Z}yczkowski}

\begin{document}

\makeatletter
\begin{titlepage}
    \begin{center}
        \hrule
        \vspace*{1cm}
        \Huge
        \textbf{\@title}
        \vspace{1cm}
        \hrule
        \vspace{0.5cm}
        \begin{figure}[htbp]
             \centering
             \includegraphics[width=.3\linewidth]{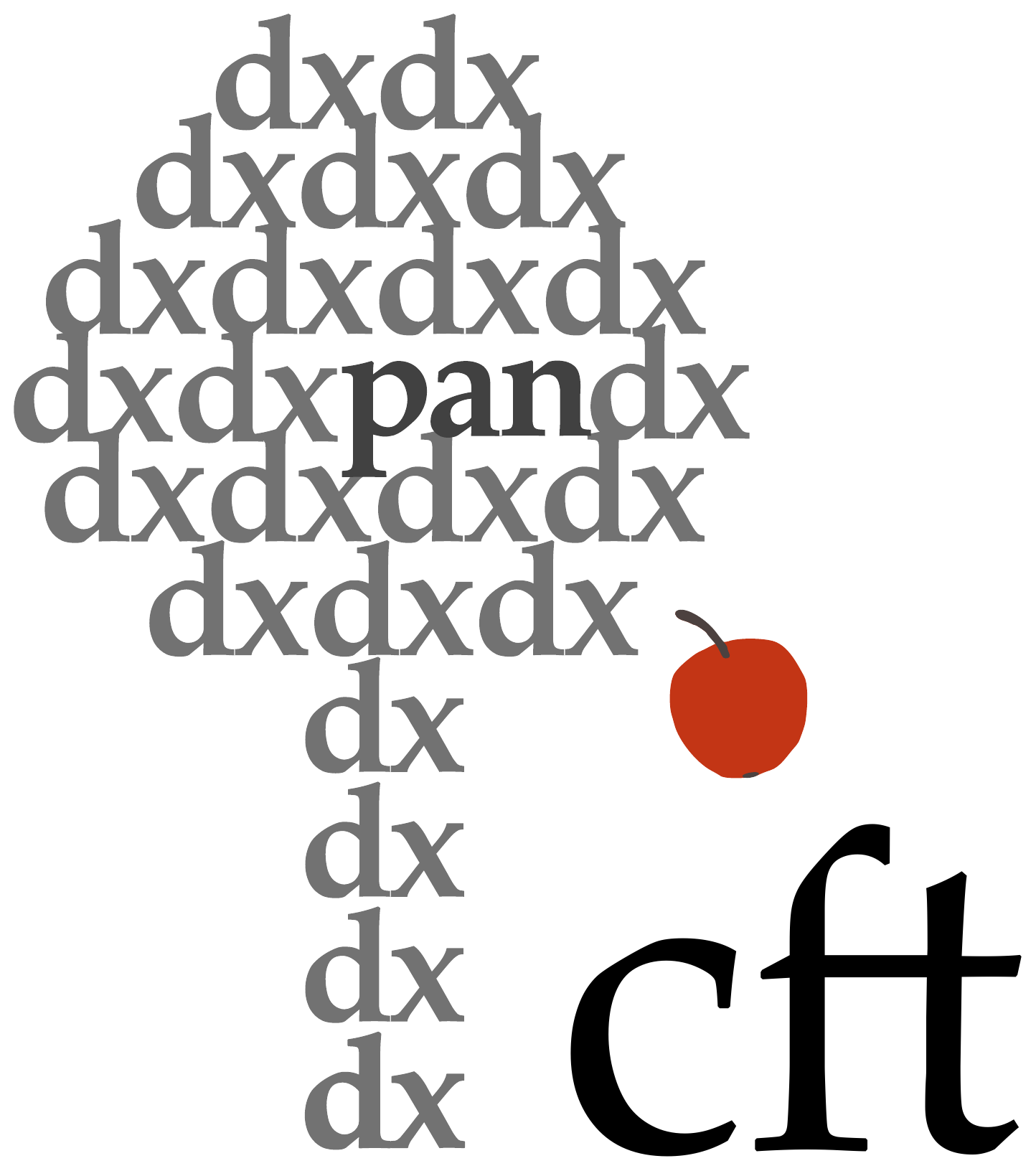}
        \end{figure}

        \vspace{1.5cm}

        \large
        \textbf{Author}: \Large{\thesisauthor{}}\\
        \large
        \textbf{Supervisor}: \Large{\supervisor{}}\\

        \vspace{1cm}
        \Large{\textbf {Centrum Fizyki Teoretycznej Polskiej Akademii Nauk}}\\
        \vspace{1cm}
       \textit{A thesis submitted in partial fulfillment of the requirements for the degree of}\\
       \textit{Doctor of Philosophy in Physics}\\
       \textit{}\\
        
        \vspace{1cm}
        {\Large September 2021}

    \end{center}
\end{titlepage}
\makeatother
\blankpage{}
\thispagestyle{empty}%
\topskip0pt
\vspace*{\fill}
\begin{center}
    \textit{Dedicated to my wife}\\
  \textit{for making me who I am today.}
\end{center}
\vspace*{\fill}
\blankpage{}
\pagenumbering{roman}
\vspace{-1cm}
\begin{center}
    \section*{Abstract}
\end{center}
\label{sec:abstract}
Quantum information emerged in the $20^{\mathrm{th}}$ century, fostering the rapid development of technologies that use the description of microscopic systems provided by quantum mechanics.
The progress of this domain of science is a result of the efforts devoted to the study of phenomena not existing in the macroscopic domain.
This includes the research of quantum protocols that use superposition, as well as quantum entanglement.
Both of these properties of quantum systems make them compelling from the point of view of the creation of quantum devices that outperform the classical ones.

Employing these characteristics of the quantum domain, scientists and inventors have the ultimate goal of creating a quantum computer of a size that shall make it useful for computations that are not feasible on the classical computers that apply the Boolean logic. 
Nonetheless, to apply the physical principles that could lead experimentalists to succeed in their search for a groundbreaking quantum device, we need to understand the intricacies of the theoretical description.
Therefore, the main goal of this thesis is to provide novel constructions useful for the comprehension of quantum mechanics from the perspective of mappings and designs.

The thesis is designed as follows.
First, in the preliminary chapter, we introduce the necessary concepts from the field of quantum information, including quantum designs and mappings.
In the subsequent part of the thesis, we investigate two instances of quantum mappings that were developed by the author and his collaborators.

The first one concerns the unistochasticity problem, which relates bistochastic and unitary matrices, both useful in the classical and the quantum domain, respectively.
This fragment of the thesis dwells on the characterization of the unistochastic set, with a presentation of the algorithm that allows determining whether a given bistochastic matrix of size 4 is unistochastic.
It contains also the proof that the simple bracelet condition is sufficient to decide the unistochasticity of a circulant matrix of size 4.
Furthermore, we investigate the unistochasticity of certain sets inside the Birkhoff polytope of bistochastic matrices of an arbitrary dimension $N$ and prove that the rays and counter-rays are unistochastic, provided there exists a robust Hadamard matrix of dimension $N$.

Moving on to the second instance of quantum mappings, we study the entangling power in the multipartite case.
The main achievement shown in this chapter is an explicit analytical formula for the average entangling power of a tripartite orthogonal gate.

Using the notion of entangling power, we develop in the subsequent chapter novel ideas concerning the search for an absolutely maximally entangled (AME) state of four quhexes, the existence of which was first shown by the author and his collaborators.
Several new methods in the search for AME states and other quantum designs are presented.
In particular, we derived the Hessian for the entangling power and the average singular entropy of a unitary quantum gate.
These can be used to evaluate the extremality of the solutions found.
Furthermore, our research revealed a curious block-like structure of the newly discovered AME state of 4 subsystems 6 levels each, which has the potential to disclose more general facts about other quantum designs.

The final chapter concerning new research is devoted to an extension of the recently established notion of quantum Sudoku (SudoQ) designs.
To characterize such objects we introduced the cardinality of a SudoQ as the number of different states forming the design.
We characterized the cardinality also as a measure of ``quantumness'' of quantum Latin squares.
Those of the highest cardinality yield families of quantum measurements of special properties.
We characterize the problem in the general case of size $N^2$.
A connection between SudoQ designs and mutually unbiased bases is demonstrated.

\clearpage

\vspace{-1cm}
\begin{center}
    \section*{Streszczenie}
\end{center}
\label{sec:pol_abstract}
Informacja kwantowa rozwinęła się w XX wieku, przyczyniając się do szybkiego postępu technologii wykorzystujących kwantowo-mechaniczny opis układów mikroskopowych.
Postęp tej dziedziny nauki jest związany z badaniem zjawisk nieistniejących w świecie makroskopowym takich jak protokoły kwantowe wykorzystujące superpozycję oraz splątanie kwantowe.
Oba zjawiska sprawiają, że układy kwantowe są atrakcyjne z perspektywy tworzenia urządzeń kwantowych dających przewagę nad odpowiednikami klasycznymi.

Wykorzystując świat kwantów, finalnym celem naukowców i wynalazców jest komputer kwantowy o rozmiarze, który uczyni go przydatnym do obliczeń niemożliwych do wykonania na klasycznych komputerach stosujących algebrę Boole'a.
Aby użyć fizyki do znalezienia przełomowego urządzenia kwantowego, trzeba zrozumieć zawiłości opisu teoretycznego.
Z tego powodu, celem tej rozprawy jest dostarczenie nowych konstrukcji przydatnych do zrozumienia mechaniki kwantowej z perspektywy odwzorowań i deseni.

Niniejsza rozprawa jest ułożona w następujący sposób.
W rozdziale wstępnym wprowadzono niezbędne pojęcia z dziedziny informacji kwantowej, w tym desenie kwantowe i odwzorowania.
W dalszej części pracy badano dwa przykłady odwzorowań kwantowych.

Pierwszy przypadek dotyczy problemu unistochastyczności łączącego macierze bistochastyczne i unitarne, przydatne odpowiednio w dziedzinie klasycznej i kwantowej.
Skupiono się na charakterystyce zbioru macierzy unistochastycznych wraz z przedstawieniem algorytmu pozwalającego określić, czy dana macierz bistochastyczna o wymiarze 4 jest unistochastyczna.
Ponadto udowodniono, iż prosty warunek łańcuszkowy jest wystarczający do unistochastyczności macierzy cyrkulantnej o rozmiarze 4.
Zbadano unistochastyczność podzbiorów wielościanu Birkhoffa macierzy bistochastycznych o dowolnym wymiarze $N$ i udowodniono, że promienie i przeciwpromienie są unistochastyczne, pod warunkiem istnienia stabilnej (\emph{robust}) macierzy Hadamarda w wymiarze $N$.

Kolejny rozdział rozprawy opisuje moc plączącą (\emph{entangling power}) w przypadku wielocząstkowym.
Głównym osiągnięciem przedstawionym w tym rozdziale jest jawny wzór analityczny na średnią moc plączącą trójcząstkowej bramki ortogonalnej.

Posługując się pojęciem mocy plączącej, w następnym rozdziale rozwijano nowatorskie pomysły dotyczące poszukiwania absolutnie maksymalnie splątanego (\emph{absolutely maximally entangled}, AME) stanu czterech układów z 6 poziomami, którego istnienie po raz pierwszy zostało wykazane przez autora i współpracowników.
Przedstawiono kilka nowych metod poszukiwania żądanego stanu, z których część może okazać się przydatna w badaniu innych układów kwantowych.
W szczególności, wyprowadzono hesjan dla mocy plączącej i sumy entropii,
a następnie wykorzystano go do badania ekstremalności znalezionych rozwiązań.
Co więcej, badania ujawniły blokową strukturę nowo odkrytego stanu AME, która może być przydatna przy szukaniu innych deseni kwantowych.

Ostatni rozdział dotyczy rozszerzenia niedawno wprowadzonego pojęcia kwantowych deseni Sudoku (SudoQ).
Aby je scharakteryzować, zdefiniowano kardynalność SudoQ, czyli liczbę różnych stanów tworzących deseń.
Kardynalność, jako miara „kwantowości”, może być istotna również przy badaniu kwantowych kwadratów łacińskich.
Takie układy o najwyższej kardynalności implikują rodziny pomiarów kwantowych o specjalnej symetrii.
Scharakteryzowano ogólny przypadek wymiaru $N^2$, a ponadto wykazano związek pomiędzy kwantowymi deseniami SudoQ a bazami wzajemnie nieobciążonymi (\emph{mutually unbiased bases}, MUBs).
\clearpage

\vspace{-1cm}
\begin{center}
    \section*{Declaration}
\end{center}
\label{sec:declaration}
The work described in this thesis was undertaken between April 2016 and April 2021 while the author was a research student under the supervision of Prof. Karol {\.Z}yczkowski at the Center for Theoretical Physics, Polish Academy of Sciences and completed his coursework between October 2018 and June 2021 at the Institute of Physics, Polish Academy of Sciences. 
No part of this thesis has been submitted for any other degree at the Center for Theoretical Physics, Polish Academy of Sciences or any other scientific institution.
\vspace{0.02\textwidth}

The thesis forms a monograph; nonetheless, some parts of the content included in the consecutive chapters of the thesis have appeared in the following papers:
\vspace{0.01\textwidth}

\begin{enumerate}
    \item Chapter~\ref{chapter_3}: \textbf{G. Rajchel}, A. G{\k{a}}siorowski, and K. \.Zyczkowski, \emph{Robust Hadamard matrices, unistochastic rays in Birkhoff polytope and equi-entangled bases in composite spaces}, Mathematics in Computer Science, vol. 12, pp. 473-490, 2018,
    
    \item Chapter~\ref{chapter_3}: \textbf{G. Rajchel-Mieldzio{\'c}}, K. Korzekwa, Z. Pucha{\l}a, and K. \.Zyczkowski, \emph{Algebraic and geometric structures inside the Birkhoff polytope}, arXiv: 2101.11288 [quant-ph], 2021,
    
    \item Chapter~\ref{chapter_4}: T. Linowski, \textbf{G. Rajchel-Mieldzio{\'c}}, and K. \.Zyczkowski, \emph{Entangling power of multipartite unitary gates}, Journal of Physics A: Mathematical and Theoretical, vol. 53, pp. 125-303, 2020,
    
    \item Chapter~\ref{chapter_6}: S. A. Rather, A. Burchardt, W. Bruzda, \textbf{G. Rajchel-Mieldzio{\'c}}, A. Lakshminarayan, and K. \.Zyczkowski, \emph{Thirty-six entangled officers of Euler}, arXiv: 2104.05122 [quant-ph], 2021,
    
    \item Chapter~\ref{chapter_7}: J. Paczos, M. Wierzbi{\'n}ski, \textbf{G. Rajchel-Mieldzio{\'c}}, A. Burchardt, and K. \.Zyczkowski, \emph{Genuinely quantum SudoQ and its cardinality}, arXiv: 2106.02967 [quant-ph], 2021. Provisionally accepted at Phys. Rev. A.
\end{enumerate}

\vspace{0.02\textwidth}
The work described in Chapters~\ref{chapter_3}-\ref{chapter_7} was performed in collaboration with the other coauthors listed above.
A detailed description of the contribution of the author is provided in Section~\ref{sec:structure_contribution}.

\vspace{0.02\textwidth}
In addition to the work presented in this thesis, the author has also worked on the following papers:
\begin{enumerate}
    \item[6.] P. T. Grochowski, \textbf{G. Rajchel}, F. Kia\l{}ka, and A. Dragan, \emph{Effect of relativistic acceleration on continuous variable quantum teleportation and dense coding}, Phys. Rev. D, vol. 95, 105005, 2017,
    \item[7.] M. Demianowicz, \textbf{G. Rajchel-Mieldzio{\'c}}, and R. Augusiak, \emph{Simple sufficient condition for subspace entanglement},  arXiv: 2107.07530 [quant-ph], 2021. Provisionally accepted at New J. Phys.
\end{enumerate}
\clearpage
\begin{center}
\section*{Acknowledgements}    
\end{center}
\label{sec:acknowledgement}

Throughout my PhD studies, many people helped me to finish my studies and finalize this thesis.
First of all, I would like to express my deep gratitude to my supervisor, Karol \.Zyczkowski.
Due to his constant mentorship, through my stipend at CFT and then PhD studies, I feel that I have gained a lot of insight into the world of science.
I can safely say that it would be extremely hard to finish the studies without his advice, both from the scientific and administrative perspectives.
At all times he was understanding and willing to help.  

Furthermore, I feel indebted to all the administrative team at CFT that helped me with numerous struggles to understand the policies needed to survive in the scientific community.
The other PhD students in room 304A made it a place that is worth coming -- the meeting of the quantum world and ``astrology'' provided a great environment.
Therefore, I would like to acknowledge the support of astrophysicists and cosmologists: Ishika Palit, Julius Sorbenta, Michele Grasso, and Suhani Gupta.

The thesis increased its quality due to numerous people that devoted their time to provide me with insights regarding corrections: Albert Rico, Arul Lakshimanarayan, Jerzy Paczos, my cousin Jim Borneman, Kamil Korzekwa, my sister Katarzyna, Marcin Wierzbi{\'n}ski, Tomasz Linowski, Vinayak Jagadish, and Wojciech Bruzda.
I express my gratitude to Konrad Szyma{\'n}ski for creating a 3D printout of the circulant set.

Having finished exactly 20 years of studies since the beginning of primary school, I feel grateful to my parents who always believed in me and provided every help they could.
It would be hard to express in words what their aid signifies to me. 

Finally, the person who has been my strongest supporter is my wife Paulina.
Her understanding and constant motivation meant everything to me during the research.
There are so many aspects in which she helped me that it would be impossible to mention them all.
Let me just mention the most visible support seen in this thesis, namely her artistic photographs that visualize scientific concepts.



%

\pagebreak
\clearpage
\thispagestyle{empty}%
\tableofcontents


\pagenumbering{gobble}
\pagebreak
\clearpage
\thispagestyle{empty}%
\pagebreak
\pagenumbering{arabic}
\chapter{Introduction}
\label{}
\vspace{-1cm}
\rule[0.5ex]{1.0\columnwidth}{1pt} \\[0.2\baselineskip]
\definecolor{airforceblue}{rgb}{0.36, 0.54, 0.66}
\definecolor{bleudefrance}{rgb}{0.19, 0.55, 0.91}


The basis for our understanding of microscopic physics was laid over one hundred years ago by researchers forming the old quantum theory.
These results, including Planck's law, the Sommerfeld rule, the Schr\"odinger equation, the matrix formulation of quantum mechanics, and many more, provided the community of physicists with a sufficient number of instruments to conduct calculations and explain previously inexplicable paradoxes.
Nonetheless, the new science brought along deep philosophical questions as well.
The solution to the Einstein-Podolsky-Rosen paradox has shown that entanglement is a crucial feature of quantum mechanics.
Then, the first efforts to use these properties of the microscopic world were proposed by theoreticians, while emerging technologies enabled experimentalists to verify them.

The present work is focused on broadening the understanding of theoretical principles governing the behavior of systems on the microscopic scale.
Therefore, we do not need to follow the historical discoveries that led to the present comprehension of nature.
Rather, we shall base our research on a solid mathematical background.
The reader is referred to textbooks that explain research in the broader historical context~\cite{Nielsen_Chuang,Bengtsson_Zyczkowski_geometry}.
If not specified otherwise, these two books provide references for all the principles of quantum information included in this chapter.

\section{Basics of quantum information}\label{sec:basics_QI}
Let us start with the basic definitions used throughout the domain of quantum information.
The most fundamental of them cover states of the quantum system, that can be described using a Hilbert space.

\begin{definition}[Hilbert space]
    A complete complex vector space $\mathcal{H}$ equipped with the inner product is called a Hilbert space.
\end{definition}

Then, any state about which we have exact information is called pure.

\begin{definition}[Pure state]
    A state $\ket{\psi}$ that can be described by a non-zero vector belonging to a Hilbert space, $\ket{\psi} \in \mathcal{H}$, is called pure.
\end{definition}

In this thesis we shall restrict to a finite, $N$-dimensional Hilbert spaces; hence, it is convenient to choose a special basis $\{\ket{i}\}_{i=1}^N$ that we will call the computational basis.
The choice of the basis usually depends on the underlying evolution provided by the Schr\"odinger equation.
However, since we are not limiting our considerations to any particular setup we shall not focus on it.

The dimension of the underlying Hilbert space, i.e.\ the number of states in a basis, determines the properties of the state.
If the Hilbert space is of the dimension $N$, we call it a qu$N$it system, e.g.\ a Hilbert space of dimension 2 describes qubits, dimension -- 6 quhexes, etc.
We shall denote the dimension of the Hilbert space by the superscript $\mathcal{H}^N$.
States describe physical entities that are compatible with the requirements of the theory of probability.
Therefore, we impose that each of them should be normalized using the inner product of the Hilbert space,
\begin{equation}\label{eq:normalization_pure_states}
    \braket{\psi|\psi} = 1.
\end{equation}

Alternative description of a pure state $\ket{\psi}$ is given by its density matrix $\ket{\psi}\bra{\psi}$.
Nonetheless, about some states we have only partial knowledge; thus, they cannot be depicted as pure states but as density matrices.
Then, to describe them it is beneficial to use statistical mixtures of matrices, in contrast to pure states which admit also the vector representation.

\begin{definition}[Mixed state]
    A state $\rho$ that cannot be written as a single projection matrix but only as a convex combination of them
    \begin{equation}
        \rho = \sum_i \alpha_i \ket{\psi_i}\bra{\psi_i}
    \end{equation}
    is called a mixed state.
\end{definition}

The normalization condition~(\ref{eq:normalization_pure_states}) in the language of density matrices transforms into the constraint on the coefficients $\alpha_i$, provided the states $\ket{\psi_i}$ are normalized and form a basis.
This constraint can be written using the trace, i.e.\ the sum of the diagonal elements of the density matrix, 
\begin{equation}
    \mathrm{Tr}\,\rho = \sum_i \alpha_i = 1.
\end{equation}

A composite system, consisting of subsystems $A$ and $B$ described by Hilbert spaces $\mathcal{H}^A$ and $\mathcal{H}^B$, is characterized by the tensor product $\mathcal{H}^A\otimes \mathcal{H}^B$.
If it is possible to avoid confusion, we shall omit the tensor symbol, denoting $\ket{\psi_1}\ket{\psi_2} \coloneqq \ket{\psi_1}\otimes\ket{\psi_2}$.
Sometimes, while considering the states from the computational basis, we shall also omit one of the kets $\ket{ij} \coloneqq \ket{i,j} =  \ket{i}\ket{j}$.
The last important definition in this section concerns quantum correlations between subsystems.

\begin{definition}[Separable state]
    A pure state $\ket{\psi}$ in a composite system $\mathcal{H} = \mathcal{H}^A\otimes \mathcal{H}^B$ is called separable if it can be written in a product form, $\ket{\psi} = \ket{\psi_A}\ket{\psi_B}$, where $\ket{\psi_A}\in \mathcal{H}^A$ and $\ket{\psi_B}\in \mathcal{H}^B$.
    Otherwise, it is called entangled.
\end{definition}

The degree of entanglement of a given state is quantified by measures of entanglement that shall be a topic of Section~\ref{sec:measures_of_entanglement}.
For now, it suffices to say that in the bipartite case there exist maximally entangled states, which are called the generalized Bell states, $\ket{\phi^+} = \frac{1}{\sqrt{N}}\sum_{i=1}^N\ket{i\,i}$.
In the subsequent section, to represent transformations on the states we recall the basic sets of matrices.
For a discussion of the last important postulate of quantum mechanics -- the measurement, we refer the reader to~\cite{Nielsen_Chuang}.

\section{Useful sets of matrices}\label{sec:sets_of_matrices}
Throughout this thesis, we shall consider only square matrices $M$, with complex conjugate denoted by $M^\dagger$ and transpose by $M^T$.
Let us recall the most basic matrix from the perspective of matrix analysis.

\begin{definition}[Identity matrix]
    An identity matrix\, $\mathbb{I}_N$ of dimension $N$ consists of ones at the diagonal, with every other element equal to zero.
    For brevity, if the dimension of the matrix is known from the context, we shall denote it by $\mathbb{I}$.
\end{definition}

Provided as an example to the reader, the identity matrix of size 3 reads,
\begin{equation}
    \mathbb{I}_3 = \begin{pmatrix}
    1 & 0 & 0 \\
    0 & 1 & 0 \\
    0 & 0 & 1
    \end{pmatrix}.
\end{equation}

The identity matrix plays a special role in quantum information since, after normalization, it describes a density matrix of a maximally mixed state, i.e.\ a state about which we do not possess any information and a unitary dynamics which does not alter the state.
Let us move on to sets of matrices important from the perspective of the evolution of quantum states.
Any transformation on states should preserve their normalization.
Therefore, a set of unitary matrices is especially convenient to quantum information scientists.

\begin{definition}[Unitary matrix]
    A complex matrix\, $U$ is unitary if all its rows and columns are orthogonal and normalized, $UU^\dagger = U^\dagger U =\mathbb{I}$.
    The set of all unitary matrices of dimension $N$ shall be denoted by $U(N)$.
\end{definition}

By a unitary transformation $U$ acting on a pure state $\ket{\psi}$ we shall mean $\ket{\psi} \mapsto U\ket{\psi}$.
Alternatively, using a density matrix $\rho$ the same transformation reads $\rho \mapsto U\rho U^\dagger$.
An especially important connection between unitary matrices of size $N$ and pure states in the system of dimension $N^2$ was named the Choi-Jamio\l{}kowski isomorphism
\begin{equation}
    U \leftrightarrow \ket{U} =  \big(U\otimes\, \mathbb{I}\big)\, \ket{\phi^+},
\end{equation}
where $\ket{\phi^+}$ is the generalized Bell state on $\mathcal{H}^N \otimes \mathcal{H}^N$.
This isomorphism shall be used throughout the thesis to study quantum states via properties of their unitary counterparts.
Due to its usefulness, as a special subset of unitary matrices, we distinguish the set of orthogonal ones.

\begin{definition}[Orthogonal matrix]
    A real matrix $\,O$ is orthogonal if it is unitary, $OO^T = O^T O =\mathbb{I}$.
    The set of all orthogonal matrices of dimension $N$ shall be denoted by $O(N)$.
\end{definition}

An interesting property of unitary and orthogonal matrices of a given dimension is that they form a Lie group.
From the perspective of this thesis, this mathematical characteristic is important only to provide derivatives on the manifold of matrices.
For more details, we refer the reader to a comprehensive textbook on Lie groups~\cite{Hall_2015}.
The underlying Lie algebra of unitary matrices is given by skew-Hermitian matrices.

\begin{definition}[Hermitian matrix]
    A complex matrix\, $H$ is Hermitian if it is equal to its complex conjugate, $H=H^\dagger$.
    Alternatively, it is called skew-Hermitian if $H = -H^\dagger$.
\end{definition}

The connection between Hermitian and skew-Hermitian matrices is straightforward.
It suffices to multiply one of them by the imaginary unit $i$ to obtain a matrix from the other set.
Therefore, any unitary matrix $U$ can be obtained by the exponentiation of a Hermitian matrix $H$ as $U = e^{iH}$.
Further importance of the set of Hermitian matrices for quantum information stems from their connection to density matrices.
A density matrix $\rho$, diagonal in one basis $\{\ket{\psi_i}\}$, will generally have off-diagonal terms in any other basis, e.g.\ $\ket{i}\bra{j}$ in the computational one. 
Consequently, it does not suffice that the trace of the density matrix is normalized, it is also necessary that the matrix is Hermitian and positive semi-definite.
Having finished defining infinite sets of matrices, let us move on to their finite subsets.

\begin{definition}[Permutation matrix]
    A binary matrix $P$ is a permutation matrix if in its every row and column all elements, apart from one of them, are equal to zero.
    Thus, in every row and every column, there is exactly one entry that reads 1.
\end{definition}

In dimension $N$, the number of different permutation matrices is $N!$, and some of the well-established matrices, such as the identity, belong to this class.
The swap matrix $S$, existing only in the square dimensions $N^2$, is an especially important permutation matrix, widely used in quantum information.
Its name stems from the fact that, while acting on bipartite systems, it exchanges the states of the subsystems,
\begin{equation}
    S\ket{\psi_1}\ket{\psi_2} = \ket{\psi_2}\ket{\psi_1}.
\end{equation}
The swap operator for a system composed of two qubits is represented by the following matrix of size 4,
\begin{equation}
    S_4 = \begin{pmatrix}
    1 & 0 & 0 & 0 \\
    0 & 0 & 1 & 0 \\
    0 & 1 & 0 & 0 \\
    0 & 0 & 0 & 1
    \end{pmatrix}.
\end{equation}

We are going to use also another notion, important from the perspective of quantum information and computer science.

\begin{definition}[Hadamard matrix]
    A complex matrix $H$ of size $N$ of unimodular entries, $|H_{ij}| = 1$, is called Hadamard if it is unitary up to rescaling, $HH^\dagger = N\,\mathbb{I}$.
\end{definition}

The above definition is not the usual one that admits only real numbers from the set $\{-1,1\}$ as the elements of the matrix; however, the complex case is more suitable to our considerations.
For every dimension $N$, it is possible to define a special Hadamard matrix namely, the Fourier matrix $F_N$, which element $(j,k)$ reads
\begin{equation}
    (F_N)_{jk} = e^{2\pi i jk/N} = \omega^{jk},
\end{equation}
where $\omega = e^{2\pi i/N}$ is the root of unity of order $N$.
It is particularly useful in quantum information to perform Quantum Fourier Transform while designing efficient algorithms.
As a final remark concerning the sets of matrices introduced in this section, we observe that all of them, up to rescaling, form subsets of the set of unitary matrices.
For a more comprehensive treatment of the above sets of matrices, we refer the reader to an invaluable book on matrix analysis by Horn and Johnson~\cite{Horn_Johnson_matrix_analysis}.


    

\section{Operations on matrices}\label{sec:operations_on_matrices}
Most of the operations on matrices that we will use, such as trace, transposition, and complex conjugation, are standard and can be found in any introductory textbook on matrix analysis.
However, some of the notions are less-known; therefore, we shall recall their definitions.

\begin{definition}[Partial trace]
    Any matrix $M$ of dimension $N_A N_B$ admits a decomposition into a sum of the tensor products,
    \begin{equation}\label{eq:tensor_decomposition}
        M = \sum_i A_i\otimes B_i,
    \end{equation}
    where $A_i$ and $B_i$ are matrices of sizes $N_A$ and $N_B$, respectively.
    The partial traces of a matrix $M$ are defined as
    \begin{equation}
        \mathrm{Tr}_A\, M = \sum_i \mathrm{Tr}(A_i)\, B_i
    \end{equation}
    and
    \begin{equation}
        \mathrm{Tr}_B\, M = \sum_i \mathrm{Tr}(B_i)\, A_i.
    \end{equation}
\end{definition}

As an example, consider partial trace $\mathrm{Tr}_B$ of a matrix $M$ of size 4 performed over its first subsystem $B$, which acts as the trace inside the $2\times 2$ blocks,
\begin{equation}
    M = \left(
    \begin{array}{cc|cc}
    M_{11} & M_{12} & M_{13} & M_{14} \\
    M_{21} & M_{22} & M_{23} & M_{24} \\ \hline
    M_{31} & M_{32} & M_{33} & M_{34} \\
    M_{41} & M_{42} & M_{43} & M_{44} 
    \end{array}\right)
    \,\,\,\,\,
    \xmapsto{\;\;\mathrm{Tr}_B\;\;}
    \,\,\,\,\,\mathrm{Tr}_B M = \left(
    \begin{array}{cc}
       M_{11}+M_{22}  & M_{13}+M_{24} \\
         M_{31} + M_{42}& M_{33}+M_{44}
    \end{array}\right).
\end{equation}


The partial trace of a density matrix, performed over a subsystem $A/B$, corresponds to our knowledge about the other subsystem $B/A$.
Thus, the reduced density matrix describes the state of a subsystem while discarding all the information available about the second one.
A closely related concept to the partial trace, also with applications in quantum information, is the partial transposition.

\begin{definition}[Partial transposition]
    The partial transposes of a matrix\, $M$ of size $N_AN_B$ are defined through its tensor decomposition, given by Eq.~(\ref{eq:tensor_decomposition}),
    \begin{equation}
        M^{\Gamma_A} = \sum_i A_i^T\otimes B_i
    \end{equation}
    and
    \begin{equation}
        M^{\Gamma_B} = \sum_i A_i \otimes B_i^T.
    \end{equation}
\end{definition} 

The partial transposition applied to a matrix $M$ of size 4, acting on the second subsystem $B$, is a transposition inside each of the $2\times 2$ blocks separately,
\begin{equation}
    M = \left(
    \begin{array}{cc|cc}
    M_{11} & M_{12} & M_{13} & M_{14} \\
    M_{21} & M_{22} & M_{23} & M_{24} \\ \hline
    M_{31} & M_{32} & M_{33} & M_{34} \\
    M_{41} & M_{42} & M_{43} & M_{44} 
    \end{array}\right)
    \,\,\,\,\,
    \xmapsto{\;\;\Gamma_B\;\;}
    \,\,\,\,\, M^{\Gamma_B} = \left(
    \begin{array}{cc|cc}
    M_{11} & M_{21} & M_{13} & M_{23} \\
    M_{12} & M_{22} & M_{14} & M_{24} \\ \hline
    M_{31} & M_{41} & M_{33} & M_{43} \\
    M_{32} & M_{42} & M_{34} & M_{44} 
    \end{array}\right).
\end{equation}

Partial transposition is especially useful in verifying entanglement, for example in the Peres-Horodecki separability criterion.
This test uses the negativity of the partial transpose of a density matrix to verify the existence of the entanglement.
If not stated otherwise, partial transposition implicitly acts on the second subsystem, $M^\Gamma \coloneqq M^{\Gamma_B}$.
Finally, the last important transformation on matrices that we use is reshuffling, sometimes also called realignment~\cite{Rather_2020}.

\begin{definition}[Reshuffling]
    A matrix $M$ of size $N^2$ admits a decomposition into the tensor product given by Eq.~(\ref{eq:tensor_decomposition}).
    Then, its elements can be written using a four-index notation, each ranging from 1 to $N$,
    \begin{equation}
        M_{klmn} = \sum_i (A_i)_{kl}\otimes (B_i)_{mn}.
    \end{equation}
    The reshuffling of this matrix changes the order of its entries,
    \begin{equation}
        M^R_{klmn} = M_{kmln}.
    \end{equation}
\end{definition}

An alternative description of reshuffling is that each $N\times N$ block of the initial matrix is treated as a consecutive row of the reshuffled matrix.
Applying reshuffling to a matrix $M$ of size 4, one obtains
\begin{equation}
    M = \left(
    \begin{array}{cc|cc}
    M_{11} & M_{12} & M_{13} & M_{14} \\
    M_{21} & M_{22} & M_{23} & M_{24} \\ \hline
    M_{31} & M_{32} & M_{33} & M_{34} \\
    M_{41} & M_{42} & M_{43} & M_{44} 
    \end{array}\right)
    \,\,\,\,\,
    \xmapsto{\;\;\; R \;\;\;}
    \,\,\,\,\, M^R = \left(
    \begin{array}{cc|cc}
    M_{11} & M_{12} &  M_{21} & M_{22}  \\
    M_{13} & M_{14} & M_{23} & M_{24} \\ \hline
    M_{31} & M_{32} & M_{41} & M_{42} \\
    M_{33} & M_{34} & M_{43} & M_{44} 
    \end{array}\right).
\end{equation}

What is more, using the notation of four indices we are also able to write down the partial transposes
\begin{equation}
    M^{\Gamma_A}_{klmn} = M_{lkmn}\,\,\,\,\, \mathrm{and} \,\,\,\,\, M^{\Gamma_B}_{klmn} = M_{klnm},
\end{equation}
as well as the partial traces
\begin{equation}
    \mathrm{Tr}_A\, M_{klmn} = \sum_{i} M_{iimn}\,\,\,\,\, \mathrm{and} \,\,\,\,\, \mathrm{Tr}_{B}\, M_{klmn} = \sum_{i} M_{klii}.
\end{equation}

The last notion of this section, of particular importance for the research presented in Chapter~\ref{chapter_6}, is a multiunitary matrix.
\begin{definition}[Multiunitary matrix]
    A unitary matrix\, $U$ of size $N^2$ is multiunitary if its reshuffle and partial transpose are also unitary, $U^R\in U(N^2)$ and $U^\Gamma \in U(N^2)$.
\end{definition}
In the notation of four indices, if a matrix $U$ is multiunitary then the tensor $U_{klmn}$ is called \emph{perfect}.
Perfect tensors have recently attracted a lot of attention in the community of quantum information and also in the field of condensed matter theory and quantum gravity~\cite{Pastawski_2015}.
Employing the above definition we will be able to explain the research on entanglement, which is the subject of the next section.

\section{Measures of entanglement}\label{sec:measures_of_entanglement}
Until now, we have discussed the existence of entanglement, as well as the existence of maximally entangled states in the bipartite case; however, in order to quantify entanglement, we need to explore the functions called entanglement measures.
The research of entanglement is already over 80 years old, dating to the 1930s and the Einstein-Podolsky-Rosen paradox~\cite{EPR_1935}.
Thus, it is natural that the qualitative description was followed by the quantitative one.
The research concluded in a plethora of entanglement measures, as well as in a formalization of the conditions that should be satisfied by such a measure.
Here, we shall follow the review paper of Plenio and Virmani~\cite{Plenio_2007}.

A function $\mathcal{F}$ acting on a bipartite Hilbert space is called an entanglement measure if it satisfies the following conditions:
\begin{itemize}
    \item its value on a separable state is zero,
    \item it does not increase under local operations and classical communication,
    \item there exist maximally entangled states, i.e.\ the function achieves the maximal value.
\end{itemize}

Out of many measures of entanglement, in this thesis, we shall only use one of them, namely the linear entropy of entanglement, under the alternative name of the generalized concurrence while working with different normalization.
The reader interested in other entanglement measures is referred to an invaluable textbook on quantum information by Nielsen and Chuang~\cite{Nielsen_Chuang}.

\begin{definition}[Linear entropy]
    The linear entropy of entanglement\, $E$ of a state $\ket{\psi} \in \mathcal{H}^A\otimes \mathcal{H}^B$ is defined as
\begin{equation}
    E(\ket{\psi}) \coloneqq 1 - \mathrm{Tr}_A\big(\mathrm{Tr}_B(\ket{\psi}\bra{\psi})\big)^2.
\end{equation}
\end{definition}

In the general case, admitting subsystems with different local dimensions, we shall be using linear entropy with an additional factor of 2, as well as we will always refer to the exact division into two subsystems of the Hilbert space, $\mathcal{H} = \mathcal{H}^A\otimes \mathcal{H}^B$.
Then, the measure is given an alternative name of the generalized concurrence $\tau_{A|B}$, and will be used in Chapter~\ref{chapter_4} while studying the properties of the multipartite unitary gates.
\begin{definition}[Generalized concurrence]
    The generalized concurrence of a bipartite state $\ket{\psi} \in \mathcal{H}^A\otimes \mathcal{H}^B$ reads
    \begin{equation}\label{eq:generalized_concurrence}
    \tau_{A|B}\big( \ket{\psi} \big) \coloneqq 2\bigg( 1 - \mathrm{Tr}_A\big( \mathrm{Tr}_{B} \ket{\psi}\bra{\psi} \big)^2\bigg).
\end{equation}
\end{definition}



Finally, let us move on to the study of the entanglement in the multipartite case.
For a system consisting of more than two subsystems, the notion of the maximally entangled state is not well defined, as it depends on the measure chosen.
Thus, in Chapter~\ref{chapter_6} we shall focus on entanglement through a reduction into bipartite systems.

To explain the research of Chapter~\ref{chapter_4}, it is beneficial to introduce another measure of entanglement.
In the case of the tripartite system, we define a measure of entanglement $\tau_1$ called one-tangle~\cite{Coffman_2000,Bengtsson_Zyczkowski_geometry}.

\begin{definition}[One-tangle]
    The one-tangle of a state in a tripartite system $\mathcal{H} = \mathcal{H}_1 \otimes \mathcal{H}_2 \otimes \mathcal{H}_3$, of local dimensions $\mathrm{dim}\;\mathcal{H}_i = d_i$, is defined as the average of generalized concurrences with respect to all possible splittings,
    \begin{equation}
    \tau_1 \big( \ket{\psi} \big) \coloneqq \frac{1}{3} \bigg( \tau_{12|3}\big( \ket{\psi} \big) + \tau_{13|2}\big( \ket{\psi} \big) + \tau_{23|1}\big( \ket{\psi} \big) \bigg).
\end{equation}
\end{definition}


The notion of one-tangle can be naturally extended to an $N$-partite setting using the generalized concurrence,
\begin{equation}
    \tau_1 (U) \coloneqq \frac{1}{2^{N-1}-1}\sum_{A|B} \tau_{A|B}(U),
\end{equation}
where the summation is understood over all possible partitions of $N$ subsystems into $A \cup B = \{1,...,N\}$ such that $A \cap B = \emptyset$.
The normalization factor $\frac{1}{2^{N-1}-1}$ is the inverse of the number of possible different splittings $A|B$.

\section{Quantum gates}\label{sec:gates_ent_power}
The entanglement is a useful notion in various quantum information protocols, ranging from dense coding and teleportation to quantum cryptography.
Therefore, it is of particular importance to investigate how much entanglement a given gate produces while acting on a separable state.
Instead of studying the statistical distribution of entanglement creation which might prove unfeasible, one can focus on a single aspect -- the average entanglement.

\begin{definition}[Entangling power]\label{def:entangling_power}
    The entangling power of a unitary gate\, $U$ acting on a Hilbert space of dimension $N^2$ is the mean entanglement created by gate $U$, averaged over the set of separable states $\ket{\psi_A}\ket{\psi_B}$ from $\mathcal{H}^N\otimes \mathcal{H}^N$ with respect to the uniform measure on two unit spheres of Hilbert spaces, given by
    \begin{equation}\label{eq:def_e_p}
        e_p(U) = \Bigg(\frac{N+1}{N-1}\Bigg)\,\,\big\langle E\big(U \ket{\psi_A}\ket{\psi_B}\big) \big\rangle_{\ket{\psi_A}\ket{\psi_B}},
    \end{equation}
    where the entanglement measure of choice $E$ is the linear entropy of entanglement, $E(\ket{\psi}) = 1 - \mathrm{Tr}_A\big([\mathrm{Tr}_B(\ket{\psi}\bra{\psi})]^2\big)$.
\end{definition}

It is convenient to normalize the entangling power $e_p$, such that it admits values from the unit interval $[0,1]$.
For brevity of the notation, any time we shall refer to the matrix of size $N^2$, which is an equivalent form of the gate, we will also denote it as $U$.
Definition~\ref{def:entangling_power} was first proposed by Zanardi et al.~\cite{Zanardi_2000}, while a more operational version was introduced by Clarisse et al.~\cite{Clarisse_2005}.
This expression for entangling power uses the matrix form of the unitary gate
\begin{equation}\label{eq:e_p_by_linear_entropy}
    e_p(U) = \frac{1}{E(S)} \big( E(U) + E(US) - E(S) \big),
\end{equation}
with $S$ denoting the swap matrix of the same dimension (see Section~\ref{sec:sets_of_matrices}) and the linear entropy for the matrix being defined by its 
Choi-Jamio{\l}kowski isomorphism: $U \mapsto \ket{U} = U\otimes \mathbb{I} \ket{\phi^+}$. 

To further simplify the calculations of the entangling power, it is possible to evaluate the average over the Haar measure utilizing reshuffling and partial transposition, as noted by Zanardi~\cite{Zanardi_2001}.
Using this result, we can rewrite the entangling power of a unitary matrix employing the operations introduced in Section~\ref{sec:operations_on_matrices}:
\begin{equation}\label{eq:e_p_using_singular_entropy}
    e_p(U) = \frac{N^2}{N^2-1}\bigg(\frac{N^2+1}{N^2}-\frac{\mathrm{Tr}\big(U^R U^{R\dagger}U^R U^{R\dagger}\big)}{N^4} - \frac{\mathrm{Tr}\big(U^\Gamma U^{\Gamma\dagger}U^\Gamma U^{\Gamma\dagger}\big)}{N^4}\bigg).
\end{equation}

Following the definition of the entangling power it is possible to introduce a second, closely related notion, which shall prove important in Section~\ref{sec:sum_of_entropies_AME} to study unitarity of certain matrices.
Thus, let us introduce the \emph{singular entropy}, denoted by $E_S(X)$.
This quantity is connected with the singular values of a general matrix $X$, which does not need to be unitary,
\begin{equation}\label{eq:definition_singular_entropy}
    E_S(X) = \frac{N}{N-1} \bigg(1 - \frac{\text{Tr}\big(XX^\dagger XX^\dagger \big)}{\big[\text{Tr}\big( XX^\dagger \big)\big]^2} \bigg).
\end{equation}
The singular entropy of any matrix $X$ is normalized such that $E_S(X) \in [0,1]$.
For our purposes, the most important property of this notion is that it achieves the maximal value of 1 if and only if the matrix $X$ is unitary.
Furthermore, in a special case in which all three matrices $U$, $U^R$, and $U^\Gamma$ are unitary, the expression for the entangling power can be written using the singular entropy:
\begin{equation}\label{eq:e_p_using_singular_entropy_only_for_AME}
    e_p (U) = E_S(U^R) + E_S(U^\Gamma) - 1,
\end{equation}
which will be a key tool to study multiunitary matrices in Section~\ref{sec:sum_of_entropies_AME}.
The two expressions for $e_p$, by the linear entropy (\ref{eq:e_p_by_linear_entropy}) and the singular one (\ref{eq:e_p_using_singular_entropy_only_for_AME}), do not coincide in general.
However, the fact that they have the same global maximum, namely any multiunitary matrix, provides a basis for two inequivalent approaches to the search for such a matrix.

A similar notion to the singular entropy was introduced while studying the vector of singular values of an operator by Roga et al.~\cite{Roga_2013}.
Note that there is a connection between the singular entropy and the linear entropy, exemplified by the common form 
\begin{equation}
    f(X) = 1 - \mathrm{Tr}\big((\ket{X}\bra{X})^2\big),
\end{equation}
where for the singular entropy we defined the transition from matrices to states via $X \mapsto \ket{X}\bra{X} = XX^\dagger / \mathrm{Tr}(XX^\dagger)$.
Nonetheless, this transformation is not the Choi-Jamio{\l}kowski isomorphism employed by the linear entropy; therefore, in general these definitions do not coincide.



In order to study the set of bipartite quantum gates, it is beneficial to consider a quantity complementary to the entanglement power.

\begin{definition}[Gate typicality]
    The gate typicality\, $g_t$ of a unitary gate\, $U$ of size $N^2$ is defined by
    \begin{equation}
        g_t(U) = \frac{1}{E(S)} \big( E(U) - E(US) + E(S) \big).
    \end{equation}
    The gate typicality takes values from the unit interval $[0,1]$, while its mean value averaged over the Haar measure on $U(N^2)$ reads $1/2$.
\end{definition}

All three notions introduced in this section will prove beneficial for further research, especially the entangling power that plays a significant role in quantum mappings and quantum designs.
We conclude this chapter with two sections providing the reader with the introduction to the topic of this thesis.

\section{Connection between maps and designs}\label{sec:quantum_designs}
The link between the classical and the quantum domain is under investigation since the advent of the quantum theory; nonetheless, several fundamental problems still remain unsolved.
To mitigate these issues, we tackle the transition between both worlds by consideration of classical maps in the quantum regime.
To clarify, quantum mappings from the perspective of this thesis are actions of unitary matrices on quantum systems.

In particular, Chapter~\ref{chapter_3} is devoted to the study of the role of bistochastic matrices in quantum mechanics through their connection to unitary matrices.
Bistochastic matrices emerge from the field of Markov chains~\cite{Gagniuc_Markov_chains} and are widely used for their practical applications to model various processes.
Then, as a slightly separated topic, we explore the entangling power of quantum mappings in the multipartite domain (see Chapter~\ref{chapter_4}).


The second, closely related quantum mechanical concept studied in this thesis, important from the theoretical as well as from the experimental perspective, is a \emph{quantum design}.
Let us start by motivating this particular line of research from the point of view of a theoretician.
Introduced by a seminal PhD thesis, written by Zauner in German in 1999 \cite{Zauner_German} with a subsequent translation into English \cite{Zauner_English}, a quantum design was defined to be a set of mutually orthogonal projection matrices. 
The thesis defined several subclasses of designs, which were later intensively investigated by other researchers, as well as extended to other concepts.
The examples of such notions include: unitary $t$-designs~\cite{Dankert_2009}, symmetric informationally complete positive-operator-valued measures (SIC-POVMs)~\cite{Renes_2004}, and mutually unbiased bases (MUBs)~\cite{Schwinger_1960}.

A quantum design in a broader sense of this thesis is a combinatorial quantum object with symmetry, particularly useful in quantum information.
Obviously, the above definition is a very broad one, but instead of a formal definition, one should understand this as more of a guideline.
This notion has a lot of applications across the wide field of quantum information science.
To study quantum designs we shall use the entangling power, showing a deep connection linking designs to quantum mappings.

In this thesis, we shall restrict to two particular examples of quantum designs, namely quantum Latin squares and quantum Sudoku.
Chapters~\ref{chapter_6} and \ref{chapter_7} will familiarize the reader with both notions as well as will try to explain their significance. 
Therefore, it is beneficial to introduce the common definitions for both of these chapters.

\begin{definition}[Latin square]
    A square array of size $N$, filled with\, $N$ different symbols, such that in every row and every column each element repeats exactly once, is called a Latin square.
\end{definition}

Quantum Latin squares, first introduced in 2016 by Musto and Vicary \cite{Musto_2016}, are generalizations of their classical counterparts.
\begin{definition}[QLS]
    A square array of size $N$, filled with $N^2$ quantum states (potentially non-unique), is called a quantum Latin square (QLS) if every row and column forms an orthogonal basis of the Hilbert space $\mathcal{H}^N$.
\end{definition}
The notion of QLS is useful and interesting as it combines combinatorics with quantum information, allowing one to solve certain experimental problems, see Chapter~\ref{chapter_7}.

\section{Structure of the thesis and the author's contribution}\label{sec:structure_contribution}
To summarize the introductory chapter, the primary topic of this thesis is the domain of classical to quantum transition, facilitated by the usage of quantum designs.

\vspace{0.5cm}
\emph{The thesis is organized as follows.}
\vspace{0.5cm}

Chapter~\ref{chapter_3} studies the transition between macro and microscopic domain in detail and provides new results on the unistochasticity problem.
Then, Chapter~\ref{chapter_4} describes novel results on the entangling power, a notion extensively used in Chapter~\ref{chapter_6} to construct an absolutely maximally entangled state of four quhexes.
This particular quantum design has no classical counterpart.
Finally, this example of a genuinely quantum design is generalized in Chapter~\ref{chapter_7}, facilitated by the usage of a quantum SudoQ.

\vspace{0.5cm}
\emph{The author's contribution to the results of each chapter is described below.}
\vspace{0.5cm}

In Chapter~\ref{chapter_3}, the first of the novel contributions is the algorithm allowing to determine whether a given bistochastic matrix of size 4 is unistochastic (Section~\ref{sec:Uffe_algorithm}).
The general idea has been proposed by the late Uffe Haagerup, while the details of the procedure, as well as its implementation in Mathematica~\cite{algorithm_Rajchel}, are sole contributions of the author.
Furthermore, the proof that the robust Hadamard matrices are sufficient for unistochasticity of the rays and the counter-rays of the Birkhoff polytope of certain dimensions
(Theorem~\ref{thm:robust_Hadamard_rays}) was also developed solely by the author, see Section~\ref{sec:rays_and_counter-rays}.
The notion of strongly complementary matrices was introduced by the author alone, together with the proof that certain triangles embedded inside the Birkhoff polytope of any even dimension are unistochastic (Section~\ref{sec:complementary_matrices}).
The last result of this section developed by the author exclusively is the proof of Theorem~\ref{thm:circulant_size4} in Section~\ref{sec:circulant_matrices}, which states that the bracelet conditions are sufficient to determine whether a given circulant bistochastic matrix of size 4 is unistochastic.


One of the two main contributions of the author in Chapter~\ref{chapter_4} is the determination of the orthogonal Weingarten functions, given by Eq.~(\ref{eq:Weingarten_final}).
Then, using these formulas, together with Tomasz Linowski, the author obtained the most important contribution of this chapter, i.e.\ the derivation of an analytical formula for the average entangling power of tripartite orthogonal gates, presented by the formula~(\ref{eq:average_e_p_tripartite_orthogonal}).

In the search for an absolutely maximally entangled (AME) state of four quhexes, the contributions of the author are as follows.
First, the author individually contributed by introducing the new family of matrices $G$ defined by Eq.~(\ref{eq:G_family}), as shown in Section~\ref{sec:introducing_G}.
Then, the sole contribution of the author is the characterization of the region provided by the $W$ family, as presented in Fig.~\ref{fig:region_W}, as well as the study of its extremality included in Section~\ref{sec:checking_W_family}.
Further, the derivation of the Hessian for the entangling power is a joint work of the author and Arul Lakshminarayan, included in Section~\ref{sec:local_maxima}, together with the subsequent, independent numerical calculations concerning eigenvalues of the Hessian.

Using the Hessian, the author introduced the algorithm concerning the matrices that reach the proximity of the $W$ matrix, with results shown in Table~\ref{tab:iteration_G_to_W}.
Another individual contribution of the author is the derivation of the formula for derivatives for the average singular entropy, as presented in Appendix~\ref{app:hessian_s_e}.
Then, employing these derivatives, the author verified the extremality of matrices $X_i$ and $X_d$ -- provided by Wojciech Bruzda.
Additionally, the author's unique involvement includes the discovery of the block-like structure of the numerically found absolutely maximally entangled state, as presented in Section~\ref{sec:block-like_AME}, with an extension in Appendix~\ref{app:block_like}.
Then, using this structure, the author alone conducted a search for an AME state, presented in Section~\ref{sec:search_block_structure}.

In the final chapter of the thesis, the individual addition of the author consists in proposing the notions of genuinely quantum Latin squares and its cardinality, as per Definitions~\ref{def:cardinality} and \ref{def:genuinely_quantum_QLS}.
Then, the author found the first example of a SudoQ of the maximal cardinality; therefore, possessing the highest degree of \emph{quantumness}, see Eq.~(\ref{eq:SudoQ_cardinality16}).
The other contribution of the author includes the work done with Jerzy Paczos and Marcin Wierzbi\'nski concerning the admissible cardinalities of the SudoQ $4\times 4$ (Theorem~\ref{thm:SudoQ_cardinalities4x4}), where the author's part involved corrections to the proof.
Then, a similar contribution of the author concerns the construction of a SudoQ of the maximal cardinality for any dimension, see Section~\ref{sec:sudoQ_maximal_cardinality_general} and Propositions~\ref{prop:maximal_cardinality1} and \ref{prop:maximal_cardinality2}.


\vspace{0.5cm}
\emph{To strengthen the idea that this thesis forms a monograph, here we list all of the results that have not been published before.}
\vspace{0.5cm}

First of them is the derivation of the average entangling power, included in Section~\ref{sec:average_e_p_tripartite_orthogonal}, which has never been made public and might be helpful for future researchers.
This is because the author feels that the evaluation of the analytical expression for the average entangling power of orthogonal tripartite gates, given by Eq.~(\ref{eq:average_e_p_tripartite_orthogonal}), deserved more treatment to explain mathematical intricacies than what was presented in the joint paper with Linowski and \.Zyczkowski~\cite{Rajchel_entangling_power}.


Furthermore, in Chapter~\ref{chapter_6}
the search for an absolutely maximally entangled state of four quhexes is not included in any publication.
In particular, this refers to the families of matrices $A$ (Section~\ref{sec:A_family}), $G$ (Section~\ref{sec:introducing_G}), and $W$ (Section~\ref{sec:introducing_W}).
Then, the same applies to the derivation of the Hessian for the entangling power (Section~\ref{sec:local_maxima}), as well as its eigenvalues presented in Table~\ref{tab:hessian_e_p_results}.
The steepest ascent algorithm (Section~\ref{sec:achieving_W_from_G}) with its results presented in Table~\ref{tab:iteration_G_to_W} has never been part of any publication.
Likewise, the results concerning the derivatives of the average singular entropy, involving the conjectured extremality of matrices presented in Table~\ref{table:4_categories} were never presented to the community of scientists.
Subsequent developments provided in Section~\ref{sec:region_W}, devoted to the study of the $W$ family, is also of a novel nature.
In addition, the block-like structure (Section~\ref{sec:block-like_AME}) and the ensuing search (Section~\ref{sec:search_block_structure}) were not included in any paper.

Finally, both Appendices have not yet been published.
Appendix~\ref{app:hessian_s_e} is dedicated to the derivation of the derivative for the average singular entropy of a matrix, while Appendix~\ref{app:block_like} covers the full explanation of the block-like structure of multiunitary matrices associated to a 4-party AME state.




\clearpage
\part{Quantum mappings}
\chapter{Unistochastic maps}
\label{chapter_3}
\vspace{-1cm}
\rule[0.5ex]{1.0\columnwidth}{1pt} \\[0.2\baselineskip]

\section{Introduction}\label{sec:introduction_unistochastic}
Unitary matrices $U(N)$ are of particular importance in quantum mechanics since they describe the evolution of the quantum states, as mentioned in Section~\ref{sec:sets_of_matrices},
\begin{equation}
    \ket{\psi(t)} = e^{-iHt/\hbar}\ket{\psi} = U \ket{\psi}.
\end{equation}

One of the valuable sets of matrices from the perspective of a physicist is the set of bistochastic matrices $\mathcal{B}_N$, i.e.\ those matrices with nonnegative entries that have rows and columns summing up to 1.
Their importance emerged during the research of modeling of various stochastic processes since they conserve the classical probability.
In particular, they are useful in the study of Markov chains -- the distribution between two time steps $p$ and $p'$ reads
\begin{equation}
    p' = Bp,
\end{equation}
where $B$ is a bistochastic matrix.
The link between bistochastic matrices and the set of unitary matrices is striking -- it suffices to take the absolute value squared of elements of a unitary matrix to obtain a bistochastic one
\begin{equation}
    |U_{ij}|^2 = B_{ij}.
\end{equation}


Nonetheless, the reverse statement is not true, that is not all bistochastic matrices have their unitary counterpart, i.e. a bistochastic matrix
\begin{equation}
    \frac{1}{2}\begin{pmatrix}
        0 & 1 & 1 \\ 1 & 0 & 1 \\ 1 & 1 & 0
\end{pmatrix}
\end{equation} 
cannot be transformed to a unitary matrix.
The problem of verifying whether a given bistochastic matrix can be converted to a unitary one is called the \emph{unistochasticity} problem.

The present chapter of this thesis shall be devoted to answering this problem in particular setups, as well as to applying these mathematical notions to certain examples in the field of quantum information, e.g.\ bases possessing a constant degree of entanglement.
A summary of some parts of this chapter, as well as an extension of the others, in which the author's involvement was less substantial, can be found in a joint paper~\cite{Rajchel_robust,Rajchel_algebraic_structures}.
If not specified differently, the author's contribution to the work covered by this chapter was significant.

\begin{figure}[H]
    \includegraphics[scale=0.25]{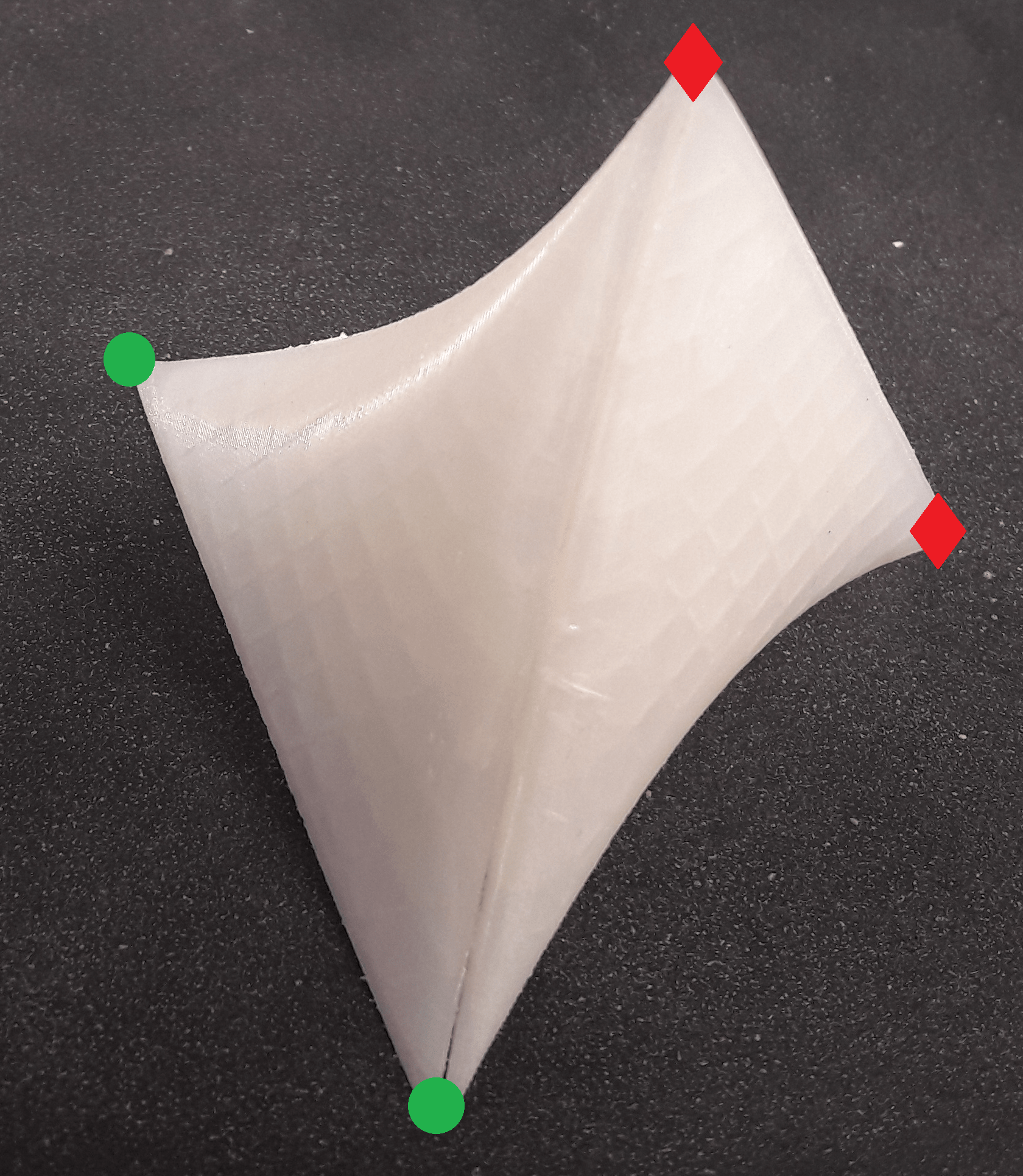}
    \caption{Unistochastic circulant matrices of dimension 4 are proved to form a non-convex subset of a tetrahedron containing all bistochastic circulant matrices.
    Two pairs of corners, green circles and red diamonds, given by permutation matrices, are distinguished.
    The author is grateful to Konrad Szyma{\'n}ski for creating this 3D printout.}
    \label{fig:model3D_unistochastic}
\end{figure}

\section{Bistochastic matrices}
We shall start by defining the core notions used throughout this chapter.

\begin{definition}[Bistochastic matrices]
    A real matrix $B$ composed of nonnegative entries is said to be bistochastic if the sum of its elements in each row and column equal 1, $\sum_i B_{ij} = 1$ and $\sum_j B_{ij} = 1$.
    We shall denote the set of bistochastic matrices of dimension $N$ by $\mathcal{B}_N$.
\end{definition}

The full set of bistochastic matrices of dimension $N$ is called the Birkhoff polytope, honoring the Birkhoff–von Neumann theorem.
This result states that any bistochastic matrix can be written as a convex combination of permutation matrices of an appropriate size.
Therefore, the Birkhoff polytope $\mathcal{B}_N$ is a convex hull of $N!$ permutation matrices, which are placed in the corners of the polytope~\cite{Birkhoff_1946}.
As a side-note for the reader wanting to compare the above definition to the other works, we remark that bistochastic matrices are also sometimes called doubly stochastic.
Out of bistochastic matrices of dimension $N$, we shall distinguish a matrix that is in some sense as central as possible.
To this end, we recall the matrix that was famously conjectured by van der Waerden to minimize permanent~\cite{van_der_Waerden_1926}.

\begin{definition}
    A bistochastic matrix of dimension $N$ composed only of entries $1/N$ is called the van der Waerden matrix and denoted by $W_N$.
\end{definition}

The centrality of this matrix shall be described later on while studying geometrical properties of some sets inside bistochastic matrices.
For now, let us observe that the uniform mixture of all permutation matrices of dimension $N$ yields the van der Waerden matrix $W_N$.
In Section~\ref{sec:introduction_unistochastic} we noted that any unitary matrix can be converted to a bistochastic one but that the converse does not hold.
Therefore, we shall distinguish those for which it is true by the name of unistochastic matrices.

\begin{definition}[Unistochastic matrices]
    A bistochastic matrix\, $B$ of size $N$ is called unistochastic if there exists a unitary matrix\, $U\in U(N)$ such that $B_{ij} = |U_{ij}|^2$.
    The set of unistochastic matrices of dimension $N$ will be denoted by $\mathcal{U}_N$.
\end{definition}

Note that the van der Waerden matrix $W_N$ is unistochastic for all dimensions $N$ since the corresponding unitary matrix can be taken to be the appropriate Fourier matrix $F_N$ (see Section~\ref{sec:sets_of_matrices}).
Unistochastic matrices are widely used in the field of quantum dynamics, where the problem of quantization of a given bistochastic map is solved by a proper unitary matrix~\cite{Pakonski_2001,Pakonski_2003}.
Another application of unistochastic maps arises from the field of elementary particles. 
To describe mixing between different quark families, physicists use the unitary Cabibbo-Kobayashi-Maskawa matrix~\cite{Bigi_2009}, which is probed experimentally by its bistochastic counterpart, composed of probabilities of conversion~\cite{Dita_2006}.
Furthermore, one can use unistochastic matrices of size $N$ to find discrete quantum walks on a graph containing $N$ vertices~\cite{Aharonov_1993}.
Having based our research on solid physical applications, we shall move on to the study of the useful notion of bracelet matrices.

\section{Bracelet matrices}\label{sec:bracelet_matrices}
Even though research on unistochasticity has a great potential for solving several problems across physics mentioned in the previous section, the unistochasticity problem is far from being well-understood.
The only properly described case are bistochastic matrices of order 3, for which simple necessary and sufficient conditions for unistochasticity are known~\cite{Au-Yeung_1979,Jarlskog_1988}.
The conditions are based upon the notion of \emph{bracelet} matrices.

\begin{definition}[Bracelet matrices]
    A bistochastic matrix $B$ of dimension $N$ is called bracelet if it satisfies the following conditions
    	\begin{subequations}
		\begin{align}
		\label{eq:bracelet_row}
		2\max_j \sqrt{B_{lj}B_{kj}}&\leq \sum_{j=1}^N \sqrt{B_{lj}B_{kj}},\\
		\label{eq:bracelet_column}
		2\max_j \sqrt{B_{jk}B_{jl}}&\leq \sum_{j=1}^N \sqrt{B_{jk}B_{jl}},
		\end{align}
	\end{subequations}	
	for any $k,l\in \{1,...,N\}$.
	Those conditions are called, respectively, row and column bracelet conditions.
	The set of bracelet matrices of dimension $N$ will be denoted as $\mathcal{L}_N$.
\end{definition}

Although apparently complicated, row and column bracelet conditions have a simple geometrical explanation that motivates their name.
In order to visualize it, suppose that $B$ is a unistochastic matrix of dimension 3
\begin{equation}
    B = \begin{pmatrix}
a & b & c \\
d & f & g \\
h  & j  &  k   
\end{pmatrix}.
\end{equation}
Then, the following matrix $U$ is a unitary matrix for certain values of real phases $\{\alpha,...,\kappa\}$
\begin{equation}
    \begin{pmatrix}
{\color{black}\sqrt{a}e^{i\alpha}} & {\color{black}\sqrt{b}e^{i\beta}} & {\color{black}\sqrt{c}e^{i\gamma}} \\
{\color{black}\sqrt{d}e^{i\delta}} & {\color{black}\sqrt{f}e^{i\epsilon}} & {\color{black}\sqrt{g}e^{i\zeta}} \\
\sqrt{h}e^{i\eta}  & \sqrt{j}e^{i\theta}  &  \sqrt{k}e^{i\kappa}   
\end{pmatrix}.
\end{equation}

The orthogonality condition for the first two rows reads
\begin{equation}
    \sqrt{ad\mathstrut}e^{i(\alpha-\delta)} + \sqrt{bf}e^{i(\beta-\epsilon)} + \sqrt{cg\mathstrut}e^{i(\gamma-\zeta)} = 0.
\end{equation}

The above equation can only be satisfied if three lines of length $\sqrt{ad\mathstrut}$, $\sqrt{bf\mathstrut}$, and $\sqrt{cg\mathstrut}$ form a triangle since one can interpret the addition of complex numbers by geometrical means on the complex plane (see Fig.~\ref{fig:triangle_bracelet}).
If a triangle cannot be formed with respect to any pair of rows/columns, then the necessary conditions for unistochasticity are not met.
Therefore, such a matrix cannot be unistochastic.
In dimension $4$, the corresponding triangle condition transforms to the quadrilateral condition, while in dimension $N$ to the $N$-polygon condition.

\begin{figure}[H]
    \includegraphics[scale=1]{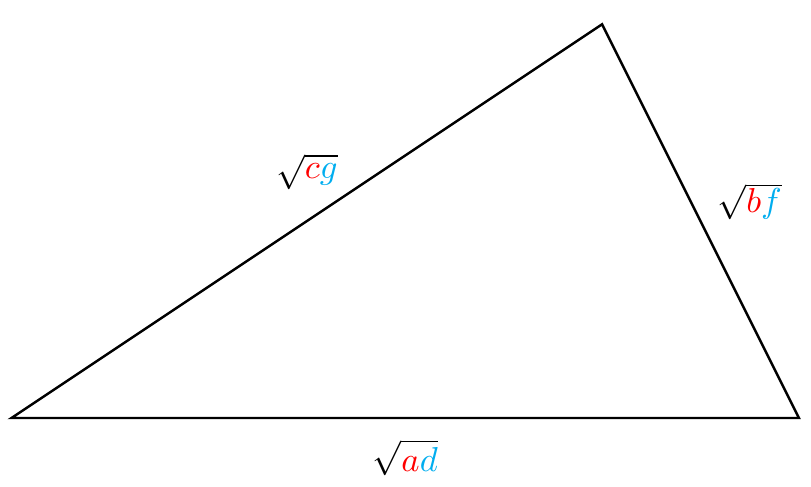}
    \caption{A $3 \times 3$ matrix 
    $\begin{pmatrix}
{\color{red}a} & {\color{red}b} & {\color{red}c} \\
{\color{cyan}d} & {\color{cyan}f} & {\color{cyan}g} \\
h  & j  &  k   
\end{pmatrix}$
    can only be unistochastic if a triangle, called the \emph{unitarity triangle}, can be formed by the entry-wise product of two rows containing square roots of the entries.
    The same orthogonality condition applies to all pairs of rows and columns.
    In the case of dimension 3, these conditions are sufficient for unistochasticity.
    Dimension 3 is special also because all unitarity triangles are equivalent; therefore, it suffices to verify the triangle inequality only for one pair of rows or columns~\cite{Jarlskog_1988}. 
}
    \label{fig:triangle_bracelet}
\end{figure}

As a result of these remarks, we conclude that the bracelet set is a superset of unistochastic matrices in any dimension $N$, which can be written in our notation as $\mathcal{U}_{N} \subset\mathcal{L}_{N}$.
What is more, in 1979 Au-Yeung and Poon found that in the case of dimension 3, these two sets are equal, $\mathcal{U}_{N}=\mathcal{L}_{N}$, making the bracelet condition not only necessary but also sufficient~\cite{Au-Yeung_1979}.
Then, the 2005 paper of Bengtsson et al.~\cite{Bengtsson_2005} described the geometrical structure of unistochastic matrices of order 3.
They form a 4-dimensional star-shaped set with the central point corresponding to the van der Waerden matrix $W_3$.

As a side note for the reader interested in the intricacies of the Birkhoff polytope, we recall the other geometrical results of this paper that might prove useful in understanding the geometrical considerations applied throughout the present chapter.
Specifically, the set of unistochastic matrices of size $3$ has a non-zero volume~\cite{Dunkl_2009} and contains a ball of unistochastic matrices, centered at $W_3$.
The ball is of radius $\sqrt{2}/3$ in the Hilbert-Schmidt metric, which defines the distance between matrices $A$ and $B$ as $\sqrt{\text{Tr}\big((A-B)(A-B)^\dagger\big)}$.

However, in the same paper the authors emphasized that the case of bistochastic matrices of dimension 4 is much more complicated.
The bracelet conditions are not sufficient for unistochasticity but also in every neighborhood of the van der Waerden matrix $W_4$ there exists a non-unistochastic matrix.
This implies the lack of any unistochastic ball around the center $W_4$ of the Birkhoff polytope $\mathcal{B}_4$.
Therefore, the state of the art for the unistochasticity problem for dimension $4$ does not admit any simple algorithm which determines whether a given bistochastic matrix has its unitary counterpart.
In the subsequent section, we shall introduce a numerical algorithm solving this task.

\section{Unistochasticity algorithm for dimension \texorpdfstring{$N=4$}{Lg}}\label{sec:Uffe_algorithm}
The core idea behind the present algorithm was proposed by the late Uffe Haagerup during informal collaboration.
Then, the extension and implementation of the algorithm was provided by the author of this thesis.
The rationale behind the algorithm will be the subject of this section.

Starting from a bistochastic matrix $B\in\mathcal{B}_4$,
\begin{equation}
    B=
    \begin{pmatrix}
        B_{11} & B_{12} & B_{13} & B_{14} \\
        B_{21} & B_{22} & B_{23} & B_{24} \\
        B_{31} & B_{32} & B_{33} & B_{34} \\
        B_{41} & B_{42} & B_{43} & B_{44} \\
    \end{pmatrix},
\end{equation}
we wish to determine whether there exists a unitary matrix $U$ such that
\begin{equation}\label{eq:unitary_matrix_uffe_algorithm}
    U=
    \begin{pmatrix}
        \sqrt{B_{11}} & \sqrt{B_{12}} & \sqrt{B_{13}} & \sqrt{B_{14}} \\
        \sqrt{B_{21}} & \sqrt{B_{22}}e^{i\phi} & \sqrt{B_{23}}e^{i\alpha_1} & \sqrt{B_{24}}e^{i\alpha_2} \\
        \sqrt{B_{31}} & \sqrt{B_{32}}e^{i\beta_1} & \sqrt{B_{33}}e^{i\delta_1} & \sqrt{B_{34}}e^{i\delta_2} \\
        \sqrt{B_{41}} & \sqrt{B_{42}}e^{i\beta_2} & \sqrt{B_{43}}e^{i\delta_3} & \sqrt{B_{44}}e^{i\delta_4} \\
    \end{pmatrix} = 
    \left( \begin{array}{c|c} 
    A & X \\ 
    \hline Y & D 
    \end{array} \right).
\end{equation}

Without loss of generality, we may assume that $U$ is dephased in such a way that elements in the first row and column are real numbers.
Furthermore, using the division into 4 blocks of size $2\times 2$ we impose the condition on $U$ that the first block is bounded in the Hilbert-Schmidt norm
\begin{equation}\label{eq:bound_on_submatrix_A}
    ||A|| = \sqrt{\text{Tr}(AA^\dagger)} < 1,
\end{equation}
what can be always achieved by a proper permutation acting on the matrix $B$.
In other words, either $B$ already satisfies this condition or we permute it.
Note that permutations do not change the unistochasticity of a matrix, i.e.\ if $B$ is unistochastic then all matrices $PBP'$, obtained by permutations $P$ and $P'$, will also be unistochastic.

To see that every bistochastic matrix can be permuted in such a way, observe that condition~(\ref{eq:bound_on_submatrix_A}) is equivalent to $B_{11}+B_{12}+B_{21}+B_{22} < 1$.
On the other hand, this can always be done if the initial matrix differs from the van der Waerden matrix $W_4$, in the opposite case a Hadamard or the Fourier matrix of dimension 4 shows unistochasticity.
Unitarity of the matrix $U$ implies that
\begin{equation}\label{eq:unitary_condition_blocks_uffe_algorithm}
    AA^\dagger + XX^\dagger = \mathbb{I}_{2} \quad \text{ and } \quad AY^\dagger + XD^\dagger = 0.
\end{equation}

Then, we can use the bound~(\ref{eq:bound_on_submatrix_A}) on the block $A$, from which we deduce that the eigenvalues of $AA^\dagger$ are smaller than 1.
Moreover, together with Eq.~(\ref{eq:unitary_condition_blocks_uffe_algorithm}), this shows that the eigenvalues of $XX^\dagger$ are positive, which in turn implies that the matrix $X$ is invertible.
Therefore, $(X^\dagger)^{-1}$ can be utilized to derive the diagonal block $D$ from the other blocks using Eq.~(\ref{eq:unitary_condition_blocks_uffe_algorithm}),
\begin{equation}\label{eq:D_block_uffe}
    D = -YA^\dagger (X^\dagger)^{-1}.
\end{equation}

Furthermore, we shall find the phases $\phi$, $\alpha_{1}$, $\alpha_{2}$, $\beta_{1}$, and $\beta_{2}$ from Eq.~(\ref{eq:unitary_matrix_uffe_algorithm}) by making use of orthogonality relations between the first two rows and columns.
In order to simplify the notation, we introduce auxiliary variables that can be interpreted as lengths of the segments appearing in the bracelet conditions,
\begin{equation}
    l_1 = \sqrt{B_{11}B_{21}}, \quad l_{2} = \sqrt{B_{12}B_{22}}, \quad l_{3} = \sqrt{B_{13}B_{23}}, \quad l_{4} = \sqrt{B_{14}B_{24}}.
\end{equation}

These allow us to rewrite the orthogonality condition imposed on the pair of the first two rows
\begin{equation}\label{eq:length_of_lines_bracelet_algorithm}
    l_1 + l_2 e^{i\phi} + l_{3} e^{i\alpha_1} + l_{4} e^{i\alpha_{2}} = 0,
\end{equation}
what we shall treat as an equation for unknown phases $\phi$, $\alpha_1$, and $\alpha_2$.
The row bracelet condition Eq.~(\ref{eq:bracelet_row}) requires that the longest segment is not longer than the sum of all other segments, which is equivalent to the condition 
\begin{equation}
    2\max_i l_i \leq \sum_{j=1}^4 l_{j}. 
\end{equation}

Failure to satisfy this condition renders the matrix non-unistochastic.
However, should this requirement be met, Eq.~(\ref{eq:length_of_lines_bracelet_algorithm}) provide two possible solutions for phases $\alpha_1$ and $\alpha_2$ when treated as functions of the phase $\phi$ 
\begin{equation}
    \alpha_1 = \alpha_1 (\phi) \quad \text{and} \quad \alpha_2 = \alpha_2 (\phi)
\end{equation}
for the non-convex case and
\begin{equation}
    \widetilde{\alpha}_1 = \widetilde{\alpha}_1 (\phi) \quad \text{and} \quad \widetilde{\alpha}_2 = \widetilde{\alpha}_2 (\phi)
\end{equation}
for the convex case, see Fig.~\ref{fig:non-convex_and_convex_Uffe_algorithm}.

\begin{figure}[H]
\newcommand{\n}{0.63}
\newcommand{\xC}{9.5}
\newcommand{\sizeText}{\normalsize}
	\begin{tikzpicture}
         \draw [line width=0.25mm] (0,0) node[anchor=north]{} 
              -- (8*\n,0) node[anchor=north]{}
              -- (9*\n,9*\n) node[anchor=south]{}
              -- (5*\n,3*\n) node[anchor=south]{}
              -- cycle;
        \node at (4*\n,-0.4*\n) {$l_1$};
        \node at (8.9*\n,4.2*\n) {$l_2$};
        \node at (6.9*\n,6.57*\n) {$l_3$};
        \node at (2.3*\n,1.9*\n) {$l_4$};
  \path []
  (10*\n,0) coordinate (a) node[right] {}
    -- (8*\n,0) coordinate (b) node[left] {}
    -- (9*\n,9*\n) coordinate (c) node[above right] {}
    pic[draw=orange, ->, angle eccentricity=1.2, angle radius=0.63cm,line width=0.2mm]
    {angle=a--b--c};
    \draw [dashed] (8*\n,0) -- (9.5*\n,0);
    \node at (9*\n,1*\n) {$\phi$};
    
    \path []
  (10*\n,9*\n) coordinate (a) node[right] {}
    -- (9*\n,9*\n) coordinate (b) node[left] {}
    -- (5*\n,3*\n) coordinate (c) node[above right] {}
    pic[draw=orange, ->, angle eccentricity=1.2, angle radius=0.63cm,line width=0.2mm]
    {angle=a--b--c};
    \draw [dashed] (9*\n,9*\n) -- (10.5*\n,9*\n);
    \node at (8.1*\n,10.1*\n) {$\alpha_1$};
    
    \path []
  (10*\n,3*\n) coordinate (a) node[right] {}
    -- (5*\n,3*\n) coordinate (b) node[left] {}
    -- (0,0) coordinate (c) node[above right] {}
    pic[draw=orange, ->, angle eccentricity=1.2, angle radius=0.63cm,line width=0.2mm]
    {angle=a--b--c};
    \draw [dashed] (5*\n,3*\n) -- (6.5*\n,3*\n);
    \node at (3.9*\n,4.1*\n) {$\alpha_2$};
    
         \draw [line width=0.25mm] (0+\xC,0) node[anchor=north]{} 
              -- (8*\n+\xC,0) node[anchor=north]{}
              -- (9*\n+\xC,9*\n) node[anchor=south]{}
              -- (3*\n+\xC,5*\n) node[anchor=south]{}
              -- cycle;
        \node at (4*\n+\xC,-0.4*\n) {$l_1$};
        \node at (8.9*\n+\xC,4.2*\n) {$l_2$};
        \node at (5.8*\n+\xC,7.4*\n) {$l_3$};
        \node at (1.*\n+\xC,2.5*\n) {$l_4$};
  \path []
  (10*\n+\xC,0) coordinate (a) node[right] {}
    -- (8*\n+\xC,0) coordinate (b) node[left] {}
    -- (9*\n+\xC,9*\n) coordinate (c) node[above right] {}
    pic[draw=orange, ->, angle eccentricity=1.2, angle radius=0.63cm,line width=0.2mm]
    {angle=a--b--c};
    \draw [dashed] (8*\n+\xC,0) -- (9.5*\n+\xC,0);
    \node at (9*\n+\xC,1*\n) {$\phi$};
    
    \path []
  (10*\n+\xC,9*\n) coordinate (a) node[right] {}
    -- (9*\n+\xC,9*\n) coordinate (b) node[left] {}
    -- (3*\n+\xC,5*\n) coordinate (c) node[above right] {}
    pic[draw=orange, ->, angle eccentricity=1.2, angle radius=0.63cm,line width=0.2mm]
    {angle=a--b--c};
    \draw [dashed] (9*\n+\xC,9*\n) -- (10.5*\n+\xC,9*\n);
    \node at (7.9*\n+\xC,10.1*\n) {$\widetilde{\alpha}_1$};
    
    \path []
  (10*\n+\xC,5*\n) coordinate (a) node[right] {}
    -- (3*\n+\xC,5*\n) coordinate (b) node[left] {}
    -- (0+\xC,0) coordinate (c) node[above right] {}
    pic[draw=orange, ->, angle eccentricity=1.2, angle radius=0.63cm,line width=0.2mm]
    {angle=a--b--c};
    \draw [dashed] (3*\n+\xC,5*\n) -- (4.5*\n+\xC,5*\n);
    \node at (1.9*\n+\xC,6.1*\n) {$\widetilde{\alpha}_2$};

    \node at (0*\n,10.7*\n) {(a)};
    \node at (15*\n,10.7*\n) {(b)};

	\end{tikzpicture}
    \caption{Unitarity quadrilateral for $N=4$.
    Two possible choices of phases ($\alpha_1$, $\alpha_2$) and ($\widetilde{\alpha}_1$, $\widetilde{\alpha}_2$) are available for some values of $\phi$ -- one of them corresponds to a non-convex (a) and another to a convex (b) polygon in the complex plane.
    Without loss of generality, we may assume that phase $\phi \in [0,\pi ]$.
    }
    \label{fig:non-convex_and_convex_Uffe_algorithm}
\end{figure}
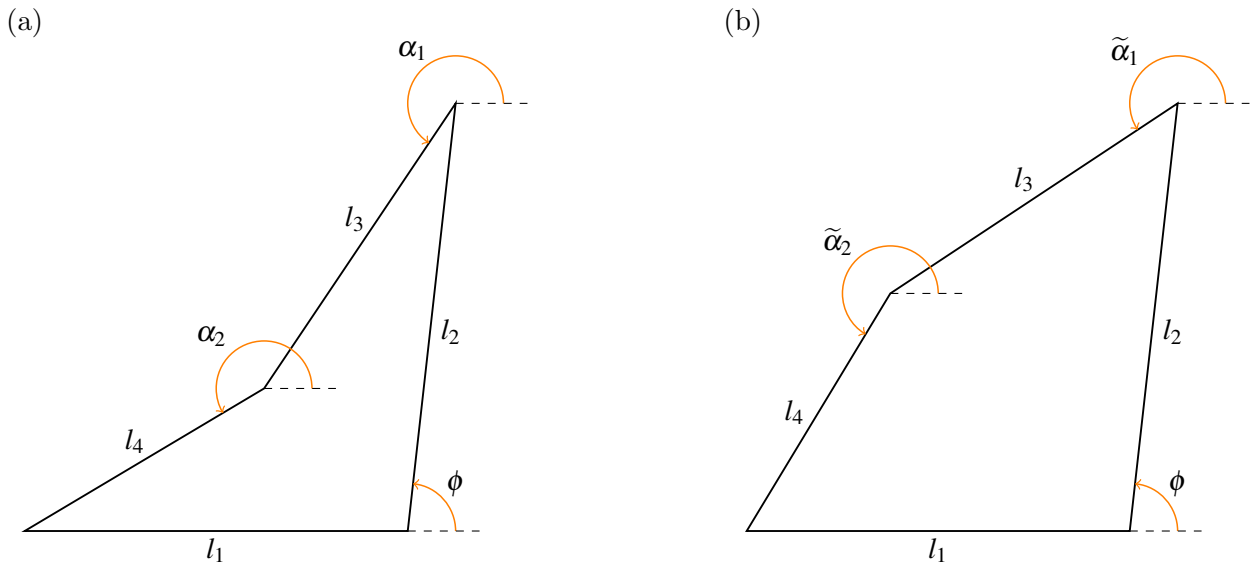

On the other hand, not all of the angles $\phi$ admit the existence of a polygon.
Typically, only for some subset of phases $\phi \in [\phi_{\text{min}},\phi_{\text{max}}]\subset [0,2\pi]$ a quadrilateral can be formed.
In the case of the extremal angles, see Fig.~\ref{fig:phi_max_min_Uffe_algorithm}, the quadrilateral will be degenerate.

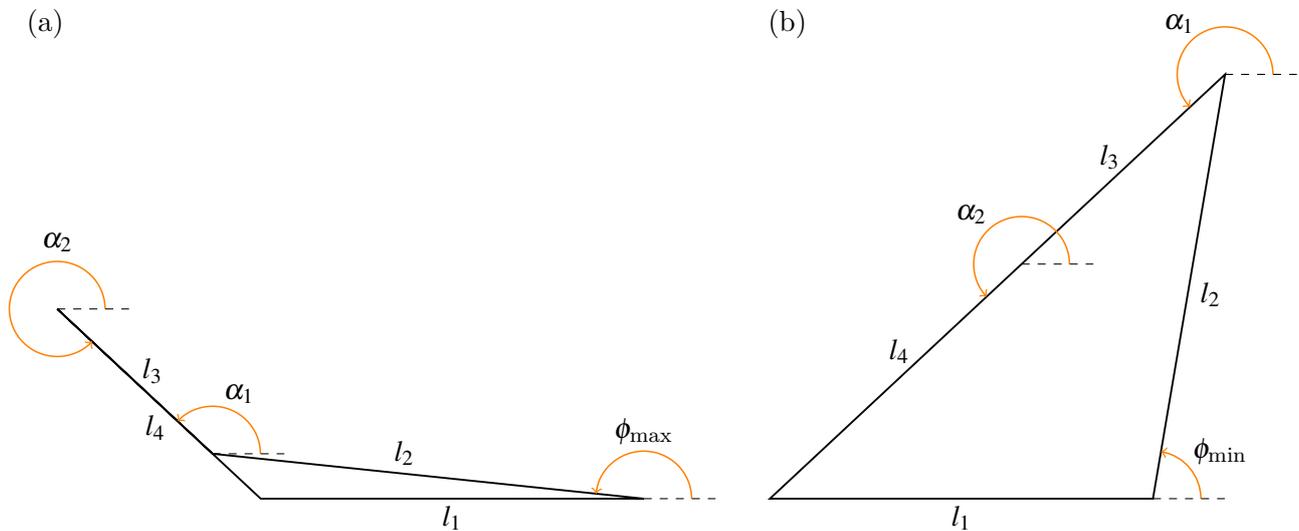
\begin{figure}[H]
    \newcommand{\x}{6.7}
\newcommand{\n}{0.63}
\newcommand{\sizeText}{\normalsize}
	\begin{tikzpicture}
	
         \draw [line width=0.25mm] (0,0) node[anchor=north]{} 
              -- (8*\n,0) node[anchor=north]{}
              -- (-1.00594*\n, 0.94495*\n) node[anchor=south]{}
              -- (-4.24993*\n, 3.99225*\n) node[anchor=south]{}
              -- cycle;
        \node at (4*\n,-0.4*\n) {\sizeText $l_1$};
        \node at (3*\n,0.97*\n) {\sizeText $l_2$};
        \node at (-2.15*\n-0.11*\n,2.8*\n-0.11*\n) {\sizeText $l_3$};
        \node at (-2.3*\n+0.1*\n,1.4*\n+0.1*\n) {\sizeText $l_4$};
  \path []
  (10*\n,0) coordinate (a) node[right] {}
    -- (8*\n,0) coordinate (b) node[left] {}
    -- (-1.00594*\n, 0.94495*\n) coordinate (c) node[above right] {}
    pic[ draw=orange, ->, angle eccentricity=1.2, angle radius=0.63cm,line width=0.2mm]
    {angle=a--b--c};
    \draw [dashed] (8*\n,0) -- (9.5*\n,0);
    \node at (8*\n,1.42*\n) {\sizeText $\phi_{\text{max}}$};
    
    \path []
  (1.00594*\n, 0.94495*\n) coordinate (a) node[right] {}
    -- (-1.00594*\n, 0.94495*\n) coordinate (b) node[left] {}
    -- (-4.24993*\n, 3.99225*\n) coordinate (c) node[above right] {}
    pic[draw=orange, ->, angle eccentricity=1.2, angle radius=0.63cm,line width=0.2mm]
    {angle=a--b--c};
    \draw [dashed] (-1.00594*\n, 0.94495*\n) -- (-1.00594*\n+1.5*\n, 0.94495*\n);
    \node at (-0.45*\n, 2.2*\n) {\sizeText $\alpha_1$};
    
    \path []
  (4.24993*\n, 3.99225*\n) coordinate (a) node[right] {}
    -- (-4.24993*\n, 3.99225*\n) coordinate (b) node[left] {}
    -- (0,0) coordinate (c) node[above right] {}
    pic[draw=orange, ->, angle eccentricity=1.2, angle radius=0.63cm,line width=0.2mm]
    {angle=a--b--c};
    \draw [dashed] (-4.24993*\n, 3.99225*\n) -- (-4.24993*\n+1.5*\n, 3.99225*\n);
    \node at (-4.2499*\n,5.4*\n) {\sizeText $\alpha_2$};
    
    
    \draw [line width=0.25mm] (0+\x,0) node[anchor=north]{} 
              -- (8*\n+\x,0) node[anchor=north]{}
              -- (9.50591*\n+\x, 8.92929*\n) node[anchor=south]{}
              -- (5.25595*\n+\x, 4.93712*\n) node[anchor=south]{}
              -- cycle;
        \node at (4*\n+\x,-0.4*\n) {\sizeText $l_1$};
        \node at (9.2*\n+\x,4.3*\n) {\sizeText $l_2$};
        \node at (7.1*\n+\x,7.2*\n) {\sizeText $l_3$};
        \node at (2.7*\n+\x,3.1*\n) {\sizeText $l_4$};
  \path []
  (10*\n+\x,0) coordinate (a) node[right] {}
    -- (8*\n+\x,0) coordinate (b) node[left] {}
    -- (9.50591*\n+\x, 8.92929*\n) coordinate (c) node[above right] {}
    pic[ draw=orange, ->, angle eccentricity=1.2, angle radius=0.63cm,line width=0.2mm]
    {angle=a--b--c};
    \draw [dashed] (8*\n+\x,0) -- (9.5*\n+\x,0);
    \node at (9.4*\n+\x,1.*\n) {\sizeText $\phi_{\text{min}}$};
    
    \path []
 (10.50591*\n+\x, 8.92929*\n) coordinate (a) node[right] {}
    -- (9.50591*\n+\x, 8.92929*\n) coordinate (b) node[left] {}
    -- (4.24931*\n+\x, 3.99155*\n) coordinate (c) node[above right] {}
    pic[ draw=orange, ->, angle eccentricity=1.2, angle radius=0.63cm,line width=0.2mm]
    {angle=a--b--c};
    \draw [dashed] (9.50591*\n+\x, 8.92929*\n) -- (9.50591*\n+\x+1.5*\n, 8.92929*\n);
    \node at (8.55*\n+\x,10.*\n) {\sizeText $\alpha_1$};
    
    \path []
  (6.25595*\n+\x, 4.93712*\n) coordinate (a) node[right] {}
    -- (5.25595*\n+\x, 4.93712*\n) coordinate (b) node[left] {}
    -- (0+\x,0) coordinate (c) node[above right] {}
    pic[draw=orange, ->, angle eccentricity=1.2, angle radius=0.63cm,line width=0.2mm]
    {angle=a--b--c};
    \draw [dashed] (5.25595*\n+\x, 4.93712*\n) -- (5.25595*\n+\x+1.5*\n, 4.93712*\n);
    \node at (5.25595*\n+\x-1.05*\n, 1.05*\n+4.93712*\n) {\sizeText $\alpha_2$};
    
    \node at (-4.5*\n,10*\n) {\sizeText (a)};
    \node at (11*\n,10*\n) {\sizeText (b)};

	\end{tikzpicture}
    \caption{Two degenerated quadrilaterals corresponding to the extremal angles $\phi_{\text{max}}$ and $\phi_{\text{min}}$.
    In the case of $\phi_{\text{max}}$, the polygon is deformed to a triangle with an additional line (a), while the angle $\phi_{\text{min}}$ allows formation of a triangle (b).
    The extremal instances of $\phi$ are the only cases, for which both of the solutions for angles $\alpha_i$ are equal, $\alpha_1 = \widetilde{\alpha}_1$ and $\alpha_2 = \widetilde{\alpha}_2$.
    The relationship between angles read $\alpha_1 = \alpha_2 +\pi$ (a) and $\alpha_1 = \alpha_2$ (b).
    }
    \label{fig:phi_max_min_Uffe_algorithm}
\end{figure}

Likewise, similar reasoning applied to the first two columns yields another subset of angles $\phi$ for which solutions ($\beta_1$, $\beta_2$) and ($\widetilde{\beta}_1$, $\widetilde{\beta}_2$) exist.
Summing up these considerations, we arrive at the intersection of two possible subsets for the phase $\phi$, yielding $\phi \in [\Phi_{\text{min}},\Phi_{\text{max}}]$.

Finally, it is possible to determine also the last block of the matrix $D$, which is given by Eq.~(\ref{eq:D_block_uffe}).
However, for a randomly chosen phase $\phi \in [\Phi_{\text{min}},\Phi_{\text{max}}]$, the block $D$ will not correspond to the expected block of the initial matrix $U$ -- the amplitudes of the complex numbers will not be the same.
To uniquely determine the unitary $4\times 4$ matrix $U$ it is enough to specify its three $2\times 2$ blocks; however, similar reasoning does not apply to bistochastic matrices, for which the resulting solution still has one degree of freedom,
\begin{equation}
    B(x)=
    \begin{pmatrix}
        B_{11} & B_{12} & B_{13} & B_{14} \\
        B_{21} & B_{22} & B_{23} & B_{24} \\
        B_{31} & B_{32} & B_{33}-x & B_{34}+x \\
        B_{41} & B_{42} & B_{43}+x & B_{44}-x \\
    \end{pmatrix}.
\end{equation}

Thus, it is possible that the unitary matrix obtained by the application of Eq.~(\ref{eq:D_block_uffe}) yields a member of the family $B(x)$ different from the desired matrix $B = B(0)$.
To solve this problem, we must search for the whole set of admissible solutions: all phases $\phi \in [\Phi_{\text{min}},\Phi_{\text{max}}]$ and resulting four different choices of angles $(\alpha_1,\alpha_2,\beta_1,\beta_2)$, $(\widetilde{\alpha_1},\widetilde{\alpha_2},\beta_1,\beta_2)$, $(\alpha_1,\alpha_2,\widetilde{\beta_1},\widetilde{\beta_2})$, and $(\widetilde{\alpha_1},\widetilde{\alpha_2},\widetilde{\beta_1},\widetilde{\beta_2})$.
Consequently, we conclude that if the search for a unitary matrix corresponding to the $B = B(0)$ matrix is successful, matrix $B$ is unistochastic.
Conversely, failure to find a corresponding unitary matrix renders $B$ non-unistochastic.

The exact implementation of the algorithm determining the unistochasticity problem for dimension 4 in the Mathematica language is available online~\cite{algorithm_Rajchel}.
Using this program we are able to verify the statement from the work by Bengtsson et al.~\cite{Bengtsson_2005} concerning non-unistochastic matrices in any neighborhood of the van der Waerden matrix $W_4$.
The graphical representation is shown in Fig.~\ref{fig:not_unistochastic_line}.

\begin{figure}[H]
	\begin{center}
		\begin{tikzpicture}
			\node (myfirstpic) at (0,0) {\includegraphics[width=0.8\columnwidth]{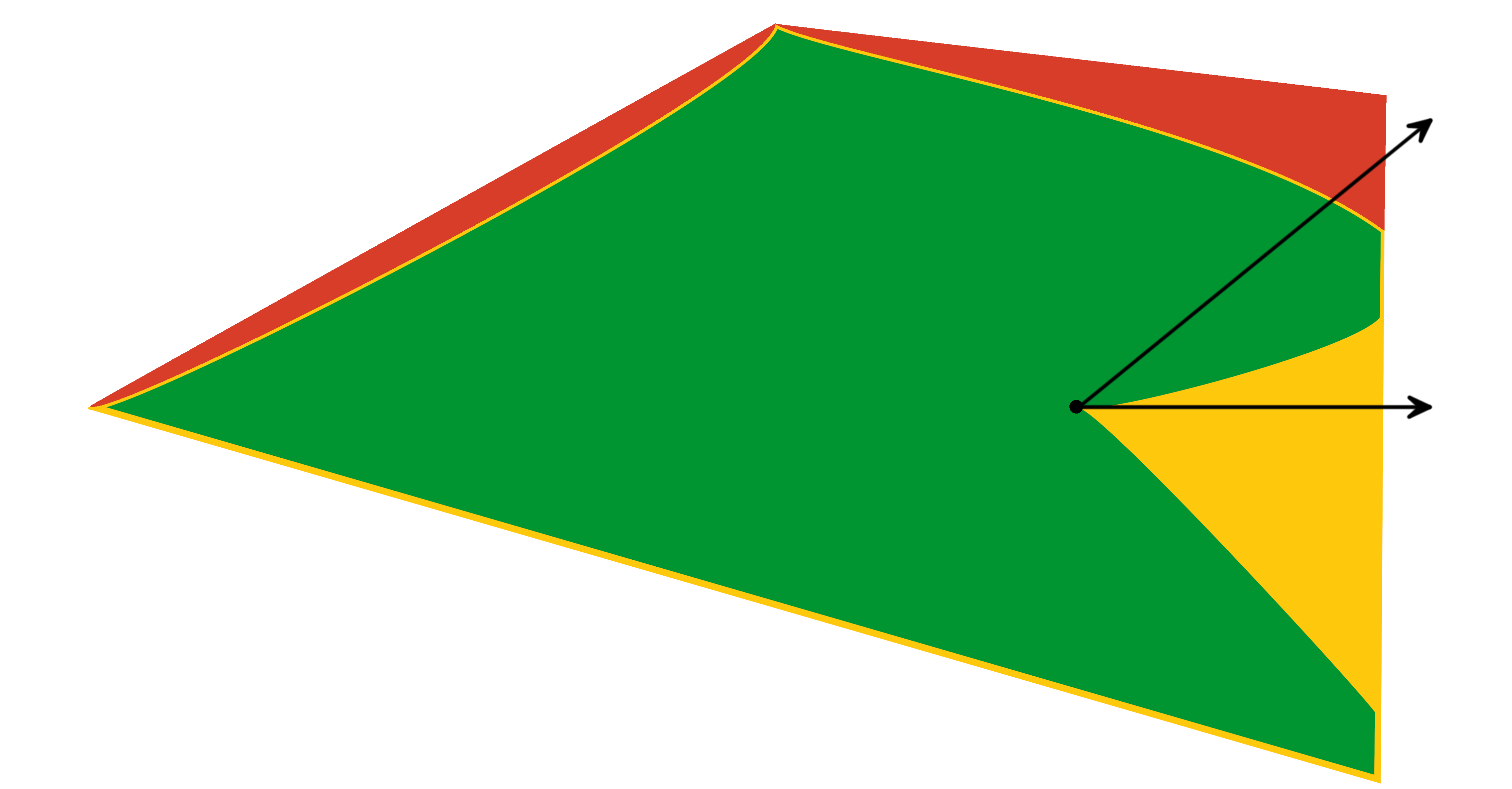}};		
			\node [label={$v_1=\begin{pmatrix}9 &-3 &-3 &-3\\-3&1&1&1\\-3&1&1&1\\-3&1&1&1\end{pmatrix}\!,$}] at (-2.8,-6){};
			\node [label={$v_2=\begin{pmatrix}7 &-1 &-1 &-5\\-1&-1&-1&3\\-1&-1&-1&3\\-5&3&3&-1\end{pmatrix}$}] at (2.,-6){};	
			
			\node at (0.14\columnwidth,-0.0\columnwidth) {\color{black}$W_4$};
			\node at (0.29\columnwidth,-0.07\columnwidth) {\color{black}$\mathcal{L}_4$};
			\node at (-0.05\columnwidth,0.01\columnwidth) {\color{black}$\mathcal{U}_4$};
			\node at (0.39\columnwidth,-0.01\columnwidth) {\color{black}$v_1$};	
			\node at (0.39\columnwidth,0.15\columnwidth) {\color{black}$v_2$};		
			\node at (0.295\columnwidth,0.135\columnwidth) {\color{black}$\mathcal{B}_4$};					
		\end{tikzpicture}
	\end{center}
	\caption{A cross-section through the Birkhoff polytope $\mathcal{B}_4$, given by the van der Waerden matrix $W_4$ and two directions $v_1$ and $v_2$.
	The green set of unistochastic matrices $\mathcal{U}_4$ does not fully surround the matrix $W_4$ since any deviation from $W_4$ in the direction $v_1$ yields a non-unistochastic matrix belonging to the yellow set of bracelet matrices $\mathcal{L}_4$.
	All the other matrices, marked in red, do not satisfy the bracelet conditions.
	Reproduced verbatim from the joint paper~\cite{Rajchel_algebraic_structures}.
	}\label{fig:not_unistochastic_line}
\end{figure}

As a final remark regarding unistochastic matrices of size 4, we have established that they do not form a monoid, which is a mathematical structure similar to a semigroup, albeit slightly stronger.
To clarify, the set $\mathcal{X}$ is a monoid if for any $x,y\in \mathcal{X}$ the binary operation does not lead out of the set, $xy\in \mathcal{X}$.
Furthermore, distinguishing it from a semigroup, monoid contains the identity element $e$ such that for any $x\in \mathcal{X}$ we have $ex=xe=x$.
The set of unistochastic matrices trivially does not form a group, which is an even stronger notion satisfying the condition of the existence of an inverse element $x^{-1}$ for every $x\in \mathcal{X}$ such that $xx^{-1} = x^{-1}x = e$.
To see why is it so, consider the van der Waerden matrix $W_N$.
Its determinant reads zero; therefore, it admits no inverse matrix.

To conclude, the set of unistochastic matrices of dimension 4 does not form a monoid, which we prove using the bistochastic matrix
\begin{equation}
J=\frac{1}{100}
\begin{pmatrix}
24 & 16 & 35 & 25 \\
38 & 21 & 12 &29 \\
23 & 24 & 14 & 39 \\
15 & 39 & 39 & 7
\end{pmatrix}.
\end{equation}
The algorithm presented in this section, implemented in Mathematica~\cite{algorithm_Rajchel}, verifies that $J$ is unistochastic while $J^2$ is not.
Thus, the condition for forming a monoid is not satisfied.

\section{Robust Hadamard matrices}\label{sec:robust_Hadamard_matrices}
In this section, we shall focus on the application of the subset of Hadamard matrices, defined in Section~\ref{sec:sets_of_matrices}, to the study of the unistochasticity problem.
To this end, we start by defining the central notion, introduced in the joint paper~\cite{Rajchel_robust}.

\begin{definition}[Robust Hadamard matrix]
    A Hadamard matrix $H$ is called robust if for any chosen indices $i\neq j$ the matrix formed by $\begin{pmatrix}
    H_{ii} & H_{ij} \\
    H_{ji} & H_{jj}
    \end{pmatrix}$
    is also Hadamard.
\end{definition}

The name of robust Hadamard matrices stems from the observation that any projection of such a matrix into a 2-dimensional subset yields a Hadamard matrix.
Robust Hadamard matrices are also connected to another subset of Hadamard matrices.

\begin{definition}[Skew Hadamard matrix]
    A real Hadamard matrix $H$ is called skew if $H+H^T = 2\mathbb{I}$.
\end{definition}

The simplest example of a skew Hadamard matrix is provided by the following matrix of order 2
\begin{equation}
    \begin{pmatrix}
            1 & 1\\
            -1 & 1
    \end{pmatrix}.
\end{equation}

The connection between these two sets of matrices, valid for any dimension $N$, is given by the subsequent remark.

\begin{remark}
    Every skew Hadamard matrix is robust Hadamard.
\end{remark}
\begin{proof}
    Using the definition, we observe that every diagonal element of a skew Hadamard matrix equals 1.
    Furthermore, any pair of off-diagonal elements, $H_{ij}$ and $H_{ji}$, consists of entries of the opposite sign.
    Therefore, we conclude that $\begin{pmatrix}
    H_{ii} & H_{ij} \\
    H_{ji} & H_{jj}
    \end{pmatrix}$ is a Hadamard matrix for any $i\neq j$.
\end{proof}

An additional connection to another set of matrices can be found using symmetric conference matrices.

\begin{definition}[Symmetric conference matrix]
    A symmetric matrix $C$ of size $N$, with elements equal to 0 on the diagonal and entries $\pm 1$ outside of it, is called symmetric conference if it satisfies the orthogonality condition $CC^T = (N-1)\mathbb{I}$.
\end{definition}

Based on the above definition, we provide the connection to the robust Hadamard matrices, also valid for any dimension $N$.

\begin{remark}
    Every matrix of the form $H = C + i\,\mathbb{I}$, where $C$ is a symmetric conference matrix, is a robust Hadamard matrix.
\end{remark}
\begin{proof}
Verification of this statement relies on the fact that every $2\times 2$ submatrix, corresponding to $i$ and $j$ indices, is Hadamard.
\end{proof}

Having finished the discussion of the connections to the widely known sets of matrices, let us apply the notion of robust Hadamard matrices to rays and counter-rays inside the Birkhoff polytope $\mathcal{B}_N$.

\section{Rays and counter-rays of the Birkhoff polytope}\label{sec:rays_and_counter-rays}

We shall start by establishing the notion of rays and counter-rays, whose names are derived from their geometrical properties.
More specifically, a ray is a line connecting the central matrix $W_N$ with one of the permutation matrices that form vertices of the polytope, whereas counter-ray is its extension inside $\mathcal{B}_N$, starting from $W_N$, in the opposite direction from the permutation matrix.

\begin{definition}[Ray and counter-ray of $\mathcal{B}_N$]
    A subset of the Birkhoff polytope $\mathcal{R}$ is called ray if it is formed by matrices that are convex combinations of a given permutation matrix $P$ and the van der Waerden matrix $W_N$:
    \begin{equation}
        \mathcal{R} = \{\alpha P + (1-\alpha) W_{N} \;| \; \alpha \geq 0\}.
    \end{equation}
    Analogously, a subset of the Birkhoff polytope $\widetilde{\mathcal{R}}$ is called counter-ray if
    \begin{equation}
        \widetilde{\mathcal{R}} = \{\alpha P + (1-\alpha) W_{N} \;| \; \alpha \leq 0\}.
    \end{equation} 
\end{definition}

Both of the above notions are illustrated in a schematic drawing (Fig.~\ref{fig:rays_counter-rays}).

\begin{figure}[H]
    \input{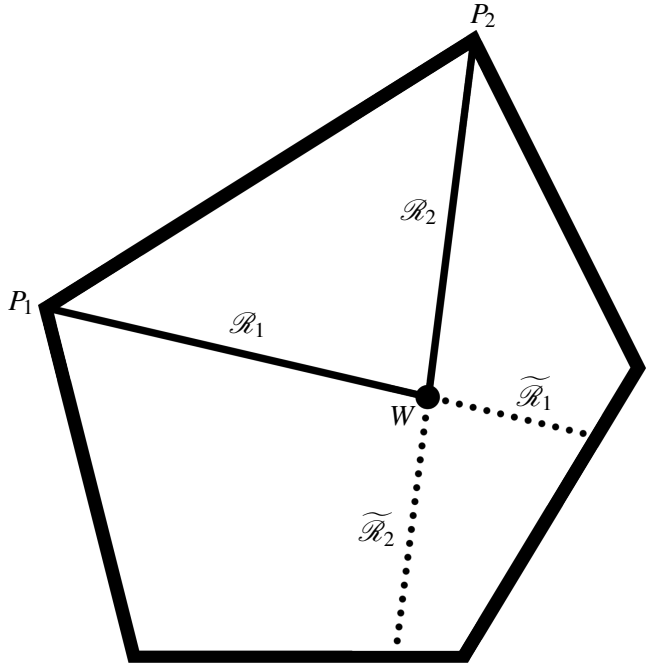}
    \caption{Sketch of a generic cross-section through the Birkhoff polytope $\mathcal{B}_N$ containing the van der Waerden matrix $W_N$ and two permutation matrices, $P_1$ and $P_2$.
    The solid lines inside the Birkhoff polypote denote the rays $\mathcal{R}_1$ and $\mathcal{R}_2$, while the dotted lines depict the counter-rays $\widetilde{\mathcal{R}_1}$ and $\widetilde{\mathcal{R}_2}$.
    }
    \label{fig:rays_counter-rays}
\end{figure}

Following the definitions of rays and counter-rays we will focus on the unistochasticity of these subsets.

\begin{lemma}
    If there exists a robust Hadamard matrix\, $H$ of size $N$ then all of the rays and counter-rays of the Birkhoff polytope $\mathcal{B}_N$ are unistochastic.
\end{lemma}
\begin{proof}
    Let us begin by showing the auxiliary statement concerning the robust Hadamard matrix $H$
    \begin{equation}\label{eq:auxiliary_robust_rays}
        H D^\dagger + D H^\dagger = 2\mathbb{I},
    \end{equation}
    where $D$ is a diagonal matrix formed by the diagonal entries of $H$.
    Note that the diagonal elements of the left-hand side amount to $2H_{ii}\overline{H_{ii}} = 2$, whereas the off-diagonal terms of the left-hand side for $i\neq j$ read $H_{ij}\overline{H_{jj}} + \overline{H_{ji}}H_{ii}$.    
    To demonstrate that these off-diagonal terms are zero, we observe that $HD^\dagger$ is also robust Hadamard.
    Then, considering the $2\times 2$ submatrix $M_2$ of $HD^\dagger$ determined by indices $i$ and $j$, we use its robust property,
    \begin{equation}
    M_{2}^{} M_2^\dagger = 2\mathbb{I} = 
        \begin{pmatrix}
        2 & H_{ij}\overline{H_{jj}} + \overline{H_{ji}}H_{ii} \\
        \overline{H_{ij}}H_{jj} + H_{ji}\overline{H_{ii}} & 2
        \end{pmatrix}.
    \end{equation}
    Since off-diagonal terms of the above matrix are zero and coincide with off-diagonal terms of $H D^\dagger + D H^\dagger$, this proves Eq.~(\ref{eq:auxiliary_robust_rays}).
    
    Now, we shall demonstrate the unistochasticity of any matrix $R$ belonging to the ray or the counter-ray connected with the identity matrix $\mathbb{I}$, provided the existence of an appropriate robust Hadamard matrix $H$.
    To this end we construct a corresponding matrix $U$
    \begin{equation}
        U = \sqrt{\alpha\mathstrut} D + \sqrt{\beta} (H - D),
    \end{equation}
    where $D$ denotes the diagonal of $H$ and the real parameters $\alpha$ and $\beta$ can be determined from the equation $(N-1) \beta +\alpha = 1$.
   By this way we can achieve \emph{any} matrix $R$ from the ray and counter-ray connected to the identity matrix, by a suitable choice of the parameter $\alpha$ or $\beta$.
    Further, we prove that $U$ is unitary by verifying its orthogonality conditions
    \begin{equation}
    \begin{split}
        UU^\dagger &= \big(\sqrt{\alpha\mathstrut} D + \sqrt{\beta} (H - D)\big)\big(\sqrt{\alpha\mathstrut} D^\dagger + \sqrt{\beta} (H^\dagger - D^\dagger) \big)  \\
        &= \beta HH^\dagger - \sqrt{\beta}(\sqrt{\beta}-\sqrt{\alpha\mathstrut})\big( HD^\dagger +DH^\dagger \big) + (\sqrt{\beta} - \sqrt{\alpha\mathstrut})^2 DD^\dagger = \big( (N-1) \beta +\alpha \big) \mathbb{I} = \mathbb{I}.
    \end{split}
    \end{equation}
    
    Finally, we observe that the matrix $R$ corresponds to the studied matrix $U$ via element-wise operation, $R_{ij} = |U_{ij}|^2$.
    This proves that the ray and counter ray connected with the identity matrix are unistochastic, provided a robust Hadamard matrix of size $N$ exists.
    Moreover, a similar result holds for all rays and counter-rays since any of them can be achieved from the ray studied previously by a multiplication by a proper permutation matrix.
    Therefore, the lemma is proved in the general case.
\end{proof}

Additionally, if the robust Hadamard matrix is real so that $U$ becomes orthogonal, the matrix $R$ is not only unistochastic but also orthostochastic.
By this we mean that the underlying unitary matrix is orthogonal.
Using the connection between robust Hadamard matrices and other sets of matrices we are able to demonstrate the following remarks.

\begin{remark}
    For every dimension $N$, for which a skew Hadamard matrix exists, all rays and counter-rays of the Birkhoff polytope $\mathcal{B}_N$ are unistochastic, and, in particular, orthostochastic.
\end{remark}

\begin{remark}
    For every dimension $N$ for which a symmetric conference matrix exists all rays and counter-rays of the Birkhoff polytope $\mathcal{B}_N$ are unistochastic.
\end{remark}

Let us note that the existence of skew Hadamard matrices of orders $N = 4k$ is known for $k< 69$, with the proper construction done by Paley in 1933~\cite{Paley_1933}.
Also, there are infinitely many higher dimensions for which their existence is confirmed~\cite{Koukouvinos_2008}.
Similarly, it is known that for $N=$ 6, 10, 14, 18 there exists a symmetric conference matrix~\cite{van_Lint_1966}.
Nonetheless, an analogous construction is not known for the order $N=22$.
Concluding the research on ray and counter-rays we are able to prove the following theorem.

\begin{theorem}\label{thm:robust_Hadamard_rays}
For any even dimension $N<22$ all rays and counter-rays of the Birkhoff polytope $\mathcal{B}_N$ are
\begin{enumerate}
    \item unistochastic, as well as orthostochastic, (for $N=$ $2$, $4$, $8$, $12$, $16$, $20$) or
    \item unistochastic (for $N=$ $6$, $10$, $14$, $18$).
\end{enumerate}
\end{theorem}

The above property holds also in infinitely many higher dimensions $N$ for which symmetric conference matrices or skew Hadamard matrices are known.


\section{Unistochasticity of certain triangles embedded inside the Birkhoff polytope}\label{sec:complementary_matrices}
In this section we shall extend the work done in 1991 by Au-Yeung and Cheng~\cite{Au-Yeung_1991}, concerning convex combinations of pairs of permutation matrices.
The sets formed by the combinations shall be called edges, even though some of them form diagonals inside the Birkhoff polytope $\mathcal{B}_N$.

\begin{definition}
    A set of convex combinations of two permutation matrices $P$ and $Q$ forms an edge $\mathcal{E}$,
    \begin{equation}
        \mathcal{E} = \{\alpha P + (1-\alpha) Q \; |\; \alpha \in [0,1] \}.
    \end{equation}
\end{definition}

Au-Yeung and Cheng introduced the notion of complementary permutations.

\begin{definition}[Complementary permutations~\cite{Au-Yeung_1991}]
    Two permutations $P$ and $Q$ of size $N$ are called complementary if for all indices $i$, $j$, $k$, and $l \in \{1,...,N\}$ equalities $P_{ij} = P_{kl} = Q_{il} = 1$ imply $Q_{kj}=1$. 
\end{definition}

In other words, two matrices are complementary if for those non-zero elements of two matrices that share the same row $P_{ij} = Q_{il} = 1$ and the same column $P_{kl} = Q_{il} = 1$, the symmetric element of the $Q$ matrix is also non-zero, $Q_{kj} = 1$.
Using this notation the following statement was proved~\cite{Au-Yeung_1991}.

\begin{proposition}[\cite{Au-Yeung_1991}]
    If permutation matrices $P$ and $Q$ of size $N$ are complementary then the entire edge $PQ$ is orthostochastic. 
    If they are not complementary, then the edge is not unistochastic, apart from the permutation matrices themselves.
\end{proposition}

Slightly modifying the definition of Au-Yeung and Cheng, we propose a stronger notion.

\begin{definition}[Strongly complementary permutations]
    Two permutations $P$ and $Q$ of size $N$ are called strongly complementary if they are complementary and if $P_{ij} = 1$ implies $Q_{ij} = 0$.
\end{definition}

The dimensions $N$ for which a pair of strongly complementary matrices exists are necessarily even.
Furthermore, for all even dimensions, such a pair exists.
A complementary matrix to the identity might share some entries with it, while for a strongly complementary matrix it is impossible.
Now, we shall prove a statement concerning strongly complementary matrices.

\begin{proposition}
    If $P$ and $Q$ are strongly complementary matrices of size $N$, for which a robust Hadamard matrix exists, then the triangle $\Delta (P$, $Q$, $W_N)$, formed by convex combinations of permutation matrices and the van der Waerden matrix, is unistochastic.
    Furthermore, if there exists a real robust Hadamard matrix, then the triangle is orthostochastic.
\end{proposition}
\begin{proof}
Due to the symmetry of the Birkhoff polytope, without loss of generality, we may assume that one of the permutation matrices is the identity, $P = \mathbb{I}_N$.
Then, every bistochastic matrix belonging to the edge connecting $\mathbb{I}_N$ and $Q$ can be written as
\begin{equation}
    \begin{pmatrix}
			b & a & 0 & 0 & \dots\\
			a & b & 0 & 0 & \dots \\
			0 & 0 & b	& a & \dots\\
			0 & 0 & a & b & \dots\\
			\vdots & \vdots & \vdots & \vdots & \ddots & \\
		\end{pmatrix},
\end{equation}
up to an irrelevant permutations of rows or columns.
Therefore, any matrix belonging to the triangle $\Delta (P$, $Q$, $W_N$) have the following form
\begin{equation}
    \begin{pmatrix}
		b & a & c & c & \dots\\
		a & b & c & c & \dots \\
		c & c & b & a & \dots\\
		c & c & a & b & \dots\\
		\vdots & \vdots & \vdots & \vdots & \ddots & \\
		\end{pmatrix},
\end{equation}
where the normalization condition requires $a+b+(N-2)c= 1$.
Using the element-wise square root and element-wise product with the robust Hadamard matrix, we obtain a unitary matrix, which completes the proof.
\end{proof}

The reasoning presented above can also be extended to a set of pair-wise strongly complementary matrices $\{P_{1},..., P_{k}\}$.
Then, all $2$-faces of the polytope, formed by the convex hull of those permutation matrices and the van der Waerden matrix, are unistochastic.
Nonetheless, it is hard to give any bound on how big a set of pair-wise strongly complementary matrices can be, apart from that it must not be larger than the dimension of matrices, $k\leq N$.
Furthermore, the question of whether the matrices inside the polytope are unistochastic is still open.

Next, we shall move on to the study of one of the applications of rays inside the Birkhoff polytope in quantum information.

\section{Equi-entangled bases}\label{sec:equi-entangled}
Several tasks in quantum information require that the states we use share the same degree of entanglement, not necessarily extremal, i.e.\ utilizing states that are not maximally entangled but also not separable.
Examples of these applications include generalized Bell state measurement with not-maximally entangled states, as well as studying the effect of entanglement on the capacity of quantum channels by encoding words into separate basis states having equal amounts of entanglement~\cite{Karimipour_2006}.

The problem of constructing bases in which every vector possesses the same degree of entanglement was initiated by Karimipour and Memarzadeh in 2006~\cite{Karimipour_2006}.
They constructed such bases by means of the generalized Pauli operator $Z = \sum_{i=1}^N \ket{i}\bra{i\oplus 1}$, where the addition is understood modulo $N$.
The 2010 follow-up paper by Gheorghiu and Looi discussed a more general construction using quadratic Gauss sums~\cite{Gheorghiu_2010}.
We shall present another construction that is based upon the properties of the robust Hadamard matrices.
Let us start by recalling the definition from the paper by Gheorghiu and Looi~\cite{Gheorghiu_2010}.

\begin{definition}
    A family of bases $\mathcal{B}_t = \{\ket{\psi^t_i}\}^N_{i=1}$ is said to form equi-entangled bases if two following conditions are satisfied
    \begin{enumerate}
        \item the family interpolates continuously between a product basis and a basis of maximally entangled states and
        \item for a fixed value of the parameter $t$, all vectors $\{\ket{\psi^t_i}\}^N_{i=1}$ have the same degree of entanglement.        
\end{enumerate}
\end{definition}

Now, we shall prove the connection between unistochastic rays inside the Birkhoff polytope and equi-entangled bases.
Consider a bistochastic matrix $B_\alpha$ belonging to the unistochastic ray parametrized by $\alpha$,
and the unitary matrix $U_\alpha$ corresponding to $B_\alpha$.
Such a family of unitary matrices exists for all dimensions for which a robust Hadamard matrix exists, see Section~\ref{sec:rays_and_counter-rays}.
This family allows us to construct a family of equi-entangled bases.
The following reasoning was put forward by other collaborators in~\cite{Rajchel_robust}.

\begin{proposition}
    Let $U$ be a unitary matrix of size $N$, connected to a bistochastic matrix $B$ belonging to a ray of the Birkhoff polytope $\mathcal{B}_N$.
    Then, the set of $N^2$ vectors, belonging to a bipartite Hilbert space $\mathcal{H} = \mathcal{H}_N \otimes \mathcal{H}_N$, defined by 
    \begin{equation}
        \ket{\psi_{ij}} = \sum_{k=1}^N U_{ik} \ket{k}\otimes \ket{k\oplus j}, \; \; \text{for } i\text{, }j \in \{1,...,N\},
    \end{equation}
    where $\oplus$ denotes addition modulo $N$, forms an equi-entangled basis of $\mathcal{H}_N \otimes \mathcal{H}_N$.
\end{proposition}
\begin{proof}
    To check that the set of vectors forms a basis, it suffices to compute their scalar product
    \begin{equation}
        \braket{\psi_{ij}|\psi_{i'j'}} = \sum_{k,k'=1}^N \overline{U_{ik}} U_{i'k'} \braket{k|k'} \braket{k\oplus j|k'\oplus j'} = \sum_{k,k'=1}^N \overline{U_{ik}} U_{i'k'} \delta_{kk'} \delta_{jj'} = \sum_{k=1}^N \overline{U_{ik}} U_{i'k} \delta_{jj'} = \delta_{ii'}\delta_{jj'},
    \end{equation}
    where the last equality follows from the orthogonality of rows of $U$.
    Then, the degree of entanglement of the state $\ket{\psi_{ij}}$ can be analyzed via relabelling of the basis vectors.
    In the second Hilbert space we consider $\ket{k}_2 \coloneqq \ket{k\oplus j}$, separately for different $j$,
    \begin{equation}
        \ket{\psi_{ij}} = \sum_{k=1}^N U_{ik} \ket{k}_1 \otimes \ket{k}_2.
    \end{equation}
    
    Consequently, the above expression formulates the Schmidt decomposition of the vector $\ket{\psi_{ij}}$.
    The degree of entanglement is specified solely by the amplitudes of the Schmidt coefficients $U_{ik}$, so by the elements of rows of the bistochastic matrix, $B = B_{\alpha}$.
    
    Since all rows of a matrix belonging to a ray consist of the same elements $\{\frac{1+\alpha(N-1)}{N}, \frac{1-\alpha}{N}, ..., \frac{1-\alpha}{N} \}$ with some permutation, the Schmidt coefficients of all states defined by the rows are the same, rendering all vectors equi-entangled.
\end{proof}

The above reasoning does not rely on the particular measure of entanglement used, as it depends only on the Schmidt coefficients.
Finally, the association of the equi-entangled basis to a bistochastic matrix $B_\alpha$ from the ray can be done for all parameters $\alpha$.
Thus, the two extremal bases are those connected to the permutation matrix and to the central van der Waerden matrix.

As a result, the extremal bases are maximally entangled, with all Schmidt coefficients equal for the van der Waerden matrix, and fully separable for the permutation matrix, with only one non-zero Schmidt coefficient equal to 1.
Therefore, any continuous function describing entanglement solely via the Schmidt coefficients (such as von Neumann entropy) yields all intermediate values.
Finally, for all possible values of entanglement an equi-entangled basis of a bipartite Hilbert space $\mathcal{H}_N\otimes \mathcal{H}_N$ can be created, provided the existence of a robust Hadamard matrix of size $N$.

Our construction is advantageous because the distribution of coefficients in the computational basis is almost uniform.
We hope that a similar construction can also be extended to the multipartite scenario.
What is more, since equi-entangled bases constructed with a help of robust Hadamard matrices are formed by the straight lines inside the Birkhoff polytope, this construction is geometrically simpler than the earlier ones.

\section{Circulant matrices of dimension \texorpdfstring{$N=4$}{Lg}}\label{sec:circulant_matrices}
This section shall be devoted to the study of a particular subset of bistochastic matrices in dimension $N=4$.

\begin{definition}[Circulant matrix]
    A bistochastic matrix $C$ of size $N$ is called circulant if its $k$-th row is obtained by $(k-1)$ translations of the first row,
    \begin{equation}
        C_{i\oplus k,j\oplus k} = C_{i,j},
    \end{equation}
    where $\oplus$ sign denotes addition modulo $N$.
\end{definition}

In general, circulant matrices do not need to be bistochastic; however, for the purpose of this chapter, we shall restrict to only those contained in the Birkhoff polytope.
Bracelet circulant matrices of size $N=3$ are unistochastic, like all of the bracelet bistochastic matrices of this size, see Section~\ref{sec:bracelet_matrices}.
The set of circulant matrices of this dimension form a 2-dimensional subset of the Birkhoff polytope, as shown in Fig.~\ref{fig:triangle_hypocycloid}.

\begin{figure}[H]
    \input{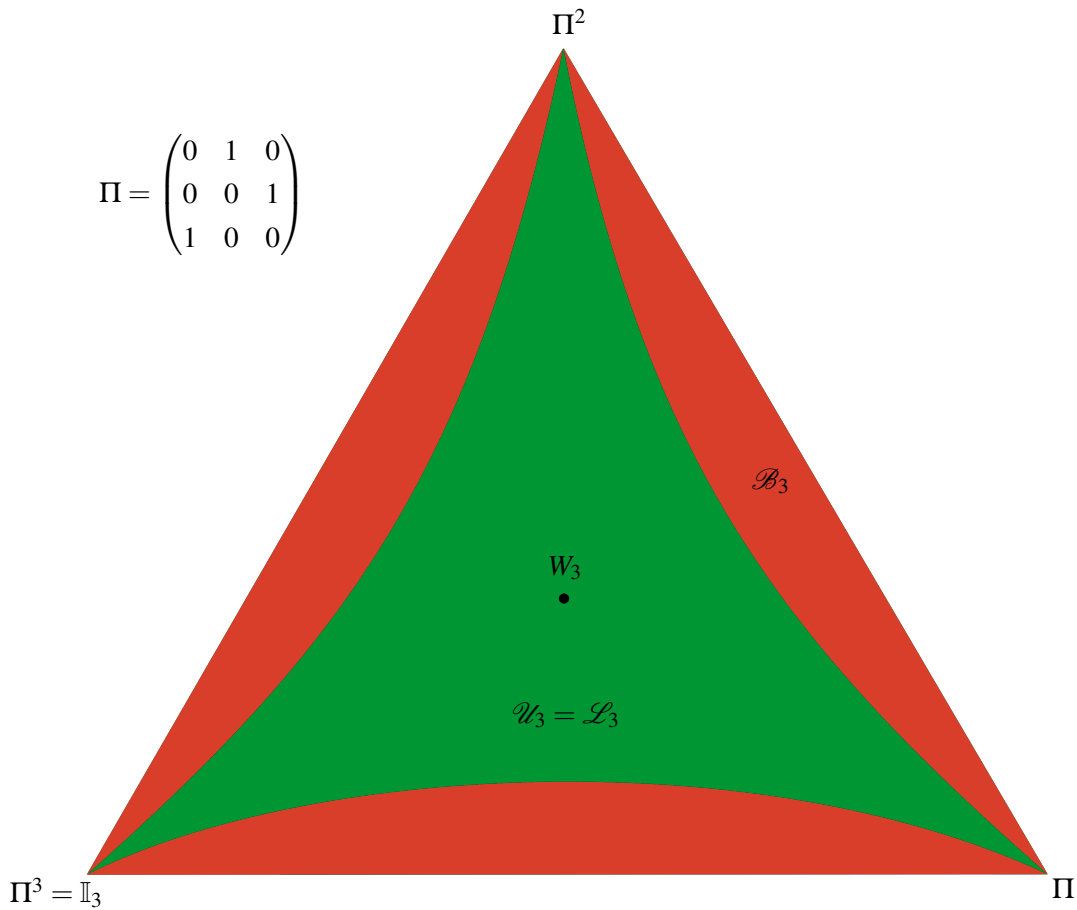}
    \caption{
    The set of circulant bistochastic matrices of order 3.
    All matrices of this size are unistochastic if they are bracelet, $\mathcal{U}_3 = \mathcal{L}_3$.
    The circulant matrices can be depicted using cross-section through the Birkhoff polytope $\mathcal{B}_3$ that consists of 3 permutation matrices: identity $\mathbb{I}_3$, a full cycle $\Pi$, and the other full cycle $\Pi^2$.
    The set of unistochastic matrices forms a full hypocycloid, with vertices given by the permutation matrices.
    Observe the centrality of the van der Waerden matrix $W_3$.
    }
    \label{fig:triangle_hypocycloid}
\end{figure}

In higher dimensions, in particular $N=4$, the bracelet conditions are only necessary, not sufficient, with existing counter-examples~\cite{Bengtsson_2005}.
Nonetheless, we shall show that regarding the subset of circulant bistochastic matrices of size 4 \emph{all} bracelet matrices are unistochastic.

\begin{theorem}\label{thm:circulant_size4}
    If bistochastic circulant matrix $C$ of size 4 is bracelet then it is unistochastic.
\end{theorem}
\begin{proof}
    Let us start with a bracelet circulant matrix $C$
    \begin{equation}
    C = \begin{pmatrix}
    a & b & c & d \\
    d & a & b & c\\
    c & d & a & b \\
    b & c & d  & a
    \end{pmatrix},
    \end{equation}
    with $a+b+c+d = 1$.
    Now, at least one of the products $ac$ or $bd$ is greater than zero since in the other case the matrix could not be bracelet.
    Furthermore, without loss of generality, let us assume that $ac \leq bd$.
    If the actual values satisfy the opposite inequality, by means of permutation of rows we are able to achieve the desired condition.
    Multiplying by a permutation matrix does not change neither bracelet nor unistochastic properties of any matrix.
    
    We create a circulant matrix with complex phases by taking the element-wise square root of $C$,
    \begin{equation}
    M = \begin{pmatrix}
    \sqrt{a} & e^{i\alpha}\sqrt{b} & e^{i\beta}\sqrt{c} & e^{i\gamma}\sqrt{d} \\
    e^{i\gamma}\sqrt{d} & \sqrt{a} & e^{i\alpha}\sqrt{b} & e^{i\beta}\sqrt{c}\\
    e^{i\beta}\sqrt{c} & e^{i\gamma}\sqrt{d} & \sqrt{a} & e^{i\alpha}\sqrt{b} \\
    e^{i\alpha}\sqrt{b} & e^{i\beta}\sqrt{c} & e^{i\gamma}\sqrt{d}  & \sqrt{a}
    \end{pmatrix}.
    \end{equation}
    
    We shall show that there exist angles $\alpha$, $\beta$, and $\gamma$ such that the matrix $M$ is unitary.
    Due to the symmetries of circulant matrices, the product $MM^\dagger$ is also highly symmetric
    \begin{equation}
    MM^\dagger = \begin{pmatrix}
    1 & m^* & n & m \\
    m & 1 & m^* & n\\
    n & m & 1 & m^* \\
    m^* & n & m  & 1
    \end{pmatrix},
    \end{equation}
    where 
    \begin{equation}\label{eq:thm_circulant4_m}
        m = e^{i\gamma}\sqrt{ad} + e^{-i\alpha}\sqrt{ab} + e^{i(\alpha-\beta)}\sqrt{bc} + e^{i(\beta-\gamma)}\sqrt{cd}
    \end{equation}
    and 
    \begin{equation}\label{eq:thm_circulant4_n}
        n = e^{i\beta}\sqrt{ac\mathstrut} + e^{i(\gamma-\alpha)}\sqrt{bd} + e^{-i\beta}\sqrt{ac\mathstrut} + e^{i(\alpha - \gamma)}\sqrt{bd} 
        = 2 \sqrt{ac\mathstrut}\cos{\beta} + 2 \sqrt{bd}\cos{(\alpha - \gamma)}.
    \end{equation}

    Since we want $m$ and $n$ to be zero, from Eq.~(\ref{eq:thm_circulant4_n}) we conclude that $\alpha = \gamma + \arccos{(-\sqrt{\frac{ac}{bd}}\cos{\beta})} = \gamma + f(\beta)$ with a new notation $f(\beta)$.
    For our purposes, we choose the function $\arccos(x)$ to map $[-1,1]$ into $[0,\pi]$.
    Inserting the value of $\alpha$ into Eq.~(\ref{eq:thm_circulant4_m}) we obtain
    \begin{equation}
            m = e^{i\gamma}\big(\sqrt{ad}+e^{i(f(\beta)-\beta)}\sqrt{bc}\big) 
    +e^{-i\gamma}\big(e^{ - if(\beta)}\sqrt{ab} + e^{i\beta}\sqrt{cd}\big).
    \end{equation}
    
    Therefore, the value of $m$ will also be zero provided that
    \begin{equation}\label{eq:thm_circulant4_gamma}
     e^{2i\gamma} = - \frac{e^{ - if(\beta))}\sqrt{ab} + e^{i\beta}\sqrt{cd}}{e^{i(f(\beta)-\beta)}\sqrt{bc}+\sqrt{ad}}.
    \end{equation}
    
    Finally, if in Eq.~(\ref{eq:thm_circulant4_gamma}) the right-hand side has modulus 1 then there exist angles $\alpha$, $\beta$, and $\gamma$ such that matrix $M$ is unitary.
    Now, we will show that this is always the case by considering a function $F: [\pi/2,3\pi/2] \to [0,2\pi)^2$ defined by 
    \begin{equation}
        F(x) = \big(x + \arccos{(\eta \cos{x})},x - \arccos{(\eta \cos{x})}\big),
    \end{equation}
    with $\eta = -\sqrt{\frac{ac}{bd}} \in [-1,0]$ since we have chosen $ac\leq bd$. 
    Importantly, for these values of the parameter $\eta$ the function $F$ is continuous with $F(\pi/2) = (\pi,0)$ and $F(3\pi/2) = (0,\pi)$.
    
    Furthermore, let us define the second function $G: [0,2\pi)^2 \to \mathbb{R}$, with positive parameters $a$, $b$, $c$, and $d \leq 1$, defined as
    \begin{equation}
        G(x,y) = |\sqrt{ab} + e^{ix}\sqrt{cd} | - |\sqrt{bc}+e^{iy}\sqrt{ad}|.
    \end{equation}
    
    Continuity of functions $G$ and $F$ implies continuity of the composite function $G \circ F$.
    Furthermore, values of $G \circ F$ are opposite at the boundaries of the domain,
    \begin{equation}
        G\big(F(\pi/2)\big) =  |\sqrt{ab} - \sqrt{cd} | - |\sqrt{bc}+\sqrt{ad}| 
        = \max{(\sqrt{ab},\sqrt{cd})}-\min{(\sqrt{ab},\sqrt{cd})}
        - \sqrt{bc}- \sqrt{ad} \leq 0
    \end{equation}
    and
    \begin{equation}
        G\big(F(3\pi/2)\big) =  |\sqrt{ab} + \sqrt{cd} | - |\sqrt{bc} - \sqrt{ad}|
        = \min{(\sqrt{bc},\sqrt{ad})} - \max{(\sqrt{bc},\sqrt{ad})}
       +\sqrt{ab} + \sqrt{cd}  \geq 0,
    \end{equation}
    where both inequalities hold if the matrix $C$ satisfies the bracelet conditions.
    
    Therefore, by application of the intermediate value theorem, we conclude that there exists some intermediate value of the parameter $\beta$ such that $G\big(F(\beta )\big) = 0$.
    The same $\beta$ yields a fraction of modulus 1 in Eq.~(\ref{eq:thm_circulant4_gamma}); therefore, it provides the parameters $\gamma$ and $\alpha = \gamma+ f(\beta)$ for which $M$ is a unitary matrix, which concludes the proof.    
\end{proof}

As a corollary following Theorem~\ref{thm:circulant_size4}, the set of circulant unistochastic matrices of size 4 is star-shaped with respect to the central van der Waerden matrix $W_4\in \mathcal{B}_4$.
Finally, using simplification provided by Theorem~\ref{thm:circulant_size4} we are able to probe the tetrahedron of unistochastic circulant matrices with high density, as depicted in Fig.~\ref{fig:tetrahedron_green}, and in the introductory 3D printout shown in Fig.~\ref{fig:model3D_unistochastic}.

\begin{figure}[H]
	\begin{tikzpicture}
		\node (myfirstpic) at (0,0) {\includegraphics[width=0.55\columnwidth]{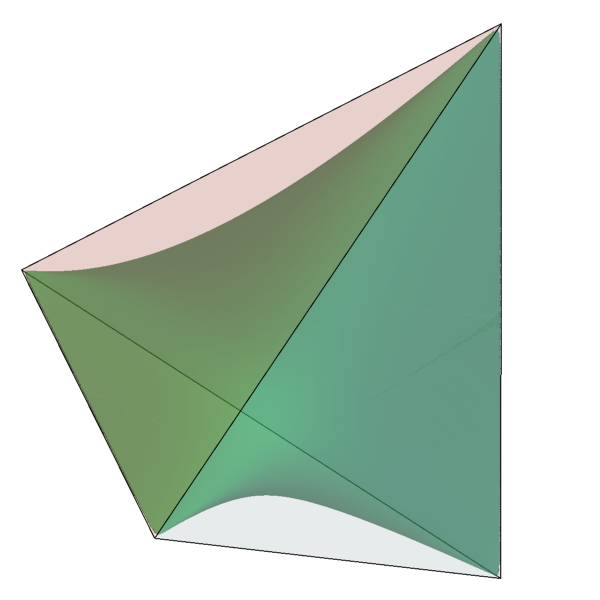}};		
		\node at (-0.17\columnwidth,-0.23\columnwidth) {\color{black}$\Pi^4_4 =\mathbb{I}_4^{}$};
		\node at (0.21\columnwidth,-0.27\columnwidth) {\color{black}$\Pi_4^{}$};
		\node at (-0.28\columnwidth,0.05\columnwidth) {\color{black}$\Pi_4^2$};
		\node at (0.21\columnwidth,0.26\columnwidth) {\color{black}$\Pi_4^3$};		
		\node [label={$\Pi_4^{}=\begin{pmatrix}0 &1 &0 &0\\0&0&1&0\\0&0&0&1\\1&0&0&0\end{pmatrix}$}] at (0.4\columnwidth,-0.09\columnwidth){};
	\end{tikzpicture}
	\caption{
	The set of $4\times 4$ circulant matrices represented as a tetrahedron, with its vertices given by four circulant permutation matrices $\mathbb{I}_4^{}$, $\Pi_4^{}$, $\Pi_4^2$, and $\Pi_4^3$.
	The emphasized non-convex subset of unistochastic matrices $\mathcal{C}_4 \cap\mathcal{U}_4$ is marked green and proved to be equal to the set of bracelet matrices (see Theorem~\ref{thm:circulant_size4}).
	The edges $(1,0,0,0) \leftrightarrow (0,0,1,0)$ and $(0,1,0,0) \leftrightarrow (0,0,0,1)$ are two of the lines at which the set of unistochastic matrices touches the surface of the tetrahedron.
	Two pairs of permutation matrices joined by the lines are strongly complementary, see Section~\ref{sec:complementary_matrices}.
    The only other remaining regions in which unistochastic matrices coincide with the surface of the tetrahedron are lines connecting midpoints of these two edges with the opposite vertices.
	Reproduced verbatim from the joint paper~\cite{Rajchel_algebraic_structures}.
	}\label{fig:tetrahedron_green} 
\end{figure}
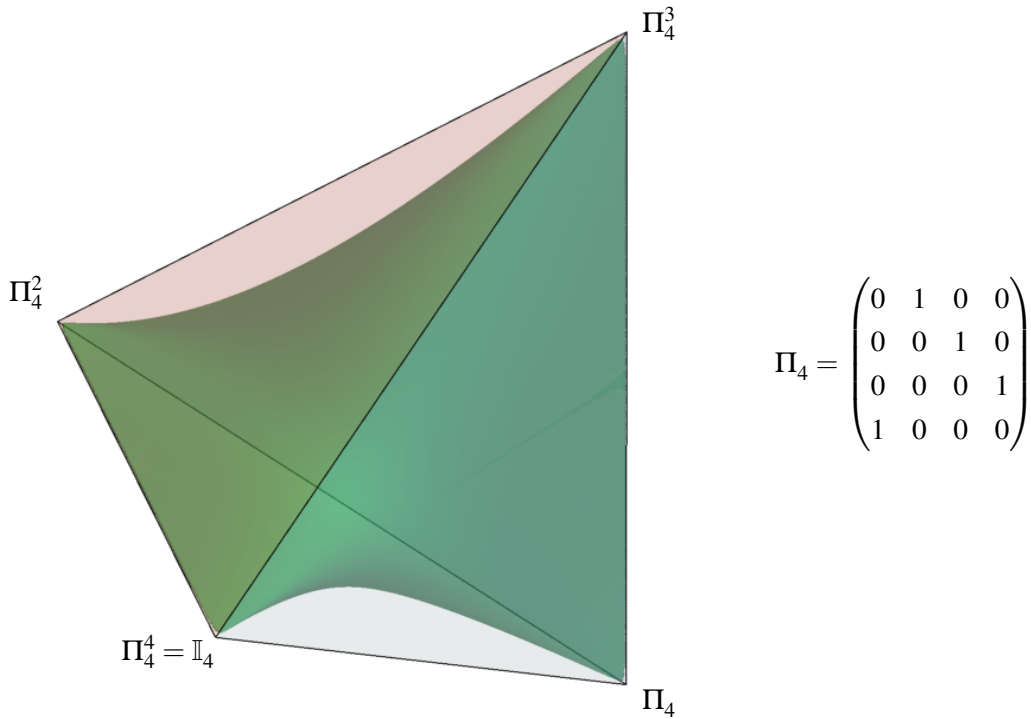


\section{Conclusions}
In this chapter we focused our study on the characterization of the set of unistochastic matrices.
We were able to find a numerical algorithm solving the unistochasticity problem as well as to prove the simple structure of circulant unistochastic matrices in the case of dimension 4.
Furthermore, we demonstrated more general results concerning the connection between robust Hadamard matrices, rays of the Birkhoff polytope, and equi-entangled bases.
The author hopes that proceeding in the direction of research started in this chapter our understanding of the set of unistochastic matrices shall be enriched in the near future, with a plethora of applications in quantum physics.
The discussion of the future prospects is delegated to Chapter~\ref{Summary}.

\clearpage
\chapter{Entangling power in multipartite systems}
\label{chapter_4}
\vspace{-1cm}
\rule[0.5ex]{1.0\columnwidth}{1pt} \\[0.2\baselineskip]

\section{Introduction}
Quantum entanglement is one of the cornerstones of the applications of quantum mechanical systems.
Many quantum information protocols rely on entanglement to achieve the gain over their classical counterparts, e.g.\ superdense coding~\cite{Bennett_1992}, measurement outperforming classical limit~\cite{Caves_1981}, quantum teleportation~\cite{Bennett_1993}, and many more~\cite{Nielsen_Chuang,Montanaro_2016}.
The study of entanglement in the bipartite case is believed to be mostly understood (see Section~\ref{sec:measures_of_entanglement}).
However, this is not the case with multipartite entanglement, for which a lot of problems arise.
One of them is the existence of the maximally entangled states, see subsequent Chapter~\ref{chapter_6} concerning AME states.

The other obstacle on the path to understanding the manipulations on multipartite systems is the limited study of entanglement creation by an action of a given unitary gate.
We face this problem by analytically evaluating the average entanglement created by unitary gates, which is the main goal of this chapter.
A summary of some parts of this chapter, as well as an extension of the others, in which the author's involvement was less substantial, can be found in a joint paper~\cite{Rajchel_entangling_power}.
If not specified differently, the author's contribution to the work covered by this chapter was significant.

\begin{figure}[H]
    \input{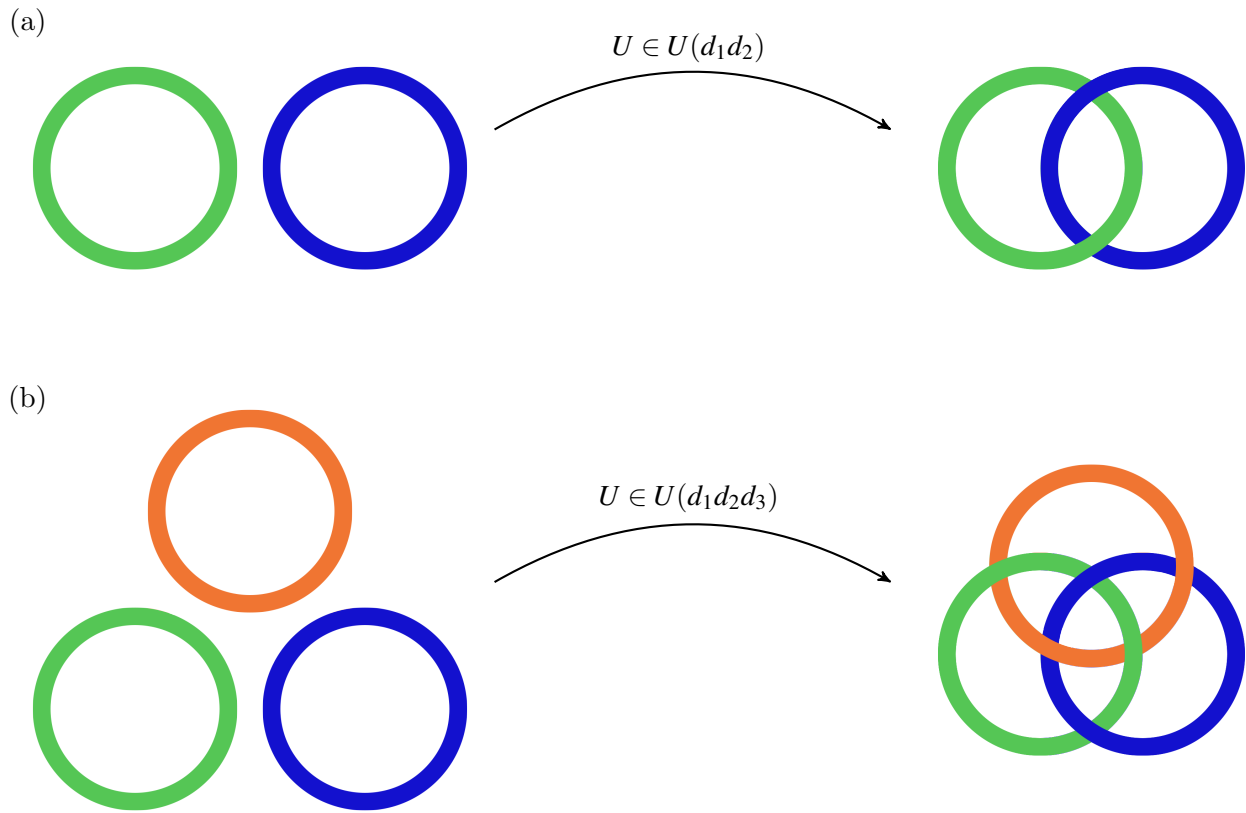}
    \caption{The entangling power of a given unitary matrix $U$ is given by the average entanglement created, while acting on separable states.
    Here, disjoint rings on the left depict separable states that, under the action of $U$, might become entangled (the right side of the figure).
    The bipartite case (a) is thoroughly explored in the literature, while a similar study for tripartite systems (b) was performed in this thesis.
    }
    \label{fig:entangling_power_sketch}
\end{figure}

\section{Entangling power in the tripartite case}
The entangling power, was introduced for the bipartite scenario in 2000 by Zanardi et al.~\cite{Zanardi_2000}.
It quantifies how much entanglement is created on average by applying a given quantum gate to a random pure state (see Section~\ref{sec:gates_ent_power}).
A similar analysis was lacking in the case of multipartite gates.
Motivated by this fact, we revisited the general requirement for the bipartite entangling power, see Eq.~(\ref{eq:def_e_p}), to extend it for the multipartite case.
An analogous condition can be imposed, the entangling power $\varepsilon_\tau$ of a unitary gate $U$ should amount to the average entanglement created by the action of the gate 
\begin{equation}\label{eq:def_entangling_power}
    \varepsilon_\tau (U) = \langle \tau\big( U\ket{\psi_{sep}} \big)\rangle_{\ket{\psi_{sep}}\in \mathcal{H}},
\end{equation}
where $\tau$ denotes the chosen multipartite measure of entanglement.
It was proved by Zanardi et al. that this formula, for the bipartite measure of entanglement being the linear entropy, can be converted to Eq.~(\ref{eq:def_e_p}).
The first to study a generalization of entangling power was Scott~\cite{Scott_2004}, whose results found applications in the context of quantum error correcting codes~\cite{Atushi_1996,Furuya_1998,Rossini_2004}.
Nonetheless, his formulae were restricted to several cases of equal dimensions, whereas our results are independent of the dimensionality of the Hilbert spaces.

As the measure of entanglement for the tripartite case we choose the one-tangle
\begin{equation}
    \tau_1 \big( \ket{\psi} \big) \coloneqq \frac{1}{3} \bigg( \tau_{12|3}\big( \ket{\psi} \big) + \tau_{13|2}\big( \ket{\psi} \big) + \tau_{23|1}\big( \ket{\psi} \big) \bigg),
\end{equation}
where the quantity $\tau_{ab|c}$ was introduced in Section~\ref{sec:measures_of_entanglement}.
Let us start describing our contribution by the lemma, for which the proof is included in the joint paper~\cite{Rajchel_entangling_power}.

\begin{lemma}
    The entangling power defined by Eq.~(\ref{eq:def_entangling_power}) for the one-tangle of a tripartite unitary gate $U$ is equivalent to
    \begin{equation}\label{eq:definition_ent_power_tripartite}
        \varepsilon_\tau (U) \equiv \varepsilon_{\tau_1}(U) = \frac{1}{3}\big( \varepsilon_{12|3}(U) +\varepsilon_{13|2}(U) + \varepsilon_{23|1} (U) \big),
    \end{equation}
    where the average entangling power of the gate $U$ with respect to the particular bipartition $ab|c$ reads
    \begin{equation}\label{eq:definition_ent_power_tripartite_one_bipartition}
        \varepsilon_{ab|c}(U) \coloneqq \big\langle \tau_{ab|c} (U\ket{\psi_{\text{\emph{sep}}}}) \big\rangle_{\ket{\phi_{\text{\emph{sep}}}}\in \mathcal{H}} = 2 \bigg[1- \bigg(\prod^3_{i=1}\frac{1}{d_i (d_i + 1)}\bigg)u_{\vec{r}}\:u_{\vec{s}}\:u_{\vec{t}}\:
	f_{\vec{r},\vec{s},\vec{t}}^{ab|c}(U) \bigg].
    \end{equation}
    Two other newly introduced notions include
    \begin{equation}\label{eq:u_r-u_s-u_t-expansion}
        u_{\vec{v}}\coloneqq\delta^{v_1}_{v_2}\delta^{v_3}_{v_4}+\delta^{v_1}_{v_4}\delta^{v_3}_{v_2}
    \end{equation}
    and
    \begin{equation}\label{eq:definition_f_function_e_p}
        f_{\vec{r},\vec{s},\vec{t}}^{ab|c}(U) \coloneqq
	\delta^{i_a}_{l_a\vphantom{(}}\delta^{i_b}_{l_b\vphantom{(}}\delta^{i_c}_{j_c\vphantom{(}}
	\delta^{k_a}_{j_a\vphantom{(}}\delta^{k_b}_{j_b\vphantom{(}}\delta^{k_c}_{l_c\vphantom{(}}\:
	U_{r_1s_1t_1}^{i_1i_2i_3}\,
	\left(U^\dag\right)_{j_1j_2j_3}^{r_2s_2t_2}\,
	U_{r_3s_3t_3}^{k_1k_2k_3}\,
	\left(U^\dag\right)_{l_1l_2l_3}^{r_4s_4t_4}.
    \end{equation}
    
    Note the Einstein summation convention.
    The vectors $\vec{r}$, $\vec{s}$, and $\vec{t}$ are of length four, each consisting of elements from the discrete sets $\{1,...,d_1\}$, $\{1,...,d_2\}$, and $\{1,...,d_3\}$, respectively.
\end{lemma}

The above, constructive expression for the entangling power will be used to find the averages over the orthogonal and the unitary group.
In the next section, we shall introduce the moments of the orthogonal group that are crucial to calculating the average entangling power.

\section{Moments of the orthogonal group}\label{sec:orthogonal_weingarten_functions}
To assist evaluation of the averages over the sets of matrices it will be helpful to recall mathematical notions that allow calculating of the moment over the unitary and the orthogonal groups.
In 1978, Weingarten introduced functions that were later named, honoring his contribution, the \emph{Weingarten functions}~\cite{Weingarten_1978}.
These are rational functions describing the integration of elements of matrices
\begin{equation}
    \int_{U(d)} U_{i_1 j_1}...U_{i_n j_n} U^\dagger_{k_1 l_1}...U^\dagger_{k_n l_n}\, dU,
\end{equation}
where the integration is taken over the entire unitary group $U(d)$ with respect to the Haar measure.
The lower indices of matrices denote the appropriate elements of the matrix.

Similar to their unitary counterparts, the orthogonal Weingarten functions were first described by Collins and {\'S}niady in 2006~\cite{Collins_2006}, with later extensions by Collins and Matsumoto in 2009~\cite{Collins_2009} and by Gu in 2013~\cite{Gu_2013}.
The orthogonal functions are given by the average over the (real) orthogonal group $O(d)$ with respect to the Haar measure,
\begin{equation}
    \int_{O(d)} O_{i_1 j_1}...O_{i_n j_n} O_{k_1 l_1}...O_{k_n l_n}\, dO.
\end{equation}

In this section we shall restrict to the example of second moments of the orthogonal group.
For a more general view we refer the reader to the paper of Collins and Matsumoto~\cite{Collins_2009}.
Regarding this special case, second moments read
\begin{equation}\label{eq:orthogonal_moments}
    \int_{O(d)} O_{i_1 j_1}O_{i_2 j_2} O_{k_1 l_1}O_{k_2 l_2}\, dO = \sum_{q,r\in \{p_1, p_2, p_3\}} \langle \text{Wg}^O(q),r\rangle\;\; \delta_{i_1}^{i_{q(1)}}\delta_{j_1}^{j_{q(1)}}\delta_{i_2}^{i_{q(2)}}\delta_{j_2}^{j_{q(2)}}
    \delta_{k_1}^{k_{q(1)}}\delta_{l_1}^{l_{q(1)}}\delta_{k_2}^{k_{q(2)}}\delta_{l_2}^{l_{q(2)}},
\end{equation}
where $\langle \text{Wg}^O(q),r\rangle$ denotes the orthogonal Weingarten function for two permutations $q$ and $r$.
The three special permutations of length four are given by pairs of disjoint transpositions $p_1 = \{(12)(34)\}$, $p_2 = \{(13)(24)\}$, and $p_{3} = \{(14)(23)\}$.

In order to evaluate the nine different values of $\langle \text{Wg}^O(p_i),p_j\rangle$, we shall follow the chain of reasoning presented by Gu~\cite{Gu_2013}.
First, we observe that the functions can be obtained by combining these two permutations into one, $p_i p_j$.
By breaking the new permutation into cycles, we are able to acquire their lengths.
Then, for every cycle of length $l$ there is a corresponding cycle of the same length, see Lemma 1.16 in~\cite{Gu_2013}.
Thus, it is possible to divide the lengths of cycles into two, the step which is necessary in order to evaluate the exact values of the orthogonal Weingarten functions.

As an example, consider one of the nine expressions given by two cycles $p_1$, reading $\langle \text{Wg}^O(p_1),p_1\rangle$.
In this case, the square of the cycle simplifies to
\begin{equation}
    p_1 p_1 = \{(12)(34)\} \; \{(12)(34)\} = \text{id} = \{(1)(2)(3)(4)\},
\end{equation}
yielding the identity permutation, which consists of 4 disjoint cycles of length 1.
Therefore, by taking half of them we obtain two cycles of length 1, what we write as $[1,1]$.
Ultimately, Gu in Appendix B of~\cite{Gu_2013} found this Weingarten function to be
\begin{equation}
    \langle \text{Wg}^O(p_1),p_1\rangle = \text{Wg}^O([1,1],d) = \frac{d+1}{d(d-1)(d+2)}.
\end{equation}

All the three functions $\langle \text{Wg}^O(p_i),p_i\rangle$ yield the same value.
The other 6 pairs of different permutations $\{p_i, p_j\}$ result in the joint permutation $p_i p_j$ consisting of two cycles of length 2, and so
\begin{equation}
    \langle \text{Wg}^O(p_i),p_j\rangle = \text{Wg}^O([2],d) = \frac{-1}{d(d-1)(d+2)} \;\;\; \text{for } i\neq j.
\end{equation}

Finally, employing Eq.~(\ref{eq:orthogonal_moments}), we conclude the results of this section by the ultimate equation for the second moments of matrices, averaged over the orthogonal group $O(d)$
\begin{equation}\label{eq:Weingarten_final}
\begin{split}
    \int_{O(d)} O_{i_1 j_1}& O_{i_2 j_2} O_{k_1 l_1}O_{k_2 l_2}\, dO = \frac{1}{d(d-1)(d+2)}\bigg((d+1)\,\big( \delta^{i_2}_{i_1}\delta^{k_2}_{k_1}\delta^{j_2}_{j_1}\delta^{l_2}_{l_1} + \delta^{k_1}_{i_1}\delta^{k_2}_{i_2}\delta^{l_1}_{j_1}\delta^{l_2}_{j_2} + \delta^{k_2}_{i_1}\delta^{k_1}_{i_2}\delta^{l_2}_{j_1}\delta^{l_1}_{j_2}\big) \\ 
    &- \delta^{i_2}_{i_1}\delta^{k_2}_{k_1}\delta^{l_1}_{j_1}\delta^{l_2}_{j_2} - \delta^{i_2}_{i_1}\delta^{k_2}_{k_1}\delta^{l_2}_{j_1}\delta^{l_1}_{j_2} - 
    \delta^{k_1}_{i_1}\delta^{k_2}_{i_2}\delta^{j_2}_{j_1}\delta^{l_2}_{l_1} - \delta^{k_1}_{i_1}\delta^{k_2}_{i_2}\delta^{l_2}_{j_1}\delta^{l_1}_{j_2} - \delta^{k_2}_{i_1}\delta^{k_1}_{i_2}\delta^{j_2}_{j_1}\delta^{l_2}_{l_1} - \delta^{k_2}_{i_1}\delta^{k_1}_{i_2}\delta^{l_1}_{j_1}\delta^{l_2}_{j_2}
    \bigg).
\end{split}
\end{equation}

Having found this expression, we shall move on to applying the above mathematical result to obtain the mean entangling power of orthogonal gates.

\section{Average entangling power of tripartite orthogonal gates}\label{sec:average_e_p_tripartite_orthogonal}
In order to evaluate how much entanglement do random gates create, it is beneficial to consider the mean of their entangling power.
This section shall be devoted to the discussion of the average entangling power of random gates from the orthogonal tripartite ensemble of size $d = d_1 d_2 d_3$, drawn with respect to the Haar measure.

Starting with Eq.~(\ref{eq:definition_ent_power_tripartite}), we note that in order to calculate the average entangling power of orthogonal gates it suffices to consider the average of only one of the three bipartitions, such as, without loss of generality, $12|3$.
Then, using Eq.~(\ref{eq:definition_ent_power_tripartite_one_bipartition}) with this particular choice of the bipartition, we arrive at the expression
\begin{equation}
    \langle \varepsilon_\tau (O) \rangle_{O(d)} = \langle \varepsilon_{12|3} (O) \rangle_{O(d)} = 2 \bigg[1- \bigg(\prod^3_{i=1}\frac{1}{d_i (d_i + 1)}\bigg)u_{\vec{r}}\:u_{\vec{s}}\:u_{\vec{t}}\:
	\big\langle f_{\vec{r},\vec{s},\vec{t}}^{12|3}(O)\big\rangle_{O(d)} \bigg].
\end{equation}

Now, making use of Eq.~(\ref{eq:definition_f_function_e_p}) we conclude that
    \begin{equation}
        \big\langle f_{\vec{r},\vec{s},\vec{t}}^{12|3}(O)\big\rangle_{O(d)} =
	\delta^{i_1}_{l_1\vphantom{(}}\delta^{i_2}_{l_2\vphantom{(}}\delta^{i_3}_{j_3\vphantom{(}}
	\delta^{k_1}_{j_1\vphantom{(}}\delta^{k_2}_{j_2\vphantom{(}}\delta^{k_3}_{l_3\vphantom{(}}\:
	\int_{O(d)} O_{r_1s_1t_1}^{i_1i_2i_3}\,
	O_{r_3s_3t_3}^{k_1k_2k_3}\,
	\left(O^T\right)_{j_1j_2j_3}^{r_2s_2t_2}\,
	\left(O^T\right)_{l_1l_2l_3}^{r_4s_4t_4} \; dO.
    \end{equation}
Applying the results of Section~\ref{sec:orthogonal_weingarten_functions} we evaluate the integral to be a rather lengthy combination of Kronecker deltas
\begin{equation}
\begin{split}
            \big\langle &f_{\vec{r},\vec{s},\vec{t}}^{12|3}(O)\big\rangle_{O(d)} =
	\delta^{i_1}_{l_1\vphantom{(}}\delta^{i_2}_{l_2\vphantom{(}}\delta^{i_3}_{j_3\vphantom{(}}
	\delta^{k_1}_{j_1\vphantom{(}}\delta^{k_2}_{j_2\vphantom{(}}\delta^{k_3}_{l_3\vphantom{(}}\:
	\big\langle O^{i_1i_2i_3}_{r_1s_1t_1}\,
	O^{k_1k_2k_3}_{r_3s_3t_3}\,
	O_{r_2s_2t_2}^{j_1j_2j_3}\,
	O_{r_4s_4t_4}^{l_1l_2l_3}
	\big\rangle_{O(d)}  \\ 
	    &= \frac{\delta^{i_1i_2i_3}_{l_1l_2j_3\vphantom{(}}\delta^{k_1k_2k_3}_{j_1j_2l_3\vphantom{(}}}{d(d-1)(d+2)}
	\bigg(\big(d+1\big)
	\big(   
	    \delta^{k_1k_2k_3}_{i_1i_2i_3}\delta^{l_1l_2l_3}_{j_1j_2j_3}\delta^{3}_{1}\delta^{4}_{2}
	    + \delta^{j_1j_2j_3}_{i_1i_2i_3}\delta^{l_1l_2l_3}_{k_1k_2k_3}\delta^{2}_{1}\delta^{4}_{3} + \delta^{l_1l_2l_3}_{i_1i_2i_3}\delta^{j_1j_2j_3}_{k_1k_2k_3}\delta^{4}_{1}\delta^{3}_{2}\big)
    - \delta^{k_1k_2k_3}_{i_1i_2i_3}\delta^{l_1l_2l_3}_{j_1j_2j_3}\delta^{2}_{1}\delta^{4}_{3} \\ 
    &- \delta^{k_1k_2k_3}_{i_1i_2i_3}\delta^{l_1l_2l_3}_{j_1j_2j_3}\delta^{4}_{1}\delta^{3}_{2} - 
    \delta^{j_1j_2j_3}_{i_1i_2i_3}\delta^{l_1l_2l_3}_{k_1k_2k_3}\delta^{3}_{1}\delta^{4}_{2} - \delta^{j_1j_2j_3}_{i_1i_2i_3}\delta^{l_1l_2l_3}_{k_1k_2k_3}\delta^{4}_{1}\delta^{3}_{2} -  
    \delta^{l_1l_2l_3}_{i_1i_2i_3}\delta^{j_1j_2j_3}_{k_1k_2k_3}\delta^{3}_{1}\delta^{4}_{2} - 
    \delta^{l_1l_2l_3}_{i_1i_2i_3}\delta^{j_1j_2j_3}_{k_1k_2k_3}\delta^{2}_{1}\delta^{4}_{3}
    \bigg),
\end{split}
\end{equation}
where, for conciseness, by $\delta^{ijk}_{lmn}$ we denote $\delta^i_l\, \delta^{j\vphantom{(}}_{m\vphantom{(}}\, \delta^k_n$, while, for $i$ and $j$ being numbers, $\delta^i_j$ means $\delta^{r_is_it_i}_{r_js_jt_j}$.

Having all of the necessary tools, we may proceed to the evaluation of the terms with Kronecker deltas.
These terms, due to their implied summation, simplify to the dimension of the appropriate subsystem $d_1$, $d_2$, or $d_3$
\begin{equation}
\begin{split}
    \big\langle f_{\vec{r},\vec{s},\vec{t}}^{12|3}(O)\big\rangle_{O(d)} =\frac{(d+1 - d_3 - d_1d_2)\,\delta^3_1 \delta^4_2 
    + (dd_3+d_3 - 1 - d_1d_2)\, \delta^2_1\delta^4_3 + 
    (dd_1d_2+d_1d_2 - 1 - d_3)\,\delta^4_1\delta^3_2}{(d-1)(d+2)}.	
\end{split}
\end{equation}

Now, employing Eq.~(\ref{eq:u_r-u_s-u_t-expansion}), we expand another term of the entangling power into products of Kronecker deltas,
\begin{equation}
\begin{split}
    u_{\vec{r}}\:u_{\vec{s}}\:u_{\vec{t}}\: = 
    &\big(\delta^{r_1 r_3}_{r_2 r_4}+\delta^{r_1 r_3}_{r_4 r_2}\big)
    \big(\delta^{s_1 s_3}_{s_2 s_4}+\delta^{s_1 s_3}_{s_4 s_2}\big)
    \big(\delta^{t_1 t_3}_{t_2 t_4}+\delta^{t_1 t_3}_{t_4 t_2}\big), 
\end{split}
\end{equation}
which leads to an additional simplification of the analyzed average,
\begin{equation}
\begin{split}
            &u_{\vec{r}}\:u_{\vec{s}}\:u_{\vec{t}}\:\big\langle f_{\vec{r},\vec{s},\vec{t}}^{12|3}(O)\big\rangle_{O(d)} = \\
            &\frac{d\big(8(d+1 - d_3 - d_1d_2) + (dd_3+dd_1d_2+d_3+d_1d_2 - 2 - d_1d_2 - d_3)(d_1+1)(d_2+1)(d_3+1)\big)}{(d-1)(d+2)}.
\end{split}
\end{equation}

Finally, the average entangling power of orthogonal tripartite gates $O(d_1d_2d_3) = O(d)$ reads
\begin{equation}\label{eq:average_e_p_tripartite_orthogonal}
    \langle \varepsilon_\tau (O) \rangle_{O(d)} = \langle \varepsilon_{12|3} (O) \rangle_{O(d)} = \frac{(3d+3-d_1-d_2-d_3-d_{1}d_{2}-d_{1}d_{3}-d_{2}d_{3})\big(d(d_1+1)(d_2+1)(d_3+1)-8\big)}{\frac{3}{2}(d-1)(d+2)(d_1+1)(d_2+1)(d_3+1)}.
\end{equation}

We conclude this section by noting that the above result can be used to verify the efficacy of experiments aimed at simulating behavior of random quantum circuits when restricted to the subclass of orthogonal ones.
The following section aims at providing an analogous result for a superset of orthogonal gates, namely the unitary ones.

\section{Comparison between orthogonal and unitary groups}
Analogous results to those from the previous section for the case of the unitary group have been derived using Weingarten functions~\cite{Rajchel_entangling_power}, with the proof contained in the joint paper~\cite{Rajchel_entangling_power}.

\begin{proposition}
    Average entangling power $\varepsilon_\tau$ of tripartite unitary gates over the Hilbert space \\ {$\mathcal{H}_{d_1}\otimes\mathcal{H}_{d_2} \otimes \mathcal{H}_{d_3}$}, drawn with respect to the Haar measure, reads
    \begin{equation}
        \langle \varepsilon_\tau \rangle_U = \frac{3d_{1}d_{2}d_{3}+3-d_1-d_2-d_3-d_{1}d_{2}-d_{1}d_{3}-d_{2}d_{3}}{\frac{3}{2}(d_{1}d_{2}d_{3}+1)}.
    \end{equation}
\end{proposition}

We compare the results for the mean entangling power of unitary and orthogonal tripartite qu$d$it gates in Fig.~\ref{fig:average_tripartite_ep}, in the case equal local dimensions.
In the simplest case of three-qubit system, the mean of orthogonal gates reads $\langle \varepsilon_\tau \rangle_{O(8)}= \frac{208}{315}\approx0.660$, while the average over the unitary gates is given by $\langle \varepsilon_\tau \rangle_{U(8)} = \frac{2}{3}$.

\begin{figure}[H]
    \input{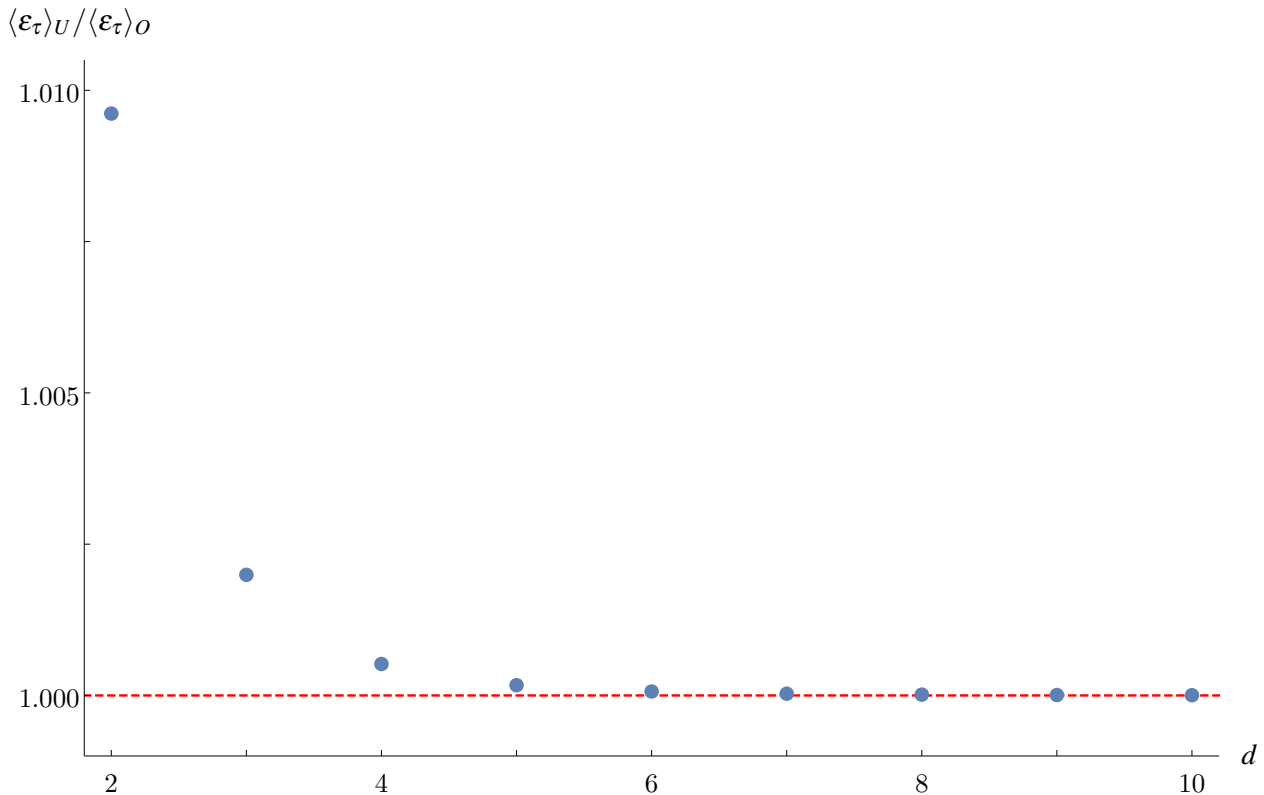}
    \caption{Fraction of the average entangling power of unitary tripartite qu$d$it gates $\langle \varepsilon_\tau \rangle_U$ over the average of the orthogonal ones $\langle \varepsilon_\tau \rangle_O$ approaches 1 as the dimension $d$ grows.
    }
    \label{fig:average_tripartite_ep}
\end{figure}

The dissimilarities between the unitary group and its proper subgroup of orthogonal operators diminish as the dimension increases.
This lead us to conjecture that in order to simulate random high-dimensional unitary gates it might be sufficient to restrict to orthogonal ones, at least from the perspective of tasks based on the average entanglement.

\section{Extension to multipartite gates}
Finally, as the most general result of the joint paper~\cite{Rajchel_entangling_power}, we present the characterization of the entangling power in the multipartite case.
The reasoning is analogous to the tripartite case.
Let us start by defining the notation.
Using elements of matrix of size $d_1...d_N$
\begin{equation}
    U^{j_1...j_N}_{i_1...i_N} = \bra{j_1...j_N}U\ket{i_1...i_N},
\end{equation}
we introduce its Choi-Jamio{\l}kowski isomorphism (see Section~\ref{sec:sets_of_matrices})
\begin{equation}
    U \mapsto \ket{U} = \frac{1}{\sqrt{d_1...d_N}}\sum_{i_1,...,i_N,j_1,...,j_N} U^{i_1...i_N}_{j_1...j_N} \ket{i_1...i_N j_1...j_N}.
\end{equation}

Then, the definition of entangling power of a multipartite operator can be expressed through the state connected by Choi-Jamio{\l}kowski isomorphism.

\begin{theorem}
    Setting the generalized concurrence as the multipartite measure of entanglement (see Section~\ref{sec:measures_of_entanglement}), the entangling power $\varepsilon_\tau$ of a multipartite gate $U$ acting on $N$ particles given by Eq.~(\ref{eq:def_entangling_power}) is equivalent to 
    \begin{equation}
        \varepsilon_\tau (U) = \frac{1}{2^{N-1}-1}\sum_{A|B} \varepsilon_{A|B} (U),
    \end{equation}
    where
    \begin{equation}
        \varepsilon_{A|B} (U)  = 2 \bigg[1 - \bigg( \prod^{N}_{i=1} \frac{d_i}{d_1 + 1} \bigg) \sum_{C|D} \text{\emph{Tr}}_{BC}\big(\text{\emph{Tr}}_{AD} \ket{U}\bra{U}\big)^2 \bigg].
    \end{equation}
\end{theorem}

The last result, relevant for our considerations, is a theorem concerning averages of entangling power of multipartite gates.

\begin{theorem}\label{thm:average_e_p_multipartite_gates}
Average entangling power of a multipartite gate $U$, acting on the Hilbert space $\mathcal{H}_{d_1}\otimes ...\otimes \mathcal{H}_{d_N}$ while averaging over
    \begin{enumerate}
        \item the orthogonal group $O(d_1...d_N)$ reads
        \begin{equation}
            \braket{\varepsilon_{\tau}}_{O(d_1\ldots d_N)}
	=2\left[1-\left(\prod_{i=1}^N\frac{1}{d_i+1}\right)
	\frac{2^N(D+1)-2B+\frac{BD-2^N}{2^{N-1}-1}C}
	{(D-1)(D+2)}\right],
        \end{equation}
        \item the unitary group $U(d_1...d_N)$ reads
        \begin{equation}
            \braket{\varepsilon_{\tau}}_{U(d_1\ldots d_N)}
	= 2\left[1-\left(\prod_{i=1}^N\frac{1}{d_i+1}\right)
	\frac{BC}
	{(2^{N-1}-1)(D+1)}\right],
        \end{equation}
    \end{enumerate}
    where the auxiliary variables denote $B\coloneqq \sum_{i_1,\ldots,i_N=0}^1 d_1^{i_1}\ldots d_N^{i_N}$, $C\coloneqq\sum_{A|B}(d_A+d_B)$, and $D\coloneqq d_1\ldots d_N$.
\end{theorem}

In the case of equal dimensions of subsystems, the expressions from Theorem~\ref{thm:average_e_p_multipartite_gates} can be simplified.
\begin{corollary}
Mean entangling power of a $N$ qu$d$it gate $U$, acting on the Hilbert space $\mathcal{H}_{d}^{\otimes N}$ while averaging over
    \begin{enumerate}
        \item the orthogonal group $O(d^N)$ reads
        \begin{equation}
            \braket{\varepsilon_{\tau}}_{O(d^N)}
	=
	\frac{[2^N(d^N+1)-2(d+1)^N][d^N(d+1)^N-2^N]}
	{(2^{N-1}-1)(d^{2N}+d^N-2)(d+1)^N},
        \end{equation}
        \item the unitary group $U(d^N)$ reads
        \begin{equation}
            \braket{\varepsilon_{\tau}}_{U(d^N)}
	= \frac{2^N(d^N+1)-2(d+1)^N}{(2^{N-1}-1)(d^N+1)}.
        \end{equation}
    \end{enumerate}
    In the asymptotic limit $N\to \infty$ both values are equal, $\lim_{N\to \infty} 
    \braket{\varepsilon_{\tau}}_{O(d^N)} = \lim_{N\to \infty} \braket{\varepsilon_{\tau}}_{U(d^N)}=\frac{1}{2}$.
\end{corollary}

The above result concludes the main achievements of the joint paper~\cite{Rajchel_entangling_power}.

\section{Conclusions}
The entangling power of a bipartite unitary gate is a widely used notion that gives a lot of theoretical insight, as well as finds its practical applications in various quantum information protocols.
Optimization of the bipartite entangling power for unitary gates from $U(36)$ is the topic of the subsequent Chapter~\ref{chapter_6}.
The results provided in this chapter concerned the characterization of the more general, multipartite case.
Though theoretically inspiring, mathematical difficulties rendered this field not entirely explored, with our work likely to be the first contribution addressing this question. 

Our research supplies with the analytical expression of statistical properties of the entangling power, for the one-tangle as a measure of entanglement, in the case of tripartite and multipartite gates.
The averages were computed with respect to Haar measures on both unitary and orthogonal groups of an appropriate size.

A similar experimental problem of estimating a distribution of outcomes, given the action of random quantum gates, was presented recently by the Google team, which they claim proved quantum advantage~\cite{Arute_2019}.
Therefore, we are tempted to believe that our results will prove useful for benchmarking quantum devices, with a possible extension into the domain of entanglement creation.
The discussion of suggested future investigations is presented in Chapter~\ref{Summary}.


\clearpage
\part{Quantum designs}
\chapter{Absolutely maximally entangled state of 4 quhexes}
\label{chapter_6}
\vspace{-1cm}
\rule[0.5ex]{1.0\columnwidth}{1pt} \\[0.2\baselineskip]

\section{Introduction}
In Chapter~\ref{chapter_4}, we remarked that the issue of multipartite entanglement is less explored than its bipartite counterpart.
For instance, in the former case the existence of maximally entangled states is not known for all dimensions and number of parties~\cite{Table_AME}.
The present chapter shall explain techniques that lead to us to construct such a state in the smallest setup that hitherto was unsettled.

To this end, we shall consider an application of the previously introduced general notion of \emph{quantum designs}.
The usage of a design will be restricted to the specific ordering or arrangements of states, which entails particular properties as a consequence.
We will show the application of designs that lead us to the discovery of an absolutely maximally entangled state of four quantum systems of dimension 6, also called quhexes.

An absolutely maximally entangled (AME) state $\ket{\psi}_{ABCD}\in \mathcal{H}_N^{\otimes 4}$ is maximally entangled with respect to all symmetric bipartitions $AB|CD$, $AC|BD$, and $AD|BC$.
A unitary matrix connected to such a state corresponds to a gate $U\in U(N^2)$ with the maximal bipartite entangling power, connecting the present chapter to the previous one.

A summary of some parts of this chapter, as well as an extension of the others, in which the author's involvement was less substantial, can be found in a joint paper~\cite{Rajchel_AME}.
If not specified differently, the author's contribution to the work covered by this chapter was significant.

\begin{figure}[H]
    \includegraphics[width=1.\columnwidth]{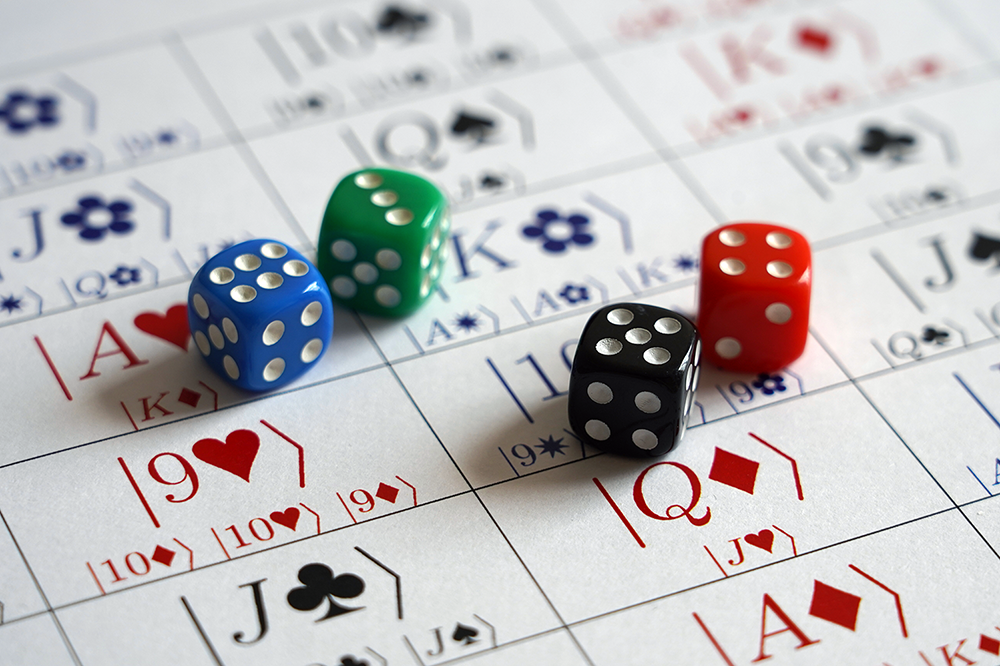}
    \caption{Absolutely maximally entangled states of four quhexes do exist and they are perfect to transmit the result of a dice game using quantum error correction codes.
    Although every pair of dice is unbiased, their outcome determines the results of the other two.
    The author is grateful to his wife for this artistic representation of an application of discovered AME states.}
    \label{fig:dice_AME}
\end{figure}

\section{Bipartite and multipartite entanglement}
Since the advent of quantum theory and the 1935 letter from Schr{\"o}dinger to Einstein~\cite{Schroedinger_1935}, the study of quantum entanglement was compelling from the point of view of theoretical physics.
Not surprisingly, the oddity of the new notion posed many challenging and long unanswered questions, to name only one: the EPR paradox.
In 1935 Einstein, Podolsky and Rosen considered a problem of apparent misalignment of quantum mechanics and theory of relativity by a consideration of two entangled particles~\cite{EPR_1935}. 
The experimental solution to this paradox was given a theoretical background by a seminal paper in 1964 by John Bell~\cite{Bell_1964}.
A specific setup provided by the group of Aspect in the early 1980s~\cite{Aspect_1981}, together with the final answer to certain loopholes~\cite{Aspect_1999}, inclined many experts to agree that entanglement is consistent with the theory of relativity and properly describes characteristics of nature.

The above example of the surprising properties of entanglement shows how the community of physicists remained not entirely convinced as to the genuine validity of quantum theory. However, at the same time, it reveals current consensus among scientists concerning bipartite entanglement, with only two quantum subsystems exhibiting quantum correlations.

To compare the degree of entanglement possessed by different quantum states, certain \emph{measures} of entanglement were introduced -- see Section~\ref{sec:measures_of_entanglement}.
The measures of entanglement in the bipartite case are currently believed to be well-understood~\cite{Plenio_2007}.


This belief stems from the two basic facts: (i) all states are fully characterized by the entanglement of formation~\cite{Bennett_1996} and (ii) entanglement monotones completely describe any state, up to local unitary (LU) operations~\cite{Vidal_1999}.
There exists a simple classification of maximally entangled pure states in a bipartite Hilbert space $\mathcal{H}_N\otimes \mathcal{H}_N $, as all of them are LU-equivalent to the generalized Bell state~\cite{Plenio_2007}
\begin{equation}
    \ket{\psi} = \frac{1}{\sqrt{N}}(\ket{11}+...+\ket{NN}).
\end{equation}

Summarizing the above description, almost ninety years after the paper by Schr{\"o}dinger some basic features of bipartite entanglement are presented in graduate textbooks~\cite{Peres_quantum_theory,Englert_lectures_on_quantum_mechanics}.
However, the situation is different in the domain of multipartite entanglement.
Consider two exemplary three-qubit states, belonging to the Hilbert space $\mathcal{H}_8 = \mathcal{H}_2\otimes \mathcal{H}_2 \otimes \mathcal{H}_2$: 
\begin{equation}
    \ket{W} = \frac{1}{\sqrt{3}}(\ket{112}+\ket{121}+\ket{211})
\end{equation}
and
\begin{equation}
    \ket{GHZ} = \frac{1}{\sqrt{2}}(\ket{111}+\ket{222}).
\end{equation}

D\"ur et al.\ showed in 2000~\cite{Dur_2000} that the states $\ket{W}$ and $\ket{GHZ}$ are not equivalent in a sense that none of them can be converted into the other one utilizing local operations and classical communication, see Section~\ref{sec:measures_of_entanglement}.
In the same paper the authors show that, in the case of more than three subsystems, two generic pure states cannot be converted into each other, not even with a small probability of success.
Both states $\ket{W}$ and $\ket{GHZ}$ are maximally entangled in their own respective classes, i.e.\ every state can be achieved by LOCC using one of them.
Therefore, the impossibility of conversion between them indicates an interesting property of existence of inequivalent maximally entangled states.
Even in a relatively simple case of three qubits, entanglement does possess profoundly different characteristics than in the bipartite systems.

\section{Absolutely maximally entangled (AME) states}
Multipartite entanglement displays an amusing trait: not always do maximally entangled states exist.
To explore the topic, we need to understand what does the term \emph{maximally entangled} mean in the case of more than two subsystems.
Since it is not possible to find a single measure of entanglement inducing an order into the set of multipartite quantum states, let us artificially divide the system into two parts, as shown in Fig.~\ref{fig:division_AME} in the example of 4 subsystems.

\begin{figure}[H]
\includegraphics[width=1\columnwidth]{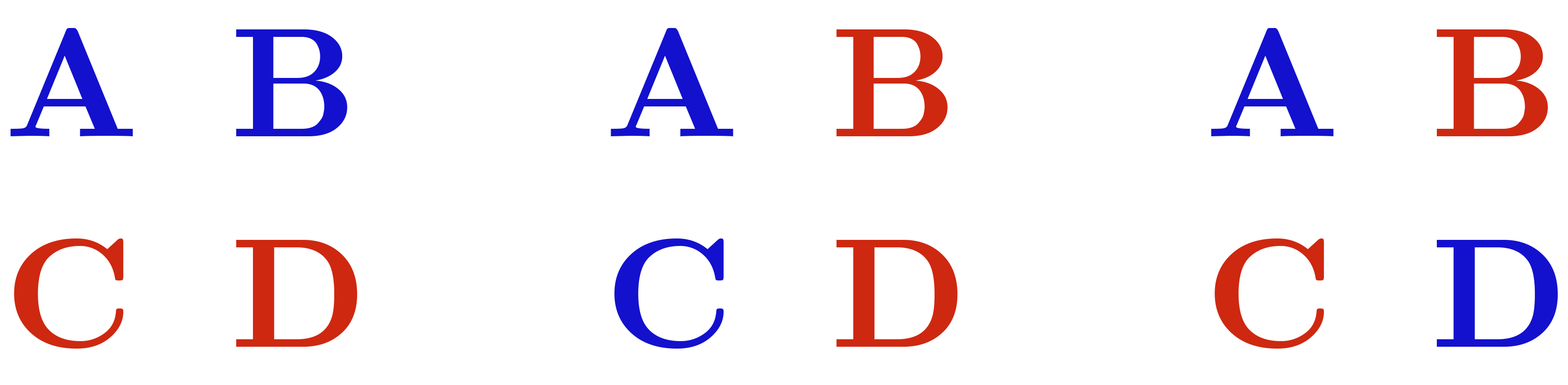}
    \caption{Three possible symmetric partitions of a system composed of four subsystems A, B, C, and D, marked by red and blue letters. 
    }\label{fig:division_AME} 
\end{figure}

Using splittings presented in Fig~\ref{fig:division_AME} one defines absolutely maximally entangled (AME) state of four subsystems as those for which all bipartitions are maximally entangled, in the sense of bipartite entanglement.
Following Helwig and Cui~\cite{Helwig_2012} we define AME states in the general case of $d$ subsystems.

\begin{definition}[AME state]\label{def:AME_state}
    A pure state $\ket{\psi}$ of Hilbert space $\mathcal{H}_N^{\otimes d}$ of $d$ subsystems labeled $Z=\{Z_1,...,Z_d\}$ of $N$ internal levels each is called absolutely maximally entangled -- AME($d$,$N$) -- if it meets one of the following equivalent requirements.
    \begin{enumerate}
        \item State $\ket{\psi}$ is maximally entangled with respect to every possible bipartition, i.e.\ splitting subsystems into two disjoint sets $X$ and $Y$ such that $X \bigcup Y = Z$ and $X\bigcap Y =\emptyset$. 
        Without loss of generality, we can assume that the number of subsystems in both partitions are $m = |X| \leq |Y| = d-m$.
        From the operational point of view, this means that the state can be written as
        \begin{equation}
            \ket{\psi} = \frac{1}{\sqrt{N^m}} \sum_{\vec{k}} \ket{k_1}_{X_1}\ket{k_2}_{X_2} ... \ket{k_m}_{X_m} |\phi(\vec{k})\rangle_Y,
        \end{equation}
        where the summation is understood over all possible multi-indices $\vec{k}$ with orthogonality relations $\langle \phi(\vec{k}) |\phi(\vec{k'})\rangle = \delta_{\vec{k}\vec{k'}}$.
        This state is LU-equivalent to the generalized Bell state; therefore, it is a maximally entangled state.
        \item Reduced density matrix of every subset of $X\subset Z$, for which the number of subsystems equals $|X| =\left \lfloor{\frac{d}{2}}\right \rfloor = m $, is maximally mixed,
        \begin{equation}
            \rho_X = \text{Tr}_{Z\setminus X}\big(\ket{\psi}\bra{\psi}\big) = \frac{1}{N^{m}}\mathbb{I}_{N^{m}}.
        \end{equation}
        \item Reduced density matrix of every subset of $X\subset Z$ that contains fewer than the half of all subsystems, $ |X| \leq \left \lfloor{\frac{d}{2}}\right\rfloor$, is maximally mixed.
    \end{enumerate}
\end{definition}

Equivalence between (2) and (3) stems from the fact that it suffices to check bipartite entanglement in extremal cases, in which the sizes of two subsystems are as close as possible.
AME states were first introduced for $N$ qubit systems by Facchi et al.\ in 2008~\cite{Facchi_2008}, with an extension for systems of a larger local dimension in 2012 by Helwig and Cui~\cite{Helwig_2012}.
These states prove to be useful from the perspective of multiple quantum protocols, including quantum error correction codes~\cite{Raissi_2018}, quantum maximal distance separable codes~\cite{Alsina_2019}, and open teleportation~\cite{Helwig_2012}.

AME states do not exist in all possible setups, specified by the number of subsystems $d$ and local dimensions $N$.
One of the goals of current research in this field is to find for which setups there exists AME($d$,$N$) state.
In many small local dimensions these attempts were successful, e.g.\ for $d$ states with local dimension 2 (qubits) it was proven that AME($d$,2) states exist only for $d=2,3,5$ and $6$, while they do not exist for all other dimensions~\cite{Scott_2004,Huber_2017}.

In general, there are several known methods for deciding whether the state AME($d$,$N$) states exist.
These techniques include constructions involving orthogonal Latin squares~\cite{Goyeneche_2015}, shadow inequalities~\cite{Huber_2018}, bound stemming from quantum error correcting codes~\cite{Scott_2004}, usage of stabilizer states~\cite{Danielsen_2012}, entanglement witnesses~\cite{Yu_2021}, and linear programming~\cite{Calderbank_1998}.
Notwithstanding the above methods, there are still cases for which the existence question remains open.
The smallest of these are presented in Fig.~\ref{fig:ame_table}, with the special emphasis on the four quhex state AME(4,6), which will be discussed in the subsequent parts of this chapter.

\begin{figure}[H]
    \includegraphics[scale=0.58]{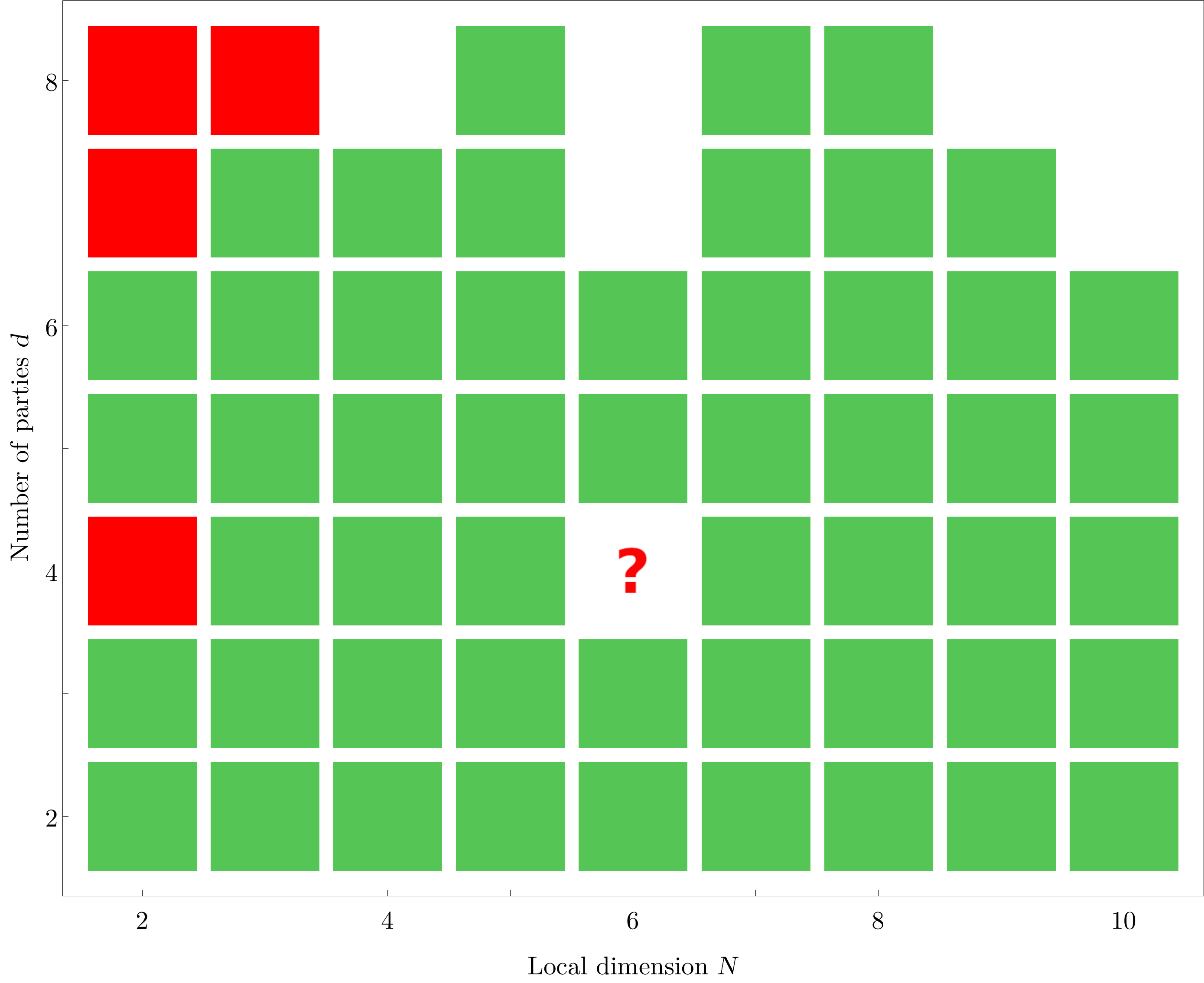}
    \caption{The existence of AME($d$,$N$) states for a small number of parties $d$ and local dimensions $N$, listed at the webpage of Huber and Wyderka~\cite{Table_AME} and known before the 2021 publication~\cite{Rajchel_AME}. 
    The green background of the square means that there exists a construction of the AME states in the respective Hilbert space, whereas the red background denotes the non-existence of the corresponding AME($d$,$N$) state.
    Blank squares depict the cases that are yet undecided.
    The question mark refers to the unknown status of the AME(4,6) state, which existence was proven in the recent joint paper~\cite{Rajchel_AME}. 
    }
    \label{fig:ame_table}
\end{figure}

\section{Orthogonal Latin squares and 36 officers of Euler}
Finding unexpected connections between seemingly distant fields of science is particularly interesting.
Such links between combinatorics and quantum information were established as combinatorial designs related to the problem of officers of Euler were applied for the search for absolutely maximally entangled states of four quhexes.
In order to appreciate the connection consider the original problem of Euler, which concerns Latin squares (see Section~\ref{sec:quantum_designs}) in dimension 6. 
Let us recall a known definition~\cite{Combinatorics_book} that lays the ground for the exact statement of the problem.

\begin{definition}[OLS]
    Two Latin squares ($A_{ij}$) and ($B_{ij}$) are called orthogonal (OLS) if, when superimposed, no pair of elements $(A_{ij},B_{ij})$ repeats.
\end{definition}

One can expand this definition for more than two mutually orthogonal Latin squares. 
However, in this thesis, we shall restrict to consideration of two OLS only.
Let us start the study of OLS by an observation that there are no two orthogonal Latin squares of size 2, exhibited by Fig.~\ref{fig:OLS_2}.

\begin{figure}[H]
    \begin{equation*}
            \large
        \begin{array}{|c|c|} \hline 
    1  & 2  \\ \hline
    2 & 1\\ \hline
    \end{array}\quad  + \quad
        \begin{array}{|c|c|} \hline 
    1  & 2  \\ \hline
    2 & 1\\ \hline
    \end{array} \quad
    =\quad
    \begin{array}{|c|c|} \hline 
    {1,1}    & 2,2  \\ \hline
    { 2,2} & 1,1\\ \hline
    \end{array} 
    \end{equation*}
    \caption{Two Latin squares of size $N=2$.
    Without loss of generality, we can assume that the first rows are the same in both Latin squares. 
    This implies that they are the same, and their superimposition shows that there are no two orthogonal Latin squares of size $N=2$.}
    \label{fig:OLS_2}
\end{figure}

Even though OLS built out of the smallest Latin squares do not exist, it is possible to construct them for all higher dimensions (apart from $N=6$), as exemplifies Fig.~\ref{fig:OLS_3} for $N=3$.

\begin{figure}[H]
    \begin{equation*} 
    \begin{array}{|c|c|c|} \hline 
1   & 2&  3  \\ \hline
2  & 3 & 1 \\  \hline
3  & 1 & 2  \\  \hline
\end{array}  
\quad + \quad \begin{array}{|c|c|c|} \hline 
1 & 2&  3  \\ \hline
3  & 1 & 2 \\  \hline
2  &3& 1  \\  \hline
\end{array}  
\quad = \quad 
\begin{array}{|c|c|c|} \hline 
1,1   & 2,2&  3,3  \\ \hline
2,3  & 3,1 &  1,2 \\  \hline
3,2 & 1,3 &  2,1 \\  \hline
\end{array}  
\quad = \quad    
\begin{array}{|c|c|c|} \hline 
{ \color{black} A \text{\spade} }    & K \text{\club} &  {\color{red}Q \text{\diamond}}   \\ \hline
{ \color{red} K\text{\diamond}}  & {\color{black} Q \text{\spade}} & A \text{\club}    \\  \hline
{ Q \text{\club}}  & {\color{red}A \text{\diamond}} &  {\color{black} K \text{\spade} }  \\  \hline
\end{array} 
    \end{equation*}
    \caption{Two orthogonal Latin squares of size 3.
    Alternative card form was also used, where the rank of the card stands for the first Latin square, while its suit denotes appropriate element from the second Latin square (in particular, 
$1\rightarrow A / \text{\spade} $, 
$2\rightarrow K / \text{\club}$, 
$3\rightarrow Q / {\color{red} \text{\diamond}}$).}
    \label{fig:OLS_3}
\end{figure}

The case of OLS of dimension 6 is different, what was already acknowledged by Euler in his notes from 1779~\cite{Euler_1779}, in which the name of 36 officers appear.
Motivation for such a name came from the problem of rearranging 36 officials of different ranks and different regiments so that no row or column contains repeating rank or regiment -- a real-world application of combinatorial designs.

Difficulties in constructing OLS of this size encouraged Euler to pose a conjecture that there are no OLS in dimensions $N \equiv 2$ (mod 4).
The smallest of these dimensions is six, and it took 121 years to prove that, indeed, there are no OLS of size $N=6$.
The proof was done in an exhaustive manner by a French lawyer and an amateur mathematician Gaston Tarry, employed by French colonial administration in Algeria\footnote{Henri Poincar\'e was so much impressed by this result that he recommended Tarry to the \emph{French Acad\'emie des Sciences}.}.
Tarry reduced the total number of more than 812 millions of combinations to 9408 cases and then showed that none of them leads to orthogonal Latin squares of size six~\cite{Tarry_1900}.
Nonetheless, the conjecture of Euler waited till Parker, Bose, and Shrikhande in 1960 proved that it does not hold for any $N>6$~\cite{Bose_1960}.
We sum up those results in the following theorem.

\begin{theorem}[Existence of OLS~\cite{Tarry_1900}]\label{theorem:OLS}
    Two orthogonal Latin squares exist in all dimensions apart from $N=2$ and $N=6$.
\end{theorem}

In the subsequent parts of the chapter, we consider in detail the link between OLS and AME states.

\section{Connection between orthogonal Latin squares and AME states}
One of the methods for constructing absolutely maximally entangled states involves using combinatorial designs.
In particular, the existence of orthogonal Latin squares of size $N$ implies the existence of an AME(4,$N$) state.

To show the connection, consider an example of OLS of size $N = 3$ presented in Fig.~\ref{fig:OLS_3}.
To its every element $(k,l)$ on the position $(i,j)$ one can associate a state in $\mathcal{H}_3^{\otimes 4}$ which reads $\ket{ijkl}$.
Note that we have chosen a slightly less common computational basis of a Hilbert space $\{\ket{1},...,\ket{N}\}$ instead of the usual one starting with $\ket{0}$.
By a summation of all $N^2 = 9$ elements (shown in the central panel in Fig.~\ref{fig:OLS_3}) one obtains an unnormalized state
\begin{equation}
\begin{split}
    \ket{\psi} &= \ket{1111} + \ket{1222} + \ket{1333} + \ket{2123} + 
    \ket{2231} + \ket{2312} + \ket{3132} + \ket{3213} + \ket{3321}.
\end{split}
\end{equation}

Following Helwig et al.~\cite{Helwig_2012} we observe that this state is AME.
In order to prove it, consider the second condition for a state to be AME (Def.~\ref{def:AME_state}), requiring that any two-systems reduced density matrix is maximally mixed.
Denoting $A$ and $B$ as subsystems corresponding to the index of an element of OLS and $C$ and $D$ as its content, observe that tracing over $A$ and $B$ gives a sum over all possible 9 pairs $\ket{k,l}\bra{k,l}$ of $\mathcal{H}_3\otimes \mathcal{H}_3$, thus a maximally mixed state.
Analogous reasoning holds for tracing over $C$ and $D$.

What remains is to verify the partial trace of systems with respect to subsystems associated to index and value.
Without loss of generality, consider the averaging over $A$ (row index) and $C$ (first value).
Since in each row every symbol from the set $\{1,2,3\}$ occurs exactly once on the first and the second position, then partial trace leaves exactly 9 pairs $\ket{k,l}\bra{k,l}$, making it maximally mixed.
The above argument holds similarly for any dimension, which proves the following lemma.

\begin{lemma}\label{lemma:OLS_AME}
    Two orthogonal Latin squares $(A_{ij},B_{ij})$ of size $N$ lead to an AME(4,$N$) state defined as $\psi = \frac{1}{N}\sum_{i,j}\ket{i,j,A_{ij},B_{ij}}$.     
\end{lemma}

Based on Theorem~\ref{theorem:OLS}, we conclude that there exists an AME(4,$N$) state for any $N$ besides 2 and 6.
These two cases are not excluded by the connection between OLS and AME states.
Nonetheless, the non-existence of OLS of sizes 2 and 6 means that searching for instances of AME(4,2) state and AME(4,6) state should be carried out by different means.
As was already mentioned, Higuchi and Sudbery found out~\cite{Higuchi_2000} that there are no AME(4,2) states, which means that the only case left is AME(4,6) state.

\section{Orthogonal quantum Latin squares}
Let us start exploring the path which led to the success in search of AME(4,6) state by introducing the notion of \emph{orthogonal quantum Latin squares} (OQLS) in correspondence to their classical counterparts.
This was first done by Goyeneche et al.~\cite{Goyeneche_2018} and then slightly modified by Rico~\cite{Rico_2020,Rico_2021}.

\begin{definition}[OQLS]\label{def:OQLS}
    A $N\times N$ arrangement of pure quantum states $\{\ket{\psi_{ij}}\}$ from $\mathcal{H}_N\otimes \mathcal{H}_N = \mathcal{H}^{A}\otimes \mathcal{H}^{B}$ is said to form a pair of orthogonal quantum Latin squares if it meets the following conditions.
     \begin{enumerate}
        \item All its elements form an orthonormal basis $\{\ket{\psi_{11}},...,\ket{\psi_{NN}}\}$.
        \item Trace over any of the two subsystems while summing over two different rows is equal to zero,
        \begin{equation}
            \text{Tr}_{A(B)} \sum_{i} \ket{\psi_{ji}}\bra{\psi_{ki}} = \delta_{jk} \mathbb{I}.
        \end{equation}
        \item Trace over any of the two subsystems while summing over two different columns is equal to zero,
        \begin{equation}
            \text{Tr}_{A(B)} \sum_{i} \ket{\psi_{ij}}\bra{\psi_{ik}} = \delta_{jk} \mathbb{I}.
        \end{equation}
    \end{enumerate}
\end{definition}

The first condition of Def.~\ref{def:OQLS} can be straightforwardly related to the analogous condition for classical orthogonal Latin squares, namely that all of the possible pairs are present inside the array.
This can be seen clearly in the case where all the states inside OQLS are elements of the computational basis.
Then, the arrangement differs from the classical one only by using the Dirac notation while referring to pair of numbers.
For this reason, OQLS are generalizations of OLS -- every orthogonal Latin square directly leads to an orthogonal quantum Latin square.

On the other hand, the connection to quantum Latin squares (see Section~\ref{sec:quantum_designs}) is not so obvious, since only in the case of separable elements of OQLS, $\ket{\psi_{ij}} = \ket{\kappa_{ij}}\otimes \ket{\eta_{ij}} \in \mathcal{H}_N\otimes \mathcal{H}_N$, one can derive two underlying quantum Latin squares $\{\ket{\kappa_{ij}}\}$ and $\{\ket{\eta_{ij}}\}$.
In analogy to the connection presented by Lemma~\ref{lemma:OLS_AME}, the existence of OQLS of size $N$ implies the existence of AME(4,$N$) state.

\begin{lemma}\label{lemma:OQLS_AME}
    A pair of orthogonal quantum Latin squares $\{\ket{\psi_{ij}}\}$ of size $N$ leads to an AME(4,$N$) state defined as $\ket{\psi} = \frac{1}{N}\sum_{ij} \ket{ij}\ket{\psi_{ij}}$.
    Conversely, any AME(4,$N$) state leads to a pair of OQLS of size $N$.
\end{lemma}
\begin{proof}
    In the proof, we shall be following the line of thought presented by Rico~\cite{Rico_2021}.
    Let us start by proving that the AME(4,$N$) state $\ket{\psi}$ leads to a pair of OQLS of size $N$.    
    Every such state can be written as
    \begin{equation}
        \ket{\psi} = \frac{1}{N}\sum_{ij}\ket{i}_A \ket{j}_B \ket{\psi_{ij}}_{CD},
    \end{equation}
    which stems from the fact that an AME state is maximally entangled with respect to any symmetric bipartition.
    Therefore, it is locally unitarily equivalent to a generalized Bell state.
    
    We shall construct the pair of OQLS treating $\{i,j\}$ as indices and $\ket{\psi_{ij}}$ as elements in the appropriate places.
    First, condition (1) from Definition~\ref{def:OQLS} is satisfied because of the maximal entanglement of this bipartition.
    What remains to prove are conditions (2) and (3).
    
    Due to properties of AME states $\rho_{AC} = \text{Tr}_{BD}\big(\ket{\psi}\bra{\psi}\big) = \frac{1}{N^2}\mathbb{I}_{AC}$, so
    \begin{equation}\label{eq:AME_OQLS}
    \begin{split}
        \frac{1}{N^2}\mathbb{I}_{AC} &= \text{Tr}_{BD} \big(\ket{\psi}\bra{\psi}\big) = \frac{1}{N^2} \text{Tr}_{BD} \big( \sum_{ijkl} \ket{ij}\bra{kl} \otimes \ket{\psi_{ij}}\bra{\psi_{kl}}  \big) = \frac{1}{N^2} \sum_{ik} \ket{i}\bra{k} \otimes \text{Tr}_D \big( \sum_{j}  \ket{\psi_{ij}}\bra{\psi_{kj}}  \big) = \\
        &= \frac{1}{N^2}\Big( \mathbb{I}_{A} \otimes \text{Tr}_D \big( \sum_{j}  \ket{\psi_{ij}}\bra{\psi_{ij}}  \big) + \sum_{i\neq k} \ket{i}\bra{k} \otimes \text{Tr}_D \big( \sum_{j}  \ket{\psi_{ij}}\bra{\psi_{kj}}  \big) \Big).
    \end{split}
    \end{equation}
    
    The left-hand side does not admit off-diagonal non-zero elements.
    Therefore, since the second part of the sum concerns $\ket{i}\bra{k}$ for $i\neq k$, the terms $\text{Tr}_D \big( \sum_{j}  \ket{\psi_{ij}}\bra{\psi_{kj}}  \big)$ must all vanish.  
    Simultaneously, $\text{Tr}_D \big( \sum_{j}  \ket{\psi_{ij}}\bra{\psi_{ij}}  \big) = \mathbb{I}_{C}$, which proves condition (2).
    Analogous reasoning verifies condition (3) using the reduced state $\rho_{AD}$.
    This concludes the proof that an AME(4,$N$) state leads to a pair of OQLS.
    Nonetheless, the reasoning exhibited by Eq.~(\ref{eq:AME_OQLS}) is also valid in the other direction.
    Together with the fact that OQLS form a basis, this proves that OQLS can be used to construct AME(4,$N$) state; thereby, proving the theorem.
\end{proof}

Lemma~\ref{lemma:OQLS_AME} motivates the search for the AME(4,6) state from the perspective of quantum designs since these two notions turn out to be equivalent.
In the next section, we shall focus on the starting point provided by an approximate pair of OLS of size 6, enabling us to progress in the quantum domain and, at last, producing the desired absolutely maximally entangled state of four quhexes.

\section{Almost orthogonal pair of classical Latin squares of size 6}
In the previous sections, we argued following Euler that there are no orthogonal Latin squares of size 6.
Nevertheless, one is tempted to find, \emph{how close} to orthogonality two Latin squares of this size can be arranged.
The answer to this puzzle is known~\cite{Hill_1986}, and one of the many equivalent arrangements, designated $P_{36}$, is presented in Fig.~\ref{fig:almost_OLS_6}.

\begin{figure}[H]
\centering
	\begin{tikzpicture}
			\node (a) at (-3,0) {$P_{36}$ =
\setlength{\tabcolsep}{3pt}
\renewcommand{\arraystretch}{1.}
\begin{tabular}{|llllll|}
\hline
$11$ & $22$ & {\color{red} $33$}  & {\color{red} $44$} & $55$ & $66$ \\ 
$23$ & $14$ & $45$ & $36$ & $61$ & $52$ \\ 
$32$ & $41$ & $64$ & $53$ & $16$ & $25$ \\ 
$46$ & $35$ & $51$ & $62$ & $24$ & $13$ \\ 
$54$ & $63$ & $26$ & $15$ & $42$ & $31$ \\ 
$65$ & $56$ & $12$ & $21$ & {\color{red} $33$} & {\color{red} $44$} \\
\hline
\end{tabular}
 };
        \node (a) at (-0.1,0) {$=$};
        \node (a) at (3.8,0) {\includegraphics[scale=0.95]{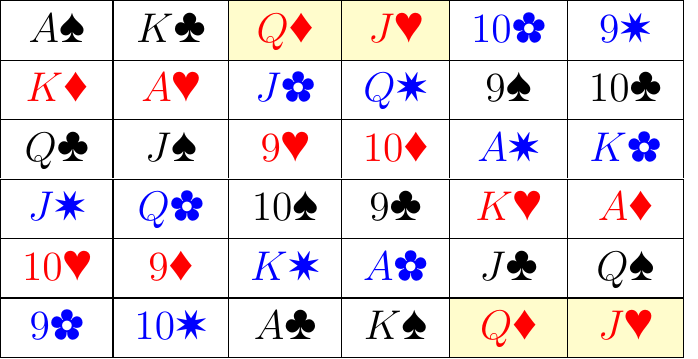}};    
    \end{tikzpicture}
    \caption{Two \emph{almost} orthogonal Latin squares of size 6, represented by a permutation matrix $P_{36}$ of order $N^2=36$, reproduced from the joint paper~\cite{Rajchel_AME}.
    The rank of the card stands for the first Latin square, while its suit denotes appropriate element from the second Latin square (in particular, 
$1\rightarrow A / \text{\spade} $, 
$2\rightarrow K / \text{\club}$, 
$3\rightarrow Q / {\color{red} \text{\diamond}}$,
$4\rightarrow J / {\color{red} \text{\heart} }$,
$5\rightarrow 10 / {\color{blue} \text{\flower} }$, 
$6\rightarrow 9 / {\color{blue}  \text{\star}}$).
    This is only an approximation to the OLS of size 6, as there are two pairs of symbols that repeat twice (highlighted by a yellow background), not fulfilling the condition that no pairs can repeat.
    Observe that two cards from the deck of 36, queen of hearts {\color{red}Q\heart} and jack of diamonds {\color{red}J\diamond}, are missing in this pattern.
    }
    \label{fig:almost_OLS_6}
\end{figure}

In order to quantify the degree of orthogonality of a given pair of Latin squares, we shall use the notion of entangling power, introduced in Section~\ref{sec:gates_ent_power} and discussed in detail in the context of multipartite gates in Chapter~\ref{chapter_3}.
By using this metric we find that the matrix $P_{36}$ leads to the solution as close to OLS of size 6 as possible, since $P_{36}$ possesses the highest entangling power attainable by any such arrangement of size six~\cite{Clarisse_2005}.
Using the normalization introduced in Section~\ref{sec:gates_ent_power}, with maximal entangling power $e_p = 1$, it is easy to compute that $e_p(P_{36}) = \frac{314}{315} \approx 0.9968$.
Hence, only slightly more than three permilles of the normalized entangling power is missing to the desired AME state for which, by construction, $e_p = 1$.

\section{Alternative description of QLS using matrices}\label{sec:matrix_description_AME}
Since using the classical form of OLS will not yield any better solution, we need to expand the search for AME(4,6) state to the domain of quantum Latin squares.
By using this method, every element of a QLS is a pure quantum state from $\mathcal{H}_{6}\otimes \mathcal{H}_{6}$; therefore, can be written as a $6\times 6$ matrix.
The injection can be visualized easily for the computational basis, namely $\ket{i}\ket{j} \rightarrow$ 1 in the $i$-th row and $j$-th column.

As an example, consider the element in the second row and first column of the matrix $P_{36}$, which reads (2,3).
Its transformation to a matrix gives
\begin{equation}\label{eq:example_OLS_to_QLS}
    (2,3) \rightarrow
    \begin{pmatrix}
    0 & 0 & 0 & 0 & 0 & 0\\ 
    0 & 0 & 1 & 0 & 0 & 0\\ 
    0 & 0 & 0 & 0 & 0 & 0\\ 
    0 & 0 & 0 & 0 & 0 & 0\\ 
    0 & 0 & 0 & 0 & 0 & 0\\ 
    0 & 0 & 0 & 0 & 0 & 0\\ 
    \end{pmatrix}.
\end{equation}

Ultimately, in the language of QLS the matrix $P_{36}$ can be written as a permutation matrix of order 36 by transforming each of its elements into one of the $6\times 6$ blocks, using the procedure exemplified by Eq.~(\ref{eq:example_OLS_to_QLS}).
The reverse procedure can also be applied, i.e.\ every matrix of size 36 can be converted to an array of size 6, composed of quantum states (possibly unnormalized), by dividing it into 36 blocks.
Slightly abusing the notation, we shall refer to QLS/matrices by the same name, e.g.\ $P_{36}$ might be understood as an array presented in Fig.~\ref{fig:almost_OLS_6} or a unitary matrix of order 36, depending on the context.
In Table~\ref{tab:P_36} we present $P_{36}$ converted to a form of a permutation matrix, $P_{36}\in U(36)$.

\begingroup
\arraycolsep=1.4pt\def\arraystretch{2.2}
\renewcommand{\arraystretch}{1}
\definecolor{green(pigment)}{rgb}{0.75, 1.0, 0.750}
\newcommand\y{\colorbox{green(pigment)}{$1$}}
\newcommand\redd{\colorbox{red}{$1$}}
\begin{table}[H]
\begin{center}
$P_{36}=
\tiny{
\left[\begin{array}{rrrrrr|rrrrrr|rrrrrr|rrrrrr|rrrrrr|rrrrrr}
      \y &.&.&.&.&.&.&.&.&.&.&.&.&.&.&.&.&.&.&.&.&.&.&.&.&.&.&.&.&.&.&.&.&.&.&.\\
      .&.&.&.&.&.&.&\y&.&.&.&.&.&.&.&.&.&.&.&.&.&.&.&.&.&.&.&.&.&.&.&.&.&.&.&.\\
      .&.&.&.&.&.&.&.&.&.&.&.&.&.&\redd&.&.&.&.&.&.&.&.&.&.&.&.&.&.&.&.&.&.&.&.&.\\
      .&.&.&.&.&.&.&.&.&.&.&.&.&.&.&.&.&.&.&.&.&\redd&.&.&.&.&.&.&.&.&.&.&.&.&.&.\\
      .&.&.&.&.&.&.&.&.&.&.&.&.&.&.&.&.&.&.&.&.&.&.&.&.&.&.&.&\y&.&.&.&.&.&.&.\\
      .&.&.&.&.&.&.&.&.&.&.&.&.&.&.&.&.&.&.&.&.&.&.&.&.&.&.&.&.&.&.&.&.&.&.&\y\\
\hline
      .&.&.&.&.&.&.&.&.&\y&.&.&.&.&.&.&.&.&.&.&.&.&.&.&.&.&.&.&.&.&.&.&.&.&.&.\\
      .&.&\y&.&.&.&.&.&.&.&.&.&.&.&.&.&.&.&.&.&.&.&.&.&.&.&.&.&.&.&.&.&.&.&.&.\\
      .&.&.&.&.&.&.&.&.&.&.&.&.&.&.&.&.&.&.&.&.&.&.&\y&.&.&.&.&.&.&.&.&.&.&.&.\\
      .&.&.&.&.&.&.&.&.&.&.&.&.&.&.&.&\y&.&.&.&.&.&.&.&.&.&.&.&.&.&.&.&.&.&.&.\\
      .&.&.&.&.&.&.&.&.&.&.&.&.&.&.&.&.&.&.&.&.&.&.&.&.&.&.&.&.&.&.&\y&.&.&.&.\\
      .&.&.&.&.&.&.&.&.&.&.&.&.&.&.&.&.&.&.&.&.&.&.&.&\y&.&.&.&.&.&.&.&.&.&.&.\\
\hline
      .&.&.&.&.&.&.&.&.&.&.&.&.&.&.&.&.&.&.&.&.&.&.&.&.&.&.&.&.&\y&.&.&.&.&.&.\\
      .&.&.&.&.&.&.&.&.&.&.&.&.&.&.&.&.&.&.&.&.&.&.&.&.&.&.&.&.&.&.&.&.&.&\y&.\\
      .&\y&.&.&.&.&.&.&.&.&.&.&.&.&.&.&.&.&.&.&.&.&.&.&.&.&.&.&.&.&.&.&.&.&.&.\\
      .&.&.&.&.&.&\y&.&.&.&.&.&.&.&.&.&.&.&.&.&.&.&.&.&.&.&.&.&.&.&.&.&.&.&.&.\\
      .&.&.&.&.&.&.&.&.&.&.&.&.&.&.&.&.&.&.&.&\y&.&.&.&.&.&.&.&.&.&.&.&.&.&.&.\\
      .&.&.&.&.&.&.&.&.&.&.&.&.&.&.&\y&.&.&.&.&.&.&.&.&.&.&.&.&.&.&.&.&.&.&.&.\\
\hline
      .&.&.&.&.&.&.&.&.&.&.&.&.&.&.&.&.&.&.&.&.&.&.&.&.&.&.&.&.&.&.&.&\y&.&.&.\\
      .&.&.&.&.&.&.&.&.&.&.&.&.&.&.&.&.&.&.&.&.&.&.&.&.&.&.&\y&.&.&.&.&.&.&.&.\\
      .&.&.&.&.&.&.&.&.&.&\y&.&.&.&.&.&.&.&.&.&.&.&.&.&.&.&.&.&.&.&.&.&.&.&.&.\\
      .&.&.&.&.&\y&.&.&.&.&.&.&.&.&.&.&.&.&.&.&.&.&.&.&.&.&.&.&.&.&.&.&.&.&.&.\\
      .&.&.&.&.&.&.&.&.&.&.&.&\y&.&.&.&.&.&.&.&.&.&.&.&.&.&.&.&.&.&.&.&.&.&.&.\\
      .&.&.&.&.&.&.&.&.&.&.&.&.&.&.&.&.&.&.&\y&.&.&.&.&.&.&.&.&.&.&.&.&.&.&.&.\\
\hline
      .&.&.&.&.&.&.&.&.&.&.&.&.&.&.&.&.&.&.&.&.&.&\y&.&.&.&.&.&.&.&.&.&.&.&.&.\\
      .&.&.&.&.&.&.&.&.&.&.&.&.&.&.&.&.&\y&.&.&.&.&.&.&.&.&.&.&.&.&.&.&.&.&.&.\\
      .&.&.&.&.&.&.&.&.&.&.&.&.&.&.&.&.&.&.&.&.&.&.&.&.&.&.&.&.&.&\y&.&.&.&.&.\\
      .&.&.&.&.&.&.&.&.&.&.&.&.&.&.&.&.&.&.&.&.&.&.&.&.&\y&.&.&.&.&.&.&.&.&.&.\\
      .&.&.&\y&.&.&.&.&.&.&.&.&.&.&.&.&.&.&.&.&.&.&.&.&.&.&.&.&.&.&.&.&.&.&.&.\\
      .&.&.&.&.&.&.&.&\y&.&.&.&.&.&.&.&.&.&.&.&.&.&.&.&.&.&.&.&.&.&.&.&.&.&.&.\\
\hline
      .&.&.&.&.&.&.&.&.&.&.&.&.&\y&.&.&.&.&.&.&.&.&.&.&.&.&.&.&.&.&.&.&.&.&.&.\\
      .&.&.&.&.&.&.&.&.&.&.&.&.&.&.&.&.&.&\y&.&.&.&.&.&.&.&.&.&.&.&.&.&.&.&.&.\\
      .&.&.&.&.&.&.&.&.&.&.&.&.&.&.&.&.&.&.&.&.&.&.&.&.&.&\redd&.&.&.&.&.&.&.&.&.\\
      .&.&.&.&.&.&.&.&.&.&.&.&.&.&.&.&.&.&.&.&.&.&.&.&.&.&.&.&.&.&.&.&.&\redd&.&.\\
      .&.&.&.&.&.&.&.&.&.&.&\y&.&.&.&.&.&.&.&.&.&.&.&.&.&.&.&.&.&.&.&.&.&.&.&.\\
      .&.&.&.&\y&.&.&.&.&.&.&.&.&.&.&.&.&.&.&.&.&.&.&.&.&.&.&.&.&.&.&.&.&.&.&.\\
\end{array}\right]
}$
\caption{Permutation $P_{36}$ matrix with the highest entangling power, $e_p(P_{36})=\frac{314}{315} \approx 0.9968$. Unit elements are highlighted in green and red, all the rest are equal to 0.
This matrix \emph{almost} satisfies the location condition of strong Sudoku -- in every block, elements are in different positions, apart from the two red pairs that spoil the orthogonality.
Sudoku designs are discussed in Chapter~\ref{chapter_7}.}\label{tab:P_36}
\end{center}
\end{table}
\endgroup

Conditions for a double QLS of size $N$ to form an OQLS can be also written in terms of operations on matrices. For more details concerning these operations and their connection to the search for the AME(4,6) state, see Section~\ref{sec:operations_on_matrices} and~\cite{Rico_2020}.

\begin{definition}[Multiunitary matrix]
    A matrix\, $U$ of size $N^2$ is called multiunitary (or 2-unitary) if the following conditions are satisfied:
    \begin{enumerate}
        \item $U$ is a unitary matrix,
        \item reshuffled matrix $U^R$ is unitary,
        \item partially transposed matrix $U^{\Gamma}$ is unitary.
    \end{enumerate}
    A multiunitary matrix attains the maximal entangling power, $e_p(U) = 1$.
\end{definition}

Using this definition, involving transformations of partial transpose and reshuffling described in Section~\ref{sec:operations_on_matrices}, we are able to formulate and prove the subsequent lemma.

\begin{lemma}\label{lemma:multiunitary}
    Every multiunitary matrix\, $U$ of size $N^2$ defines a pair of orthogonal quantum Latin squares of size $N$.
    Moreover, any pair of OQLS defines a multiunitary matrix.
\end{lemma}
\begin{proof}
The proof will use the equivalence between AME states and OQLS, which was asserted by Lemma~\ref{lemma:OQLS_AME}.
Using the generalized Bell state $\ket{\phi^+} = \frac{1}{\sqrt{N}}\sum_i \ket{i}\otimes\ket{i}$ across the bipartition $AB$ versus $CD$, the normalized 
state defined by the multiunitary matrix on four subsystems reads
\begin{equation}
    \ket{\psi} = \frac{1}{N}\sum_{ijkl=1}^N \alpha_{ijkl} \ket{ijkl} = \big(U_{AB}\otimes \mathbb{I}_{CD}\big) \frac{1}{N}\sum_{ij=1}^{N}\ket{ij}_{AB}\otimes \ket{ij}_{CD}= \big(U_{AB}\otimes \mathbb{I}_{CD}\big) \ket{\phi^+}_{AB|CD},
\end{equation}
where $U_{AB}$ denotes matrix $U$ acting on subsystems $A$ and $B$.
Since $U = \sum_{ijkl} \alpha_{ijkl} \ket{ij}\bra{kl}$, due to the properties of the reshuffling, $U^R = \sum \alpha_{ijkl} \ket{ik}\bra{jl}$, and similarly for the partial transposition $U^\Gamma = \sum \alpha_{ijkl} \ket{il}\bra{jk}$.
Unitarity of both these matrices implies that the state $\ket{\psi}$ can be written down with respect to two different bipartitions,
\begin{equation}
    \ket{\psi} = \big(U^R_{AC}\otimes \mathbb{I}_{BD}\big) \ket{\phi^+}_{AC|BD} = \big(U^\Gamma_{AD} \otimes \mathbb{I}_{BC}\big) \ket{\phi^+}_{AD|BC}.
\end{equation}

The above equation provides the proof of maximal entanglement across all bipartitions provided unitarity of all matrices $U$, $U^R$, and $U^\Gamma$.
Therefore, we proved that a multiunitary matrix of size $N^2$ implies the existence of an AME(4,$N$) state and OQLS of size $N$.
Reasoning in the other direction holds analogously, as the property of maximal entanglement is equivalent to the unitarity of the appropriate matrices.
\end{proof}

The above lemma shows that the problem of finding AME(4,6) state can be posed not only by the existence of OQLS of size 6 but also by the existence of a multiunitary matrix of size 36.
This remark is particularly useful for numerical maximization of the entangling power since computing $e_p$ of a given unitary matrix is a relatively easy task.

By changing the perspective from two quantum Latin squares of size $N=6$ to matrices of order $N^2=36$, one obtains, by the Choi-Jamio{\l}kowski isomorphism (see Section~\ref{sec:sets_of_matrices}), a 4-partite pure state in the Hilbert space $\mathcal{H}_6^{\otimes 4}$.
Even though \emph{a priori} the entangling power of a unitary matrix does not provide insight into the entanglement properties of the said pure state, it is still a valid measure since its optimization may lead to an absolute maximum, $e_p=1$, which corresponds to the desired AME(4,6) state.

\section{Visualization of matrices}
Written in terms of matrices, $P_{36}$ is unitary and, in particular, a permutation matrix.
In fact, all matrices connected with classical design are permutation matrices, provided that no two elements in the same row or column have the same rank or suit.
This means that the first and second elements of corresponding arrangements should all be different, which is always true for Latin squares.
We have relaxed the condition for OLS which concerns distinct pairs of elements, as it cannot be met for two Latin squares of size 6.

By extending the set of permutation matrices to the unitary set one can use many more degrees of freedom since there are only 36!\ different permutation matrices, while unitary matrices must be described by continuous $36^2 = 1296$ real parameters; therefore, there are infinitely many unitary matrices.
Nonetheless, these numbers show that the search space is huge and it is unfeasible to use brute-force methods with 1296 real parameters.
Moreover, what additionally enhances the difficulty of the problem, the entangling power is a function which we conjecture to admit local maxima that are not global, as will be postulated in Section~\ref{sec:local_maxima}.

Although multidimensional, the space of unitary matrices can be visualized disregarding coordinates which are not important from our perspective.
We aim to find the matrix which has the highest entangling power or, equivalently, is as close to the multiunitary matrix as possible.
Visualizing the results on the real line by using only a single parameter does not lead to a clear picture.
Furthermore, we would like to distinguish between different matrices possessing the same entangling power.
Thus, it is beneficial to take into account a second, complementary function that takes as a variable a unitary matrix.
To this end we chose the gate typicality $g_t\in [0,1]$, introduced in Section~\ref{sec:gates_ent_power}.
In Fig.~\ref{fig:ent_power_gt}, we present a 2D picture of unitary matrices of dimension 36 in these coordinates.

\begin{figure}[H]
    \input{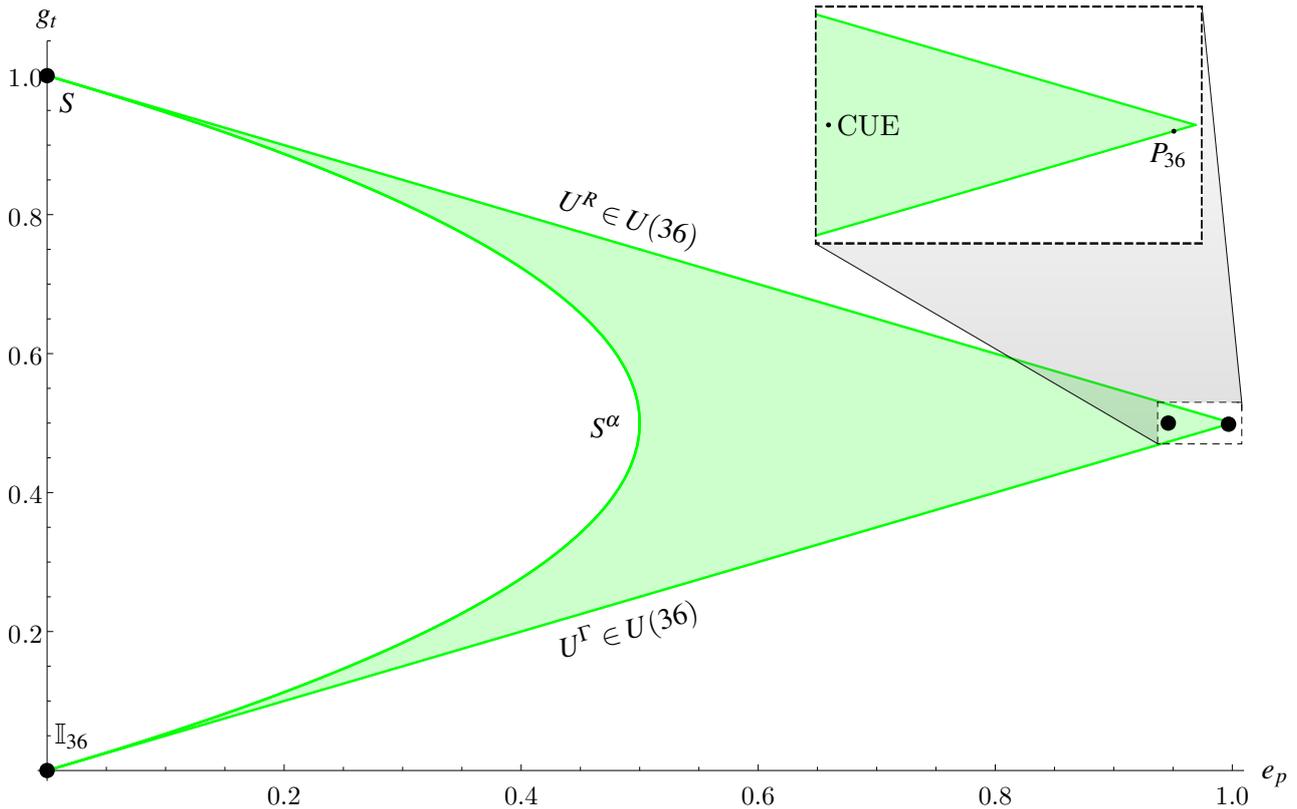}
    \caption{
    The set of unitary matrices $U(36)$ projected into the plane ($e_p$,\,$g_t$) -- entanglement power versus gate typicality.
    The parabola on the left denotes the power of the swap matrix $S^\alpha$ for $\alpha \in [0,1]$.
    On the right side of the picture, enlarged on the magnification above, the circulant unitary ensemble (CUE) point denotes the average entangling power and gate typicality attained for matrices from the circulant unitary ensemble of order 36.
    The point labeled by $P_{36}$ indicates the permutation matrix with the highest entangling power; indeed, it is barely distinguished from the right corner of the triangle -- multiunitary matrix $U_{36}$.
    The other two straight lines restricting the region correspond to matrices that are unitary and either their reshuffle or partial transpose is unitary.
    To obtain a mirror of a matrix across the line $g_t = 1/2$ it suffices to multiply a given matrix by the swap matrix $S$.
    }
    \label{fig:ent_power_gt}
\end{figure}

The extremal points on the left at $e_p=0$ are matrices (or, equivalently, gates) possessing zero entangling power.
However, zero entangling power does not mean the gate is local -- consider the swap gate $S$ in the upper left corner, already discussed in Section~\ref{sec:sets_of_matrices}.
It is maximally non-local but does not create entanglement since it only exchanges both subsystems.
The right corner of the triangle denotes the ultimate goal of the research -- the multiunitary matrix $U_{36}$, representing an absolutely maximally entangled state of four quhexes.

\section{Search for a quantum approximation of the AME state better than \texorpdfstring{$P_{36}$}{Lg}}\label{sec:A_family}
We start our search for the AME(4,6) state using the matrix description provided in Section~\ref{sec:matrix_description_AME}.
All subsequent efforts shall be devoted to the search for the multiunitary matrix $U_{36}$ that corresponds to the AME(4,6) state.
Most of unitary matrices, projected into the plane of Fig~\ref{fig:ent_power_gt}, lie very close to the point corresponding to the average entangling power of unitary gates of order 36, drawn from the circulant unitary ensemble, $\langle e_p(U)\rangle_{U\in U(36)} = \frac{35}{37} \approx 0.9459$~\cite{Zanardi_2000}.
This behavior is not atypical in the case of larger dimensions -- compare with, prevalent in high dimensional analysis, phenomenon of concentration of measure, presented in~\cite{Gross_2009,Hayden_2004}. 
Hence, the algorithm based on optimization of random matrices proved to be a blind alley, whether the set from which we draw a matrix being unitary matrices or its proper subset such as Hadamard matrices, rescaled by the square root of the dimension, $\sqrt{N}$.

Therefore, the practical point to start the search for a multiunitary matrix is the permutation matrix $P_{36}$ exhibiting the maximal entangling power, presented in Fig.~\ref{fig:almost_OLS_6}.
Before moving to the explanation of further advancement, let us focus on the study of this matrix -- in particular on its characteristics leading to non-multiunitarity.
We conclude from Lemma~\ref{lemma:multiunitary} that at least one of the matrices $P_{36}^R$ and $P_{36}^\Gamma$ is not unitary.
It can be easily verified that $P_{36}^\Gamma$ is a unitary matrix and that $P_{36}^R$ is not,
as it contains two repeating pairs of blocks in $P_{36}$.
It is not possible to alter those two pairs while using the regime of classical Latin squares; however, it is viable by allowing for quantum states as the elements of the array.
The following one-parameter family of double quantum Latin squares (equivalently matrices from $U(36)$) was devised by Arul Lakshminarayan

\begin{equation}
    A(x) = \begin{pmatrix}
    \ket{11} & \ket{22} & c_x \ket{33}-s_x \ket{43} &s_x\ket{34}+c_s\ket{44}&\ket{55}& \ket{66} \\
    \ket{23} & \ket{14} & \ket{45} & \ket{36} & \ket{61} & \ket{52} \\
    \ket{32} & \ket{41} & \ket{64} & \ket{53} & \ket{16} &\ket{25} \\
    \ket{46} & \ket{35} & \ket{51} & \ket{62} & \ket{24} & \ket{13} \\
    \ket{54} & \ket{63} & \ket{26} & \ket{15} & \ket{42} & \ket{31} \\
    \ket{65} & \ket{56} & \ket{12} & \ket{21} & c_x \ket{33}+s_x \ket{43} & -s_x \ket{34}+c_x \ket{44}
    \end{pmatrix},
\end{equation}
where $c_x$, $s_x$ denote $\cos(x)$ and $\sin(x)$ respectively.  
Note that the above family of double QLS simplifies to two classical Latin squares $P_{36}$ from Fig.~\ref{fig:almost_OLS_6} for the value of parameter $x=0$.
The value of the parameter for which the entangling power of $A(x)$ is optimal reads $x = \pi/4$, and $e_{p}\big(A(\pi/4)\big) \approx 0.9976$.
From this time on we shall refer to this particular matrix as simply $A\coloneqq A(\pi/4)$.
This result is the first observation yielding a positive answer to the question posed in~\cite{Clarisse_2005}, whether there exist unitary gates of size 36 with a higher entangling power than the best classical case, $P_{36}$.

\section{Addition of a new rotation}\label{sec:introducing_G}
After the discovery of the matrix $A$, it was found that increasing the number of rotated elements using cosine and sine functions improves the entangling power.
Introducing angle, as in the $A$ family, enhances orthogonality relations between the appropriate elements but changes other orthogonality relations as well.
Including the third rotation and allowing for unequal angles yields the $G$ family,
\begin{equation}\label{eq:G_family}
    G(x,y,z) = \begin{pmatrix}
    \ket{11} & \ket{22} & c_x \ket{33}-s_x \ket{43} &s_x\ket{34}+c_s\ket{44}&\ket{55}& \ket{66} \\
    \ket{23} & \ket{14} & c_y\ket{35}+s_y\ket{45} & s_y\ket{36}-c_y\ket{46} & \ket{61} & \ket{52} \\
    \ket{32} & \ket{41} & \ket{64} & \ket{53} & \ket{16} &\ket{25} \\
    \ket{46} & \ket{35} & \ket{51} & \ket{62} & \ket{24} & \ket{13} \\
    \ket{54} & \ket{63} & \ket{26} & \ket{15} & \ket{42} & \ket{31} \\
    \ket{65} & \ket{56} & \ket{12} & \ket{21} & c_z\ket{33}+s_z \ket{43} & -s_z \ket{34}+c_z \ket{44}
    \end{pmatrix},
\end{equation}
with $c$ and $s$ subscripted meaning cosine and sine with appropriate parameters.
Obviously, this family reduces to the previous family $A(x)$ for a suitable choice of parameters, $A(x) = G(x,0,x)$.
Nevertheless, by using different values of parameters the $G$ family outperforms the $A$ family, with the optimal values of parameters given by the matrix $G \coloneqq G( \pi/4,  3\pi/8, \pi/8)$.

\section{Family with five rotations}\label{sec:introducing_W}
A 5-parameter family of unitary matrices was proposed by Wojciech Bruzda, and was later given the name $W$ family.
We present this family in the form of double QLS:
\begin{equation}
    \begin{split}
    &W(x,y,z,u,w) = \\
    &\begin{pmatrix}
    \ket{11} & \ket{22} & c_x \ket{33}-s_x \ket{43} &s_x\ket{34}+c_s\ket{44}&\ket{55}& \ket{66} \\
    \ket{23} & \ket{14} & c_y\ket{35}+s_y\ket{45} & s_y\ket{36}-c_y\ket{46} & \ket{61} & \ket{52} \\
    \ket{32} & \ket{41} & \ket{64} & \ket{53} & \ket{16} &\ket{25} \\
    c_u\ket{36}+s_u\ket{46} & -s_u\ket{35}+c_u\ket{45} & \ket{51} & \ket{62} & \ket{24} & \ket{13} \\
    \ket{54} & \ket{63} & \ket{26} & \ket{15} & c_w\ket{32}-s_w\ket{42} & c_w\ket{31}+s_w\ket{41} \\
    \ket{65} & \ket{56} & \ket{12} & \ket{21} & c_z\ket{33}+s_z \ket{43} & -s_z \ket{34}+c_z \ket{44}
    \end{pmatrix},
    \end{split}
\end{equation}
where the convention regarding cosine and sine functions is maintained.
This family generalizes the previous $G$ family, as $G(x,y,z)=W(x,y,z,0,0)$.
The optimal entangling power achievable by this family of unitary matrices is obtained for $W \coloneqq W(-2\pi/12, -\pi/12, \pi/12,2\pi/12,3\pi/12)$.
The entangling power of this matrix reads $e_p(W)  = \frac{208 + \sqrt{3}}{210}\approx 0.9987$.
Thus, only a bit more than a permille of the normalized entangling power separates it from the desired multiunitary matrix $U_{36}$, for which $e_p(U_{36})=1$.

\section{Research into the \texorpdfstring{$W$}{Lg} family }\label{sec:checking_W_family}
To visualize the status of the search for the multiunitary matrix corresponding to AME(4,6) state we present the selected matrices from respective families in the $e_p$/$g_t$ plot, as depicted in Fig.~\ref{fig:corner_W_G_A}, which forms a magnification of the right fragment of Fig.~\ref{fig:ent_power_gt}.

\begin{figure}[H]
    \input{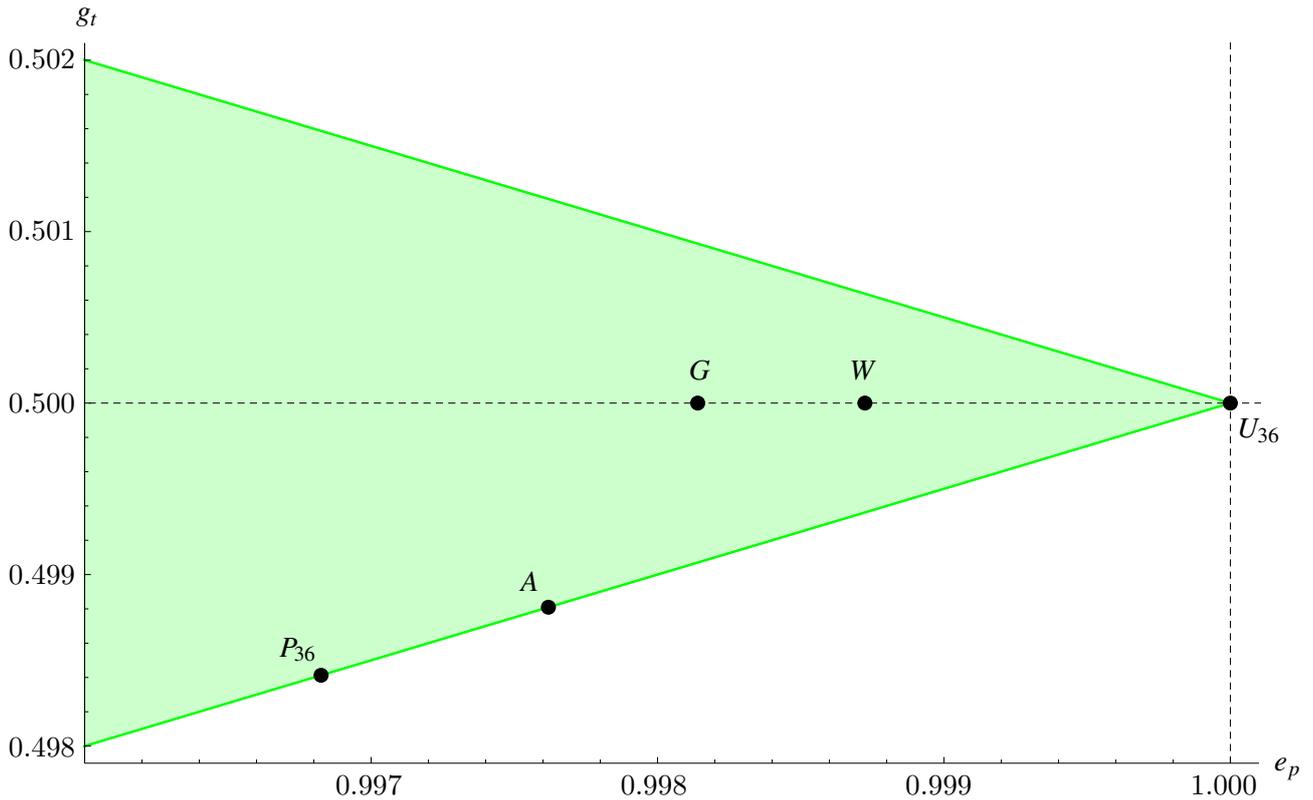}
    \caption{Projection of the set of unitary matrices of order 36 into the ($e_p$,\,$g_t$) plane near the goal -- the multiunitary matrix $U_{36}$ with coordinates $(e_p,\, g_t) = (1,1/2)$, corresponding to the state AME(4,6).
    Best permutation matrix $P_{36}$ was a starting point to the search, with $A$, $G$, and $W$ denoting representatives of appropriate families with the largest entangling power.
    }
    \label{fig:corner_W_G_A}
\end{figure}

One can try to generalize the reasoning which led to discoveries of all three families of matrices that do not yield a multiunitary matrix.
The problem we need to overcome is that even though a given matrix $U$ is unitary, its reshuffled $U^R$ or partially transposed form $U^\Gamma$ is not.
This stems from the fact that rearranging elements of matrices, done by reshuffling or partial transposition, breaks orthogonality conditions between rows/columns.
To prevent the violation of orthogonality, we can introduce new rotations, thereby fixing the problems. 
However, this comes at the price of breaking other orthogonality conditions between blocks, called strong Sudoku.

It is possible to obtain a matrix with intermediate properties -- with fewer orthogonality conditions satisfied, but with a higher entangling power like $A$, $G$, or $W$.
This results from the addition of entropies inside the entangling power, see Eq.~(\ref{eq:e_p_by_linear_entropy}), which provides a possibility of acquiring a larger sum by the inclusion of two almost maximal values instead of one maximal and the other one smaller.
Therefore, the general rationale behind the introduction of families is the reducing the non-orthogonality between certain rows/columns.
What follows naturally is the question of whether this line of thought might be applied more than five times as in the $W$ family.

In this section, we will try to explain the results starting from this angle.
First, let us consider the number of different rotations.
Since there are 36 rows of a unitary matrix, there are $\frac{36\times 35}{2} = 630$ possible two-dimensional rotations.
For the time being, all the families ($A$, $G$, and $W$) included rotations chosen in such a way that they fix the non-orthogonal relations of certain rows.
Similar reasoning cannot be applied further, since the reshuffled matrix $W^R$ does not provide additional paths for new rotations -- the supposed sixth rotation should exist between the rows which are already rotated in this matrix.
Therefore, there is no space for further direct optimization, and larger families do not improve the entangling power.

Hence, it is necessary to consider another direction, namely that rotations are applied without direct relation to the disturbed orthogonality conditions.
In order to clarify the optimization procedure, we shall refer to maximization as a numerical as well as an analytical approach evaluated using the computer.
It is feasible to apply this method for a system described by less than 50-100 parameters; therefore, a careful choice of rotations over which we optimize is crucial.

In all families of matrices, rotations acted on the neighboring rows, allowing to obtain the highest entangling power.
Thus, it is reasonable to proceed further along similar lines, rendering the maximal number of 2-dimensional rotations equal to 18 pairs ($1 \leftrightarrow 2,...,35 \leftrightarrow 36$).
Nonetheless, this family of matrices parametrized by 18 real parameters (angles of rotations) does not yield any higher value of the entangling power than $W$. 
Similarly, if we translate the pairs between which we make the link ($2 \leftrightarrow 3,...,36 \leftrightarrow 1$), the new family does not provide any member closer to the multiunitary matrix than $W$.

However, this still does not prove the extremality of the $W$ matrix (or any of the matrices that share the same entangling power).
It is hardly feasible to check every combination of 2-dimensional rotations performed on rows due to the abundance of possibilities.
To strengthen our conjecture concerning the local extremality of this matrix, we conducted a more comprehensive study using $\frac{36\times 35}{2}=630$ directions.
As a result, we obtained a full characterization given by the analytical form of $e_p$.
The one-parameter family of matrices that we studied reads

\begin{equation}
    V_\alpha = W R_\alpha,
\end{equation}
with $R_\alpha$ given by a matrix which acts as a rotation matrix on two rows/columns while having trivial action on the rest, $R_\alpha = P\Bigg[\begin{pmatrix}
\cos \alpha & - \sin \alpha \\
\sin \alpha & \cos \alpha 
\end{pmatrix}\oplus \mathbb{I}_{34}\Bigg] P^T$, where $P$ denotes an appropriately chosen permutation matrix.

In the case of the single rotations, the entangling power depended on the angle of rotation $\alpha$ by a mixture of sine and cosine function, raised to appropriate even powers, with the general formula

\begin{equation}
\begin{gathered}
    e_p (V_\alpha) = \sum^2_{k=0}\sum^2_{m=0} a_{km}\sin^{2k} (\alpha) \cos^{2m} (\alpha),
\end{gathered}
\end{equation}
where coefficients $a_{km}$ depend on the choice of the rotation, i.e.\ the rows/columns on which the transformation acts.
The research shows that the optimal value of $e_p$ is attained in all of the 630 cases by the $W$ matrix;
therefore, we conjecture that this is a local maximum of this function.
Further arguments in favor of this hypothesis shall be presented in subsequent sections concerning the Hessian matrix.

Similar results for other matrices ($A$ and $G$) were also obtained. However, these matrices admit directions that enable to increase entangling power.
This is clear in the case of the matrix $A$, since $G$ can be obtained from it by using a single rotation only, but analogous reasoning cannot be extended to the $G$ matrix itself.
Therefore, the study of single 2-dimensional rotations proves to be nontrivial and gives an insight into the extremality of the matrices under consideration.

Having concluded the search on the simplest possible rotations, one is tempted to extend this setup into larger sets of matrices, e.g.\ all unitary matrices of dimension 2.
Since exploring 2-dimensional rotations proved to be useful, it is also alluring to inspect the set of 4- or 6-dimensional matrices, either in the form of general orthogonal matrices or the full unitary group.
Notwithstanding interesting prospects, all of the procedures mentioned above did not yield a single matrix possessing a higher entangling power than $W$. 
These methods are local in a sense that they attempt to find a better matrix than $W$ using its elements.
Furthermore, it is not feasible to verify all of the possibilities provided by these bigger sets of rotation matrices, since already 2-dimensional rotations were hard to consider from the most general case.
By a careful tuning of the setup in each respective case, always the value $e_p (W) \approx 0.9987$ was the highest achievable.
This result which prompted us to attempt a more exhaustive search using the Hessian, to which goal we shall devote subsequent sections.

\section{Hessian as a tool to study extremality of a solution}\label{sec:hessian}
The extremality of the $W$ matrix was considered in Section~\ref{sec:checking_W_family} by several trials to exceed its entangling power via more rotations or more carefully suited ones.
Inability to do so prompted us to study it from a more systematic approach.
The field of mathematical analysis comes to avail with a tool -- namely, the Hessian of a function.
Let us briefly recall this notion, starting from the intuitions given by the simplest case of analysis of single-variable functions.
In this case, if a function $f$ has a derivative at point $x_0$ equal to zero $f'(x_0) = 0$, then in order to determine whether $x_0$ is an extremum one needs to verify the value of second derivative.
Its negativity indicates a local maximum, positivity a local minimum, while zero value shows that the test is inconclusive and higher derivatives must be considered. 

Similarly, in the case of multivariate functions (with $n$ variables), instead of a scalar function given by derivative, partial derivatives form a vector of length $n$.
Signature that the point $\vec{x}$ might be a local extremum is given if all partial derivatives disappear $\nabla f(\vec{x}) = [0,...,0]$.
Analogously to the single variable functions, it is not enough to determine the behavior of the function in the neighborhood of $\vec{x}$ by using only partial derivatives.
One needs to resort to second-degree partial derivatives, which at this point form an $n\times n$ matrix

\begin{equation}\label{eq:definition_of_hessian}
    H_f = \begin{pmatrix}
    \frac{\partial^2 f}{\partial x_1^2} & \frac{\partial^2 f}{\partial x_1 \partial x_2} & ... & \frac{\partial^2 f}{\partial x_1 \partial x_n} \\ 
    \frac{\partial^2 f}{\partial x_2 \partial x_1} & \frac{\partial^2 f}{\partial x_2^2} & ... & \frac{\partial^2 f}{\partial x_2 \partial x_n} \\ 
    \vdots & \vdots & \ddots & \vdots \\
        \frac{\partial^2 f}{\partial x_n \partial x_1} & \frac{\partial^2 f}{\partial x_n \partial x_2} & ... & \frac{\partial^2 f}{\partial x_n^2} \\ 
    \end{pmatrix}.
\end{equation}

A more comprehensive study of mathematical intricacies connected with Hessian can be found in~\cite{Binmore_2007}.
From the point of view of this thesis, the most important reason for introducing Hessians is studying potential local maxima of the entangling power.
The research enabled us to assert several important facts concerning families of matrices.
Thus, it is beneficial to recall the mathematical theorem used to study local maxima of multivariate functions.

\begin{theorem}[\cite{Stewart_2004}]
    Let $f: X \rightarrow Y$ be a multivariate function with Hessian matrix $H_f$.
    If its derivative equals zero ($\nabla f = 0$) at point $\vec{x}$, then:
    \begin{enumerate}
        \item if $H_f$ is positive definite (equivalently, with only positive eigenvalues) at point $\vec{x}$, then $\vec{x}$ is a local minimum.
        \item if $H_f$ is negative definite (equivalently, with only negative eigenvalues) at point $\vec{x}$, then $\vec{x}$ is a local maximum.
        \item if $H_f$ has both positive and negative eigenvalues at point $\vec{x}$, then $\vec{x}$ is a saddle point.
    \end{enumerate}
\end{theorem}

In the cases not listed above the Hessian test is inconclusive.
Nonetheless, it can still yield a lot of insight into the behavior of a function $f$ in the neighborhood of the point $\vec{x}$.
This shall be of paramount importance for the further study of the extremality of different matrices introduced in the previous sections.
What is more, the Hessian will be a motivation for introducing a numerical algorithm that further argues for the local optimality of the value $e_p(W) \approx 0.9987$.

\section{Hessian matrix of the entangling power}\label{sec:local_maxima}
Since the entangling power $e_p$ of a unitary matrix might be thought of as a multivariate function mapping matrices to real numbers, it is possible to use the setup presented in the previous section.
However, in the case of unitary matrices, one needs to take into account their underlying structure. 
Unitary matrices do not form a linear vector space, but they do form a manifold.

In order to find directions in the neighborhood of a given unitary matrix, it is useful to use the Lie group structure of unitary matrices.
Their Lie algebra is formed by Hermitian matrices, for more details on Lie group theory and the unitary group, see Section~\ref{sec:sets_of_matrices} and~\cite{Hall_2015}.
The algebra of Hermitian matrices of dimension $N$ is a vector space over the field of real numbers.
Therefore, it is convenient to operate in one of the bases of this space, consisting of the matrices we choose to be appropriate.
Without loss of generality, we can set basis as
\begin{equation}
\begin{split}
    H_{ii} &= \ket{i}\bra{i} \qquad\quad\quad\quad\;\; \text{for} \;\; i \in \{1,...,N\}, \\
    H^+_{kl} &= \ket{k}\bra{l} + \ket{l}\bra{k}\quad\quad \text{for} \;\; k\text{ and } l \in \{1,...,N\}, \;\; k\neq l, \\
    H^-_{kl} &= i(\ket{k}\bra{l} - \ket{l}\bra{k}) \quad \text{for} \;\; k\text{ and } l \in \{1,...,N\}, \;\; k\neq l,\\
\end{split}
\end{equation}
with the proper arrangement of the elements so that the final basis of $36^2 = 1296$ elements consists of all sets $\{H^{}_{ii}, H^+_{kl}, H^-_{kl}\}$.
For simplicity, we shall refer to this basis with one index as $H_{i}$.
The forthcoming computation has been made in parallel by Arul Lakshminarayan and the author of this thesis.

The derivative of any function $f$, in particular of the one that acts on matrices, can be defined as
\begin{equation}\label{eq:definition_partial_derivative}
    \nabla_i \; f(U) = \lim_{\varepsilon_i \to 0}\frac{f(U_{\varepsilon_i})-f(U)}{\varepsilon_i},
\end{equation}
where the vector of deviations $\vec{\varepsilon} = \{\varepsilon_1,...,\varepsilon_{N^2}\}$ is set to zero apart from the non-zero value $\varepsilon_i$.
We shall focus on two terms of Eq.~(\ref{eq:e_p_using_singular_entropy}) 
\begin{equation}\label{eq:focus_of_hessian}
    X_R = \mathrm{Tr}(U^RU^{R\dagger}U^RU^{R\dagger})\quad \mathrm{and } \quad X_\Gamma = \mathrm{Tr}(U^\Gamma U^{\Gamma\dagger}U^\Gamma U^{\Gamma\dagger}),
\end{equation}
restricting to unitary matrices $U$ and unitary deviations.
Every unitary matrix can be obtained as an exponent of an appropriate Hermitian matrix $U = e^{iH}$. 
Therefore, we shall decompose small deviations $U_\varepsilon$ from a given matrix $U$ as
\begin{equation}
    \tilde{U} = \exp \bigg( i\sum^{N^2}_{j=1} \varepsilon_j H_j \bigg),
\end{equation}
and the resulting matrix after the perturbation to be $U_\varepsilon = U\tilde{U}$.
We assume that the perturbation parameters are real and small, $|\varepsilon_j| \ll 1$.
This shows that, without loss of generality, we may assume one-side multiplication, i.e.\ right to the initial matrix $U$.
Furthermore, it allows us to expand the exponent into the series,
\begin{equation}\label{eq:expension_up_to_second_order_varepsilon}
    U_\varepsilon = U\bigg(\mathbb{I} + i\sum_{j}\varepsilon_j H_j +O(\varepsilon^2)\bigg). 
\end{equation}

The vector of the first-order derivatives is obtained by subtraction of the initial matrix from the final one,
\begin{equation}\label{eq:expansion_delta_U}
    \delta U = U_\varepsilon - U = i\sum_{j}\varepsilon_j UH_j + O(\varepsilon^2).
\end{equation}
Due to the linearity of reshuffling and partial transposition we expand $(U_\varepsilon)^R = (U + \delta U)^R = U^R + \delta U^R$ and $(U_\varepsilon)^\Gamma = U^\Gamma + \delta U^\Gamma$.
Therefore, the expression concerning reshuffled $U^R_\varepsilon$ that occurs in Eq.~(\ref{eq:focus_of_hessian}), up to terms linear in $\varepsilon$, reads
\begin{equation}\label{eq:delta_U^R}
\begin{split}
         X_R = \,\,&\mathrm{Tr}(U^RU^{R\dagger}U^RU^{R\dagger}) + \mathrm{Tr}(\delta U^RU^{R\dagger}U^RU^{R\dagger})+\mathrm{Tr}(U^R\delta U^{R\dagger}U^RU^{R\dagger})\\
         &+\mathrm{Tr}(U^RU^{R\dagger}\delta U^RU^{R\dagger})+\mathrm{Tr}(U^RU^{R\dagger}U^R\delta U^{R\dagger}).
\end{split}
\end{equation}

Then, using the cyclic property of the trace, we conclude that Eq.~(\ref{eq:delta_U^R}) can be rewritten as
\begin{equation}
            X_R = \mathrm{Tr}(U^RU^{R\dagger}U^RU^{R\dagger}) + 2\;\mathrm{Tr}(\delta U^RU^{R\dagger}U^RU^{R\dagger})+2\;\mathrm{Tr}(\delta U^{R\dagger}U^RU^{R\dagger}U^R),
\end{equation}
what, given the properties of the trace under complex conjugation, transforms to
\begin{equation}\label{eq:hessian_derivative_reshuffle_simplified}
    X_R = \mathrm{Tr}(U^RU^{R\dagger}U^RU^{R\dagger}) + 4\;\mathrm{Re}\;\mathrm{Tr}(\delta U^RU^{R\dagger}U^RU^{R\dagger}).
\end{equation}
Similarly the second expression from Eq.~(\ref{eq:focus_of_hessian}) reads
\begin{equation}\label{eq:hessian_derivative_partial_transpose_simplified}
   X_\Gamma = \mathrm{Tr}(U^\Gamma U^{\Gamma \dagger}U^\Gamma U^{\Gamma\dagger}) + 4\;\mathrm{Re}\;\mathrm{Tr}(\delta U^\Gamma U^{\Gamma \dagger}U^\Gamma U^{\Gamma\dagger}).
\end{equation}


Finally, plugging Eq.~(\ref{eq:hessian_derivative_reshuffle_simplified}) and~(\ref{eq:hessian_derivative_partial_transpose_simplified}) into the definition of the partial derivative~(\ref{eq:definition_partial_derivative}) and employing the expression for the entangling power~(\ref{eq:e_p_using_singular_entropy}), we arrive at
\begin{equation}
\begin{split}
    \nabla_i \; e_p(U) = -\frac{4N^2}{N^4(N^2-1)}  \bigg(\lim_{\varepsilon_i \to 0} \frac{ \mathrm{Re}\;\mathrm{Tr}(\delta U^RU^{R\dagger}U^RU^{R\dagger})}{\varepsilon_i} +
    \lim_{\varepsilon_i \to 0} \frac{\mathrm{Re}\;\mathrm{Tr}(\delta U^\Gamma U^{\Gamma \dagger}U^\Gamma U^{\Gamma\dagger})}{\varepsilon_i}\bigg),   
\end{split}
\end{equation}
what, upon expansion of $\delta U^R$ and $\delta U^\Gamma$ via Eq.~(\ref{eq:expansion_delta_U}), simplifies to
\begin{equation}\label{eq:first_order_derivative_e_p}
    \begin{split}
        \nabla_i \; e_p(U) &= -\frac{4N^2}{N^4(N^2-1)} \bigg( \lim_{\varepsilon_i \to 0}\frac{ \mathrm{Re}\;i\varepsilon_i\;\mathrm{Tr}\big(( UH_i)^RU^{R\dagger}U^RU^{R\dagger}\big)}{\varepsilon_i} + \lim_{\varepsilon_i \to 0}
        \frac{ \mathrm{Re}\;i\varepsilon_i\;\mathrm{Tr}\big(( UH_i)^\Gamma U^{\Gamma\dagger}U^\Gamma U^{\Gamma\dagger}\big)}{\varepsilon_i}
        \bigg) \\
        &= -\frac{4N^2}{N^4(N^2-1)}\; \mathrm{Im}\; \bigg( \mathrm{Tr}\big(( UH_i)^RU^{R\dagger}U^RU^{R\dagger}\big)+ 
        \mathrm{Tr}\big(( UH_i)^\Gamma U^{\Gamma\dagger}U^\Gamma U^{\Gamma\dagger}\big) \bigg).
    \end{split}
\end{equation}
The above equation is the ultimate expression for the $i$-th component of the vector of derivatives of the entangling power.
Employing this analytical result, we conducted numerical computations of the derivatives of the entangling power for the matrices of our interest.
The only matrix with its derivative equal to zero was $W$. 
This results prompted us to check also the second-order derivative of the entangling power in the case of $W$.
The calculations of the Hessian are provided below.

To find second partial derivatives, we employ the formula for the first-order~(\ref{eq:first_order_derivative_e_p}).
Analogously to the steps taken while deriving $\nabla_i \; e_p(U)$, we concentrate on two terms involving traces.
In terms up to the linear order, we expand $Y_R = \mathrm{Tr}\big((U_\varepsilon H_i)^{R\vphantom{)}} U^{R\dagger\vphantom{)}}_\varepsilon U^{R\vphantom{)}}_\varepsilon U^{R\dagger\vphantom{)}}_\varepsilon \big)$ as
\begin{equation}
\begin{split}
       Y_R \approx \,\,&\mathrm{Tr}\big(( UH_i)^RU^{R\dagger} U^RU^{R\dagger}\big) + \mathrm{Tr}\big((\delta  UH_i)^RU^{R\dagger}U^RU^{R\dagger}\big) + \mathrm{Tr}\big(( UH_i)^R\delta U^{R\dagger}U^RU^{R\dagger}\big) \\&+\mathrm{Tr}\big(( UH_i)^RU^{R\dagger}\delta U^RU^{R\dagger}\big) +
       \mathrm{Tr}\big(( UH_i)^RU^{R\dagger}U^R\delta U^{R\dagger}\big).
\end{split}
\end{equation}
Using similar relations for the partial transposition part of the matrix $U$ of size 36 and employing the definition of the partial derivative, we arrive at the element $(i,j)$ of Hessian, with $i,j\in \{1,...,N^2\}$,
\begin{equation}\label{eq:second_order_derivative_e_p}
    \begin{split}
        \nabla_j\nabla_i \; e_p(U) = \frac{4N^2}{N^4(N^2-1)}\; \mathrm{Re} \big(B_R(U) + B_\Gamma(U)\big) ,
    \end{split}
\end{equation}
where by the notation $B_R(U)$ we mean
\begin{equation}
\begin{split}
       B_R(U) &= \mathrm{Tr}\big((UH_j H_i)^RU^{R\dagger}U^RU^{R\dagger}\big) + \mathrm{Tr}\big(( UH_i)^R(UH_j)^{R\dagger}U^RU^{R\dagger}\big) \\
       &+ \mathrm{Tr}\big(( UH_i)^RU^{R\dagger}(UH_j)^RU^{R\dagger}\big) +
       \mathrm{Tr}\big(( UH_i)^RU^{R\dagger}U^R(UH_j)^{R\dagger}\big),
\end{split}
\end{equation}
with a similar expression for $B_\Gamma (U)$.
Observe that the expression does not seem to be symmetric; however, we have verified numerically that the symmetry property in the case of second-order derivatives is satisfied.

Using the final expression for the elements of the Hessian~(\ref{eq:second_order_derivative_e_p}) we evaluated them  for each of the matrices separately, with the characterization of their eigenvalues shown in Table~\ref{tab:hessian_e_p_results}.
Concluding the results on the extremality of the analyzed unitary matrices, we observe that the only matrix with no positive eigenvalues of its Hessian was $W$.
Therefore, we surmise that $W$ is a local maximum of the entangling power.
The following section shall be devoted to the study of a separate way for obtaining the $W$ matrix from $G$.

 \begin{table}[H]
 \centering
\begin{tabular}{|c|c|c|c|c|}
        \hline
        name of the matrix & $W$ & $G$ & $A$ & $P_{36}$  \\
        \hline
        \multirow{6}{6em}{\centering Most positive eigenvalues} & \multirow{6}{5em}{\centering none} 
          & 0.332 &  & 1.32\\ 
        & & 0.259 &  & 1.32\\ 
        & & 0.0506 & 0.00964 & 0.500\\
        & & 0.0506 & 0.00964 & 0.500 \\ 
        & & 0.493 &  & 0.310 \\ 
        & & 0.493 &  & 0.310 \\ 
        \hline
        \multirow{2}{7em}{\centering Number of eigenvalues $> 0$} & \multirow{2}{6em}{\centering 0} & \multirow{2}{7em}{\centering 84} & \multirow{2}{6em}{\centering 2} & \multirow{2}{6em}{\centering 14}  \\
         & & & &\\
        \hline
        \multirow{2}{7.3em}{\centering Number of eigenvalues = 0} & \multirow{2}{6em}{\centering 157} & \multirow{2}{6em}{\centering 64} & \multirow{2}{6em}{\centering 157} & \multirow{2}{6em}{\centering 60}  \\
         & & & &\\
        \hline
        \multirow{2}{7.3em}{\centering Number of eigenvalues < 0} & \multirow{2}{6em}{\centering 1139} & \multirow{2}{6em}{\centering 1148} & \multirow{2}{6em}{\centering 1137} & \multirow{2}{6em}{\centering 1222}  \\
         & & & &\\
         \hline
        \multirow{6}{7em}{\centering Most negative eigenvalues}  
        & -4.046 & -4.094 & -4.91 & -4.28\\ 
        & -4.046 & -4.094 & -4.58 & -4.28\\ 
        & -4.044 & -4.067 & -4.58 & -4.20\\ 
        & -4.044 & -4.067 & -4.14 & -4.20\\ 
        & -4.040 & -4.053 & -4.14 & -4.18\\ 
        & -4.040 & -4.053 & -4.09 & -4.18 \\
        \hline
\end{tabular}
 \caption{Properties of the spectrum of Hessian of the entangling power in the instances of matrices $W$, $G$, $A$, and $P_{36}$ introduced in this chapter. 
 The spectrum of $H$ consists of $36^2$ real eigenvalues.
 Since $W$ was the only matrix without positive eigenvalues of $H(W)$, we observe that the Hessian test does not exclude it as a local maximum of the entangling power.
 Furthermore, we conjecture that grouping of eigenvalues into pairs is related to the fact that for every direction $U$ there is a symmetric one $US$ given by the swap matrix $S$.
 }\label{tab:hessian_e_p_results}
\end{table}

\section{Achieving matrix \texorpdfstring{$W$}{Lg} from matrix \texorpdfstring{$G$}{Lg}}\label{sec:achieving_W_from_G}
The purpose of this section is to convince the reader that the value of entangling power of $W$, being $e_p(W) \approx 0.9987$, is special in the sense of admitting several local maxima.
In order to do so, we shall demonstrate the efforts to achieve this value of entangling power from the matrix $G$ by means of the derivatives.
First, we investigated the steepest ascent, utilizing the expression for the first-order derivative (\ref{eq:first_order_derivative_e_p}).
Nonetheless, by applying this direction, the increase of $e_p$ saturates quickly and it is not possible to reach $e_p(W) \approx 0.9987$.
Therefore, the second attempt involved the Hessian of the matrix $G$, evaluated in the previous section, see Eq.~(\ref{eq:second_order_derivative_e_p}).
Then, we used its largest positive eigenvalue and the corresponding leading eigenvector as a direction for further transformation.
Optimizing over this one-parameter family we were able to achieve a matrix $G_1$ with a slightly higher entangling power.

The procedure was applied several times, yielding other new matrices $G_2$, $G_3$, $G_4$, and $G_5$.
Their values of entangling power are gradually approaching the value $e_p(W)$; however, they are not close in terms of squared Hilbert-Schmidt distance, given by $\mathrm{dist}(A,B) = \sum_{i,j}|A_{ij}-B_{ij}|^2  $.
The lack of significant advances towards $W$ in the sense of elements can be explained by the introduction of non-trivial complex phases to the matrices by the algorithm.
However, $W$ is a real matrix, thus does not admit any non-trivial phases.
To sum up, we observe that the value $e_p (W)$ is special in the sense that we were unable to surpass it using our algorithm.

 \begin{table}[H]
 \centering
\begin{tabular}{|c|c|c|c|c|c|}
\hline
 matrix & $e_p$ &$g_t$& \# positive eigenvalues & the largest eigenvalue & dist$(\text{matrix},W)$ \\ \hline
 $G$&  $0.998\mathbf{139}$& 0.500000&52 & $0.6644$ & 27.863 \\ \hline
 $G_1$& $0.998\mathbf{631}$ &0.500000& 81 & $0.1938$ & 27.643 \\ \hline
 $G_2$&  $0.998\mathbf{694}$& 0.499965&73 & $0.0864$ & 26.557 \\ \hline
  $G_3$&  $0.998\mathbf{713}$& 0.499998&73 & $0.0747$ & 27.124 \\ \hline
 $G_4$& $0.998\mathbf{720}$ & 0.500035&74 & $0.0348$ & $26.881$ \\ \hline
  $G_5$& $0.998\mathbf{722}$ & 0.500012&-- & -- & 26.310 \\ \hline
 $W$& $0.998\mathbf{723}$ &0.500000& 0 & 0 & 0 \\ \hline
\end{tabular}
 \caption{The steepest ascent algorithm connected with the Hessian introduced new unitary matrices of order 36, denoted as $G_1$, $G_2$, $G_3$, $G_4$, and $G_5$.
 These matrices gradually approach $W$ in terms of the entangling power.
 Here, we present properties of their Hessians, such as the number of positive eigenvalues and the largest one, as well as the H-S distances to the $W$ matrix.
 After five iterations the algorithm produced matrix $G_5$ with entangling power close to $e_p(W)$.
 In general, obtained matrices $G_i$ do not lie on the $g_t = 1/2$ line.
 }\label{tab:iteration_G_to_W}
\end{table}

\section{Average singular entropy}\label{sec:sum_of_entropies_AME}
The sole measure of multiunitarity that was used prior to this section was the entangling power $e_p$.
Optimization of this measure might yield results that are locally extremal.
Nonetheless, it is possible that there exist other measures which have the same optimal point, i.e.\ multiunitary matrix, but with different local maxima landscape.
In order to overcome this obstacle, we considered another measure of multiunitarity.
For a matrix $X$, in general not unitary, the \emph{average singular entropy} $s_e(X)$ is defined as
\begin{equation}\label{eq:definition_average_singular_entropy}
    s_e(X) = \frac{1}{3} \bigg(E_S(X) + E_S(X^R) + E_S(X^\Gamma)\bigg).
\end{equation}
Here $E_S(X)$ stands for the singular entropy of an operator, introduced in Section~\ref{sec:gates_ent_power}.
Due to Lemma~\ref{lemma:multiunitary}, the maximal value of the average singular entropy is achieved for a multiunitary matrix.
Even though this quantity does not have the same clear operational meaning as the entangling power, it characterizes the mean bipartite entanglement of the corresponding 4-party pure state.
Thus, the average singular entropy provides a slightly different insight into the extremality of a given matrix.
It is possible that discarding the unitarity $U\mapsto X$, of an optimal unitary matrix $U$ in the sense of entangling power $e_p(U)$, can be compensated by an increase of the average singular entropy $s_e(X)$.
As we shall see in the forthcoming parts of this section, there are some matrices with the higher average singular entropy than $W$.
Nonetheless, the matrix $W$ remains the unitary matrix with the highest $e_p$ if we restrict to 2-dimensional rotations from the optimal permutation matrix $P_{36}$.

Altogether, allowing for non-unitarity of the original matrix $X$ moves the 2-dimensional problem to the 3-dimensional domain.
There are three quantities to be optimized, $E_S(X)$, $E_S(X^R)$, and $E_S(X^\Gamma)$, see Eq.~(\ref{eq:e_p_using_singular_entropy_only_for_AME}).
All in all, this setup shows that points obtainable by a matrix must lie inside the unit cube spanned by three entropies.
The main diagonal of this, so called, \emph{entropy cube} lies between points (0,0,0) and (1,1,1).
The corner (1,1,1) refers to a matrix that has all entropies maximal; thus, it would correspond to a multiunitary matrix.
The main diagonal line refers to matrices that share the same value of singular entropies, $E_S(X) = E_S(X^R) = E_S(X^\Gamma)$; thus, making them highly symmetric matrices with respect to their properties.

The search conducted by random rotations led Wojciech Bruzda to the discovery of two other matrices.
These were found numerically as a limit of a convergence procedure, see Table~\ref{table:4_categories}.
The table shows that the absolute maximum in all these categories is currently the multiunitary matrix $U_{36}$, corresponding to the golden AME(4,6) state, described in Section~\ref{sec:analytical_AME}.
Similarly to the derivative procedure explained in the previous sections, we have verified the extremality conditions also for the function~(\ref{eq:definition_average_singular_entropy}), with the details of the calculations delegated to Appendix~\ref{app:hessian_s_e}.
Applying the resulting formulae to the above-mentioned matrices showed that in all of these cases the derivative vector is zero.
We conjecture that all of these points form local maxima of the average singular entropy.

\begin{table}[H]
\begin{tabular}{|c|c|c|c|}
\hline
\diagbox{cube parts}{closest to $U_{36}$} & \pbox{20cm}{distance to $U_{36}$ \\\centering in 1-norm} & \pbox{20cm}{name of the matrix } & property \\ \hline
&&&\\[-1em]
corner (2-unitary) & 0 & $U_{36}$ & \pbox{20cm}{corresponds to the\\\centering golden AME state}\\&&&\\[-1em] \hline
entire cube& 0.000413 & $X_i$ & $\nabla s_e = 0$ \\ \hline
face (unitary) & 0.000425 & $W$ & $\nabla s_e = 0$ \\ \hline
edge (dual unitary) & 0.000793 & $A$ & $\nabla s_e = 0$ \\ \hline
diagonal &  0.000534 & $X_d$ & $\nabla s_e = 0$ \\ \hline
\end{tabular}
\caption{The state of the art concerning the search for $U_{36}$ before application of the Rather's algorithm.
Each row corresponds to a different part of the entropy cube and the matrix that minimizes the distance to the unitary matrix in the average singular entropy $s_e(X) = E_S(X) + E_S(X^R) + E_S(X^\Gamma)$).
Property $\nabla s_e = 0$ denotes that vector of derivatives is equal to zero.
Matrices $X_i$ and $X_d$ are available at~\cite{matrices_Bruzda}.}\label{table:4_categories}
\end{table}

To summarize, both of these approaches strengthen the surmise concerning the local optimality of the matrix $W$.
Furthermore, we conclude that the procedures are useful in the search for optimal matrices.
Extension of these methods to other setups might also prove relevant, such as optimization over other sets of matrices.

\section{The region of the \texorpdfstring{$W$}{Lg} family}\label{sec:region_W}
Obtaining new families of matrices of order 36 (and the corresponding pure states in $\mathcal{H}_6^{\otimes 4}$) was important also from another perspective.
It fosters the search for AME(4,6) state by a proof that quantum Latin squares \emph{really} can be beneficial but also by establishing states sharing high entangling power that are not isolated.
Although the $W$ family is parametrized by five angles, the subsequent visualization will be 2-dimensional, using projection of the set of the unitary matrices into a plane spanned by $e_p$ and $g_t$.
Given \emph{a priori} set of unitary matrices, it is not clear whether the projection of such a set is more than 1-dimensional.
We demonstrated that this set forms a 2-dimensional subregion in the plane $(e_p$,\;$g_t)$.
The set fills the entire region between the boundary lines forming the allowed triangle in  Figs.~\ref{fig:ent_power_gt}-\ref{fig:region_W}, but does not contain the region close to the corner -- the matrix $U_{36}$.

\begin{figure}[H]
    \includegraphics[scale=1.08]{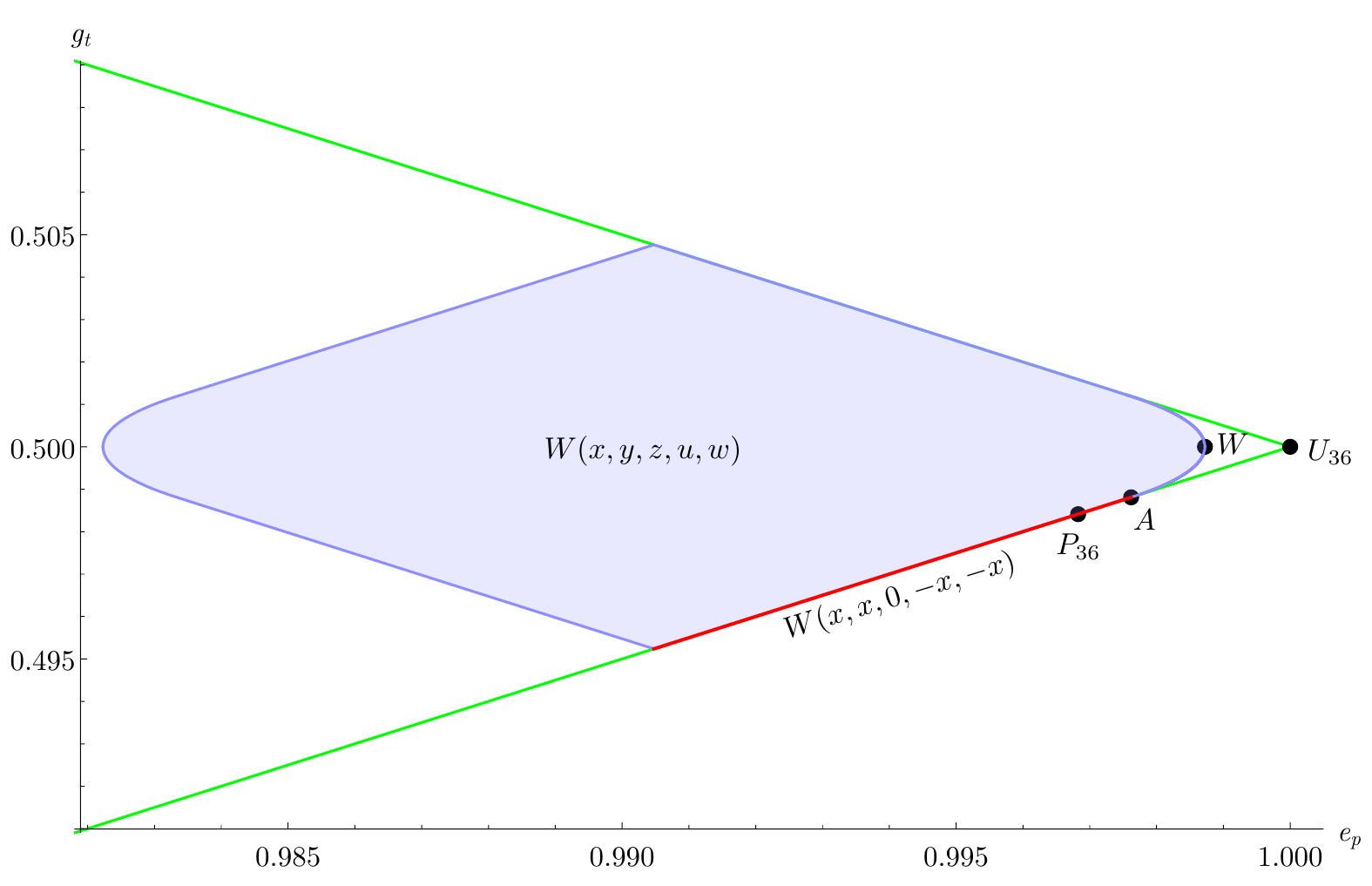}
    \caption{Unitary matrices $U(36)$ in the plane $(e_p,\, g_t)$ with the region attainable by the family $W(x,y,z,u,w)$.
    The line on the bottom, equaling the border of the unitary set, is given by a parametrization $W(x,x,0,-x,-x) $.
    The upper line can be obtained from the bottom using the mirror-like relation given by multiplying a given matrix by the swap matrix $S$.
    This set does not include the multiunitary matrix $U_{36}$ at the right corner of the triangle.
    }
    \label{fig:region_W}
\end{figure}

The region seen in Fig.~\ref{fig:region_W} is a 2D projection from a more dimensional set in the unitary manifold.
Thus, it is necessary that some of the points do not specify uniquely the matching matrices.
Indeed, the point corresponding to the matrix $A$, obtained from the parametrization of the boundary line, does not coincide with the $A$ matrix itself.
Similarly, the $W$ matrix which maximizes the entangling power has its different counterparts with the same $e_p$.
Nonetheless, since their properties are similar to properties of $W$ the study of extremality was devoted to $W$ only.

There are two regions in Fig.~\ref{fig:region_W} of a special interest.
One of them is the border of the whole set of unitary matrices, covered by the family $W(x,x,0,-x,-x)$.
This is the first characterization of extremal unitary matrices in these coordinates so close to the multiunitary matrix.
The other interesting parts of border are ellipses-like lines, bounding the set from both sides.
We have verified that these are indeed ellipses, and that they can be described as a simple, one-parameter family inside the $W$ family.
The ellipse at the right side of the figure, closer to the multiunitary matrix is parametrized by $W(-\pi/6, -\pi/12 + x,\pi/12 +x,\pi/6,\pi/4+x)$, while parameter $x \in (-\pi/12,\pi/12)$.
Note that the ellipse is tangent to the border of unitary matrices exactly at the point corresponding to the $A$ matrix.
We conjecture that the whole set is dense, i.e.\ that to every point inside the set one can associate a unitary matrix with corresponding $e_p$ and $g_t$, which has been checked only numerically.
Let us now move to the description of an algorithm allowing us to obtain an AME(4,6) state.

\section{Numerical algorithm used to find AME(4,6) state}\label{sec:Suhail_alg}
Previous sections described the search for the new families of complex matrices of order 36.
In this section, we sketch the first successful algorithm which allows obtaining the AME(4,6) state.
While in the past there were many numerical attempts, with a few published in 2020~\cite{Rather_2020,Rico_2020}, the first one to finally obtain the numerical approximation to the multiunitary matrix was Suhail Ahmad Rather~\cite{Rajchel_AME}.
His algorithm acts on any unitary matrix $U_0$ of order $N^2$ and consists of three steps.
\begin{enumerate}
    \item Take any initial unitary matrix $U_0$, and reshuffle it obtaining $U_0^R$.
    \item Perform a partial transpose on this matrix, which yields $(U_0^R)^\Gamma$.
    \item Apply the polar decomposition to the, in general non-unitary, matrix $\big(U_0^R\big)^\Gamma = VH$, where $VV^\dagger = \mathbb{I}$ and $H=H^\dagger \geq 0$. Take $U_1 \coloneqq V$.  
\end{enumerate}

Applying this algorithm several times ($U_0 \rightarrow U_1 \rightarrow ... \rightarrow U_n$) for matrices in dimensions $d^2 = 3^2$ and $d^2 = 4^2$ yields an AME(4,$d$) state with a high probability for initial matrices taken as a seed.
This is true even for random unitary matrices taken from the circular unitary ensemble~\cite{Rather_2020}.
Nevertheless, the case of $d^2 = 6^2$ is more complicated, since one needs to pick particular matrices for a seed.
It could be expected that $P_{36}$, as the permutation matrix with the highest entangling power, or $W$, or $G$ would be matrices leading to the success.
However, none of these starting matrices lead to a multiunitary matrix.
Rather unexpectedly, the matrix which proves to give a successful seed is not the one with the highest $e_p$ known before, but a slight modification to the matrix $P_{36}$ or, more precisely, a matrix in the neighborhood of the matrix $\tilde{P_s}$, visualized in Fig.~\ref{fig:seed_matrix}.

\begin{figure}[H]
\centering
	\begin{tikzpicture}
		\node (a) at (-3,0) {$\tilde{P_s}$ =
\setlength{\tabcolsep}{3pt}
\renewcommand{\arraystretch}{1.}
\begin{tabular}{|llllll|}
\hline
$11$ & $22$ & $33$ & $44$ & $55$ & $66$ \\ 
$23$ & $14$ & $45$ & $36$ & $61$ & $52$ \\ 
$32$ & $41$ & $64$ & $53$ & $16$ & $25$ \\ 
$46$ & $35$ & $51$ & $62$ & $24$ & $13$ \\ 
$64$ & $56$ & $26$ & $15$ & {\color{blue}  $43$} &  {\color{blue}  $31$} \\ 
$55$ & $63$ & $12$ & $21$ &  {\color{blue}  $42$} &  {\color{blue}  $34$} \\
\hline
\end{tabular}
 };
    \node (a) at (-0.1,0) {=};
    \node (a) at (3.8,0) {\includegraphics[scale=0.95]{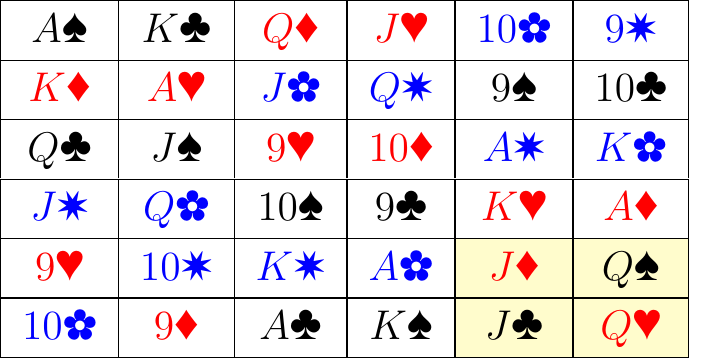}};
	\end{tikzpicture}
	\caption{Permutation matrix $\tilde{P}_s$ which, after a slight perturbation $\tilde{P}_s \mapsto \tilde{P}_s e^{i\varepsilon H}$, yields a multiunitary matrix after iterations of the Rather's algorithm.
	The mismatching elements are marked in yellow.
	}\label{fig:seed_matrix}
\end{figure}

The algorithm returns a multiunitary matrix with a certain probability, while sampling the unitary neighborhood of $\tilde{P_s}$, for more details see the joint paper~\cite{Rajchel_AME}.
Usefulness of the algorithm stems from the fact that it provides a numerical AME(4,6) state up to the machine precision.
Furthermore, some of the returned multiunitary matrices were of a special, block-like form.
Due to the dimensions of these matrices, we are unable to reproduce this particular complex matrix of size 36.
In the next section we improve our understanding of its peculiar structure.

\section{Block-like properties of the numerical multiunitary matrix}\label{sec:block-like_AME}
We focused on the particular form of the numerical multiunitary matrix, found using the iterative algorithm described in the previous section.
It facilitated the further investigation of the conditions for multiunitarity, leading to the following lemma.

\begin{lemma}\label{lemma:block_AME}
A set of 24 vectors of length 6, called $a_1,...,a_{12}, b_1,..., b_{12}$, such that three matrices $M$, $M_R$, and $M_{\Gamma}$ of order 12, defined in Table~\ref{tab:block_small}, are unitary, leads to a multiunitary matrix $U$ of size 36.
\end{lemma}
\begin{proof}
The proof of this lemma relies on the particular arrangement of these 24 vectors into the matrix of order 36.
For the full description of this arrangement, we refer the reader to Appendix~\ref{app:block_like}.
\end{proof}

\begin{table}[H]
                \vspace{3.5cm} 
                $M\coloneqq
                \quad\quad\quad\quad\quad\quad\quad\quad\quad\,\,\,\,\,\,\,\quad\,\, ,\;\;\,\,\,\,\;\,
                M_R\coloneqq
                \quad\quad\quad\quad\quad\quad\quad\quad\,\,\,\,\,\,\quad\quad\,\, ,\;\;\,\,\,\,\;\,
                M_\Gamma\coloneqq
                \quad\quad\quad\quad\quad\quad\quad\quad\quad\quad\quad\;\; $
                \\[-11.1em]
                \quad\quad\,
                \begin{tabular}[t]{|p{1.6cm}|p{1.6cm}|}
                    \hline
                    $a_1$&$b_1$\\ \hline
                    $a_2$&$b_2$\\ \hline
                    $a_3$&$b_3$\\ \hline
                    $a_4$&$b_4$\\ \Xhline{4\arrayrulewidth}
                    $a_5$&$b_5$\\ \hline
                    $a_6$&$b_6$\\ \hline
                    $a_7$&$b_7$\\ \hline
                    $a_8$&$b_8$\\ \Xhline{4\arrayrulewidth}
                    $a_9$&$b_9$\\ \hline
                    $a_{10}$&$b_{10}$\\ \hline
                    $a_{11}$&$b_{11}$\\ \hline
                    $a_{12}$&$b_{12}$\\ \hline
                \end{tabular}
                \hfill
                \begin{tabular}[t]{|p{1.6cm}|p{1.6cm}|}
                    \hline
                    $a_1$&$a_2$\\ \hline
                    $b_1$&$b_2$\\ \hline
                    $a_3$&$a_4$\\ \hline
                    $b_3$&$b_4$\\ \Xhline{4\arrayrulewidth}
                    $a_5$&$a_6$\\ \hline
                    $b_5$&$b_6$\\ \hline
                    $a_7$&$a_8$\\ \hline
                    $b_7$&$b_8$\\ \Xhline{4\arrayrulewidth}
                    $a_9$&$a_{10}$\\ \hline
                    $b_{9}$&$b_{10}$\\ \hline
                    $a_{11}$&$a_{12}$\\ \hline
                    $b_{11}$&$b_{12}$\\ \hline
                \end{tabular}
                \hfill
                \begin{tabular}[t]{|p{1.6cm}|p{1.6cm}|}
                    \hline
                    $a_1$&$a_3$\\ \hline
                    $a_2$&$a_4$\\ \hline
                    $b_1$&$b_3$\\ \hline
                    $b_2$&$b_4$\\ \Xhline{4\arrayrulewidth}
                    $a_5$&$a_7$\\ \hline
                    $a_6$&$a_8$\\ \hline
                    $b_5$&$b_7$\\ \hline
                    $b_6$&$b_8$\\ \Xhline{4\arrayrulewidth}
                    $a_9$&$a_{11}$\\ \hline
                    $a_{10}$&$a_{12}$\\ \hline
                    $b_{9}$&$b_{11}$\\ \hline
                    $b_{10}$&$b_{12}$\\ \hline  
                \end{tabular}
                \, \mbox{}
                \caption{Three analyzed unitary matrices of size 12, each consisting of 24 vectors of length 6.
                The block structure, denoted by bold lines, shows that the content of 3 blocks of size $4\times 12$ remains the same for $M$, $M_R$, and $M_{\Gamma}$.
                If in any block of matrix $M$ we trim the vectors $a_i$ and $b_i$ to be of length 2, then it will form a matrix of order 2.
                The respective two blocks from matrices $M_R$ and $M_\Gamma$ will be formed by reshuffling and partial transposition of this block; therefore, motivating the lower index notation.
                }\label{tab:block_small}
\end{table}

Distinguishing the 3 blocks from Table~\ref{tab:block_small}, we shall name them correspondingly to the vectors they are formed of: AB1 = $\{a_1,a_2,a_3,a_4,b_1,b_2,b_3,b_4\}$, AB5 = $\{a_5,...,b_8\}$, and AB9 = $\{a_9,...,b_{12}\}$.
Slightly abusing notation, we shall refer to the set of vectors and their particular arrangement into the blocks interchangeably.
All three matrices are formed by elements of the union of the blocks AB1$\,\cup\,$AB5$\,\cup\,$AB9.
The vectors forming the first block AB1 are shown in Table~\ref{tab:block_AB1}.

\begin{table}[H]
                \vspace{1.0cm} 
                AB1 = \quad\quad\quad\quad\quad\quad\quad\quad\quad\quad\quad\,\,\quad \mbox{}
                \\[-4.8em]
                \quad\quad\quad\quad\quad\quad\quad\,
                \begin{tabular}[t]{|p{1.8cm}|p{1.8cm}|}
                    \hline
                    $a_1$&$b_1$\\ \hline
                    $a_2$&$b_2$\\ \hline
                    $a_3$&$b_3$\\ \hline
                    $a_4$&$b_4$\\ \hline
                \end{tabular}
                \quad\quad\quad\quad\mbox{}
                \caption{The arrangement of 8 vectors of length 6 leading to a matrix of size $4\times 12$, given the name of AB1 block.
                These vectors, in different orders, form the first blocks in matrices $M$, $M_R$, and $M_{\Gamma}$.}\label{tab:block_AB1}
\end{table}

In Table~\ref{tab:U_block_vectors} we provide the general form of the numerical multiunitary matrix found by the algorithm described in Section~\ref{sec:Suhail_alg}, with a particular emphasis on its non-zero elements.


\begin{table}[H]
$U = \begin{pmatrix}
\begin{tabular}{p{2cm}p{2cm}p{2cm}p{2cm}p{2cm}p{2cm}}
\cline{1-2}
\multicolumn{1}{|l|}{$a_1$} & \multicolumn{1}{l|}{$b_1$} &                       &                       &                       &                       \\ \cline{1-2}
\multicolumn{1}{|l|}{$a_2$} & \multicolumn{1}{l|}{$b_2$} &                       &                       &                       &                       \\ \cline{1-4}
                       & \multicolumn{1}{l|}{} & \multicolumn{1}{l|}{$c_1$} & \multicolumn{1}{l|}{$d_1$} &                       &                       \\ \cline{3-4}
                       & \multicolumn{1}{l|}{} & \multicolumn{1}{l|}{$c_2$} & \multicolumn{1}{l|}{$d_2$} &                       &                       \\ \cline{3-6} 
                       &                       &                       & \multicolumn{1}{l|}{} & \multicolumn{1}{l|}{$e_1$} & \multicolumn{1}{l|}{$f_1$} \\ \cline{5-6} 
                       &                       &                       & \multicolumn{1}{l|}{} & \multicolumn{1}{l|}{$e_2$} & \multicolumn{1}{l|}{$f_2$} \\ \cline{1-2} \cline{5-6} 
\multicolumn{1}{|l|}{$a_3$} & \multicolumn{1}{l|}{$b_3$} &                       &                       &                       &                       \\ \cline{1-2}
\multicolumn{1}{|l|}{$a_4$} & \multicolumn{1}{l|}{$b_4$} &                       &                       &                       &                       \\ \cline{1-4}
                       & \multicolumn{1}{l|}{} & \multicolumn{1}{l|}{$c_3$} & \multicolumn{1}{l|}{$d_3$} &                       &                       \\ \cline{3-4}
                       & \multicolumn{1}{l|}{} & \multicolumn{1}{l|}{$c_4$} & \multicolumn{1}{l|}{$d_4$} &                       &                       \\ \cline{3-6} 
                       &                       &                       & \multicolumn{1}{l|}{} & \multicolumn{1}{l|}{$e_3$} & \multicolumn{1}{l|}{$f_3$} \\ \cline{5-6} 
                       &                       &                       & \multicolumn{1}{l|}{} & \multicolumn{1}{l|}{$e_4$} & \multicolumn{1}{l|}{$f_4$} \\ \cline{5-6} 
                       &                       &                       & \multicolumn{1}{l|}{} & \multicolumn{1}{l|}{$e_5$} & \multicolumn{1}{l|}{$f_5$} \\ \cline{5-6} 
                       &                       &                       & \multicolumn{1}{l|}{} & \multicolumn{1}{l|}{$e_6$} & \multicolumn{1}{l|}{$f_6$} \\ \cline{1-2} \cline{5-6} 
\multicolumn{1}{|l|}{$a_5$} & \multicolumn{1}{l|}{$b_5$} &                       &                       &                       &                       \\ \cline{1-2}
\multicolumn{1}{|l|}{$a_6$} & \multicolumn{1}{l|}{$b_6$} &                       &                       &                       &                       \\ \cline{1-4}
                       & \multicolumn{1}{l|}{} & \multicolumn{1}{l|}{$c_5$} & \multicolumn{1}{l|}{$d_5$} &                       &                       \\ \cline{3-4}
                       & \multicolumn{1}{l|}{} & \multicolumn{1}{l|}{$c_6$} & \multicolumn{1}{l|}{$d_6$} &                       &                       \\ \cline{3-6} 
                       &                       &                       & \multicolumn{1}{l|}{} & \multicolumn{1}{l|}{$e_7$} & \multicolumn{1}{l|}{$f_7$} \\ \cline{5-6} 
                       &                       &                       & \multicolumn{1}{l|}{} & \multicolumn{1}{l|}{$e_8$} & \multicolumn{1}{l|}{$f_8$} \\ \cline{1-2} \cline{5-6} 
\multicolumn{1}{|l|}{$a_7$} & \multicolumn{1}{l|}{$b_7$} &                       &                       &                       &                       \\ \cline{1-2}
\multicolumn{1}{|l|}{$a_8$} & \multicolumn{1}{l|}{$b_8$} &                       &                       &                       &                       \\ \cline{1-4}
                       & \multicolumn{1}{l|}{} & \multicolumn{1}{l|}{$c_7$} & \multicolumn{1}{l|}{$d_7$} &                       &                       \\ \cline{3-4}
                       & \multicolumn{1}{l|}{} & \multicolumn{1}{l|}{$c_8$} & \multicolumn{1}{l|}{$d_8$} &                       &                       \\ \cline{3-4}
                       & \multicolumn{1}{l|}{} & \multicolumn{1}{l|}{$c_9$} & \multicolumn{1}{l|}{$d_9$} &                       &                       \\ \cline{3-4}
                       & \multicolumn{1}{l|}{} & \multicolumn{1}{l|}{$c_{10}$} & \multicolumn{1}{l|}{$d_{10}$} &                       &                       \\ \cline{3-6} 
                       &                       &                       & \multicolumn{1}{l|}{} & \multicolumn{1}{l|}{$e_{9}$} & \multicolumn{1}{l|}{$f_{9}$} \\ \cline{5-6} 
                       &                       &                       & \multicolumn{1}{l|}{} & \multicolumn{1}{l|}{$e_{10}$} & \multicolumn{1}{l|}{$f_{10}$} \\ \cline{1-2} \cline{5-6} 
\multicolumn{1}{|l|}{$a_9$} & \multicolumn{1}{l|}{$b_{9}$} &                       &                       &                       &                       \\ \cline{1-2}
\multicolumn{1}{|l|}{$a_{10}$} & \multicolumn{1}{l|}{$b_{10}$} &                       &                       &                       &                       \\ \cline{1-4}
                       & \multicolumn{1}{l|}{} & \multicolumn{1}{l|}{$c_{11}$} & \multicolumn{1}{l|}{$d_{11}$} &                       &                       \\ \cline{3-4}
                       & \multicolumn{1}{l|}{} & \multicolumn{1}{l|}{$c_{12}$} & \multicolumn{1}{l|}{$d_{12}$} &                       &                       \\ \cline{3-6} 
                       &                       &                       & \multicolumn{1}{l|}{} & \multicolumn{1}{l|}{$e_{11}$} & \multicolumn{1}{l|}{$f_{11}$} \\ \cline{5-6} 
                       &                       &                       & \multicolumn{1}{l|}{} & \multicolumn{1}{l|}{$e_{12}$} & \multicolumn{1}{l|}{$f_{12}$} \\ \cline{1-2} \cline{5-6} 
\multicolumn{1}{|l|}{$a_{11}$} & \multicolumn{1}{l|}{$b_{11}$} &                       &                       &                       &                       \\ \cline{1-2}
\multicolumn{1}{|l|}{$a_{12}$} & \multicolumn{1}{l|}{$b_{12}$} &                       &                       &                       &                       \\ \cline{1-2}
\end{tabular}
\end{pmatrix}$
\caption{A general form of a numerical matrix of size 36 obtained with the iterative algorithm, where non-zero vectors are of length 6.
Every blank vector consists of entries equal to 0.
Due to Lemma~\ref{lemma:block_AME}, if the first two columns of vectors, $a_1,...,a_{12},b_1,...,b_{12}$, form a unitary matrix $M$ of order 12, such that $M_R$ and $M_{\Gamma}$ are unitary then, by duplicating these vectors according to Eq.~(\ref{eq:block_equality}), $U$ of order 36 is transformed to a multiunitary matrix.
}\label{tab:U_block_vectors}
\end{table}

Table~\ref{tab:U_block_vectors} admits a particular structure of blocks invariant under reshuffling and partial transposition, see Appendix~\ref{app:block_like}. 
Analogously to the smaller case discussed in Lemma~\ref{lemma:block_AME}, we shall distinguish 9 different blocks, as presented in Table~\ref{tab:define_blocks}.

\begin{table}[H]
              \begin{tabular}[t]{|l|l|l|}
                    \hline
                     AB1 = $\{a_1,...,b_4\}$& CD1 = $\{c_1,...,d_4\}$ &  EF1 = $\{e_1,...,f_4\}$\\ \hline
                    AB5 = $\{a_5,...,b_8\}$ & CD5 = $\{c_5,...,d_8\}$& EF5 = $\{e_5,...,f_8\}$\\ \hline
                     AB9 = $\{a_9,...,b_{12}\}$ &CD9 = $\{c_9,...,d_{12}\}$ &EF9 = $\{e_9,...,f_{12}\}$\\ \hline
                \end{tabular}
                \caption{Non-zero vectors of the matrix $U$ of order 36, presented in Table~\ref{tab:U_block_vectors}, arranged into 9 blocks of size $4\times 12$, called AB1, AB5, AB9, CD1, CD5, CD9, EF1, EF5, and EF9.
                Each column of this array corresponds to non-zero vectors from two consecutive columns of $U$.
                }\label{tab:define_blocks}
\end{table}

The block notation was introduced to highlight the non-mixing property between different blocks, as any matrix $U$ admitting the form of Table~\ref{tab:U_block_vectors} will retain its structure of non-zero vectors under both reshuffling and partial transposition, see Appendix~\ref{app:block_like}.
The operations exchange the positions of blocks, e.g.\ $U \mapsto U^\Gamma$ exchanges AB5 and CD1.
Nonetheless, any block in the matrices $U$, $U^R$, and $U^\Gamma$ is composed of the vectors from a single block of Table~\ref{tab:define_blocks}.

The most important achievement of Lemma~\ref{lemma:block_AME} is the reduction in the dimensionality of the search for a multiunitary matrix of size 36.
Instead of considering all $36^2=1296$ elements of a unitary matrix of order 36, it suffices to find 3 blocks of size $4\times 12$.
Therefore, we reduced the number of non-zero elements by a factor of 9.
These three blocks form a matrix $M$ of order 12 which, provided it satisfies the conditions of Lemma~\ref{lemma:block_AME}, can be extended to a multiunitary matrix.
To this end, we aim to find AB1, AB5, and AB9 blocks such that matrices $M$, $M_R$, and $M_\Gamma$, defined in Table~\ref{tab:block_small}, are unitary.
Then, by setting all the other 6 non-zero blocks of $U$, given by Table~\ref{tab:U_block_vectors}, to be
\begin{equation}\label{eq:block_equality}
    \mathrm{AB1 = CD5 = EF9,\;\;\;\; AB5 = CD9 = EF1, \;\;\; and \;\;\; AB9 = CD1 = EF5}
\end{equation}
we obtain a multiunitary matrix of order 36.
This can be proven by the resulting unitarity of $U^R$ and $U^\Gamma$, as shown in Appendix~\ref{app:block_like}.
In our notation, the equality of blocks should be understood as the equality of the corresponding vectors.


We shall search for three blocks of size $4\times 12$ that satisfy the conditions of Lemma~\ref{lemma:block_AME}.
The smallest structure consists of setting all $4\times 12$ blocks AB1, AB5, and AB9 in such a way that the non-zero elements of different blocks do not overlap in columns.
Thus, every block consists of a single non-zero submatrix of size $4\times 4$.
Due to the non-overlapping, it is sufficient to select only one of these blocks, i.e.\ one matrix of order 4, while the others might be duplicated.
Then, the constraints imposed by Lemma~\ref{lemma:block_AME} on the block are \emph{equivalent} to the conditions for the multiunitarity of this block, see Table~\ref{tab:block_AME42}.
However, this would lead to the corresponding AME(4,2) state, which is known to be non-existent~\cite{Higuchi_2000}; therefore, such an arrangement is impossible.

\begin{table}[H]
    \vspace{1.cm} 
        $M\coloneqq
    \quad\quad\quad\quad\quad\quad\quad\quad\quad\,\,\,\,\,\,\,\quad\,\, ,\;\;\,\,\,\,\;\,
    M_R\coloneqq
    \quad\quad\quad\quad\quad\quad\quad\quad\,\,\,\,\,\,\quad\quad\,\, ,\;\;\,\,\,\,\;\,
    M_\Gamma\coloneqq
    \quad\quad\quad\quad\quad\quad\quad\quad\quad\quad\quad\;\; $
    \\[-4.8em]
    \quad\quad\,    
    \begin{tabular}[t]{|p{1.6cm}|p{1.6cm}|}
        \hline
        $a_1$&$b_1$\\ \hline
        $a_2$&$b_2$\\ \hline
        $a_3$&$b_3$\\ \hline
        $a_4$&$b_4$\\ \hline
    \end{tabular}
    \hfill
    \begin{tabular}[t]{|p{1.6cm}|p{1.6cm}|}
        \hline
        $a_1$&$a_2$\\ \hline
        $b_1$&$b_2$\\ \hline
        $a_3$&$a_4$\\ \hline
        $b_3$&$b_4$\\ \hline
    \end{tabular}
    \hfill
    \begin{tabular}[t]{|p{1.6cm}|p{1.6cm}|}
        \hline
        $a_1$&$a_3$\\ \hline
        $a_2$&$a_4$\\ \hline
        $b_1$&$b_3$\\ \hline
        $b_2$&$b_4$\\ \hline
    \end{tabular}
    \, \mbox{}
    \caption{Three matrices $M$, $M_R$, and $M_\Gamma$ of order 4, each formed by 8 vectors $a_1,a_2,a_3,a_4,b_1,b_2,b_3,$ and $b_4$ of length 2.
    Note that these matrices form the first blocks of the respective matrices defined in Table~\ref{tab:block_small}, provided that we extend vectors to length 6 by appending zero entries.
    Unitarity of these three matrices implies that matrix $M$ is multiunitary since $M_R=M^R$ and $M_\Gamma=M^\Gamma$.    
    }\label{tab:block_AME42}
\end{table}

Ultimately, all three blocks AB1, AB5, and AB9 cannot be treated independently.
In the next section, we shall describe the further progress in the search for a multiunitary matrix of size 36 of the block structure.

\section{Search using the block structure}\label{sec:search_block_structure}
Using Lemma~\ref{lemma:block_AME} we conclude that it is possible to tremendously simplify the search for an analytical multiunitary matrix of order 36 and the corresponding AME(4,6) state, exploiting the block structure thus decreasing the number of parameters by a factor of 9.
Then, the search can be conducted e.g.\ by utilizing the set of real Hadamard matrices, the notion recalled in Section~\ref{sec:sets_of_matrices}.
We tried several different combinations, with choosing vectors $a_i$ and $b_i$ that form the blocks as the consecutive rows of real Hadamard matrices of sizes 6, 8, and 12.
Finally, applying Eq.~(\ref{eq:block_equality}), one can extend the resulting three blocks AB1, AB5, and AB9 to all 9 blocks of the matrix of order 36, given in Table~\ref{tab:U_block_vectors}.
However, the search yielded no solutions that approach the entangling power of the best permutation matrix.

Therefore, we utilized another set of matrices in our investigation.
Complex Hadamard matrices, after rescaling, form a subset of unitary matrices that, due to their high degree of symmetry, is a good candidate for a search space.
The numerical exploration was conducted starting from a randomly chosen Hadamard matrix, then multiplying rows and columns by a random complex phase.
Applying this procedure, the highest entangling power achieved is equal to $e_p \approx 0.979$, which is lower than the corresponding value of entangling power for the best permutation matrix, $e_p(P_{36})\approx 0.9968$.

This concluded the search in the subset of Hadamard matrices.
The last method utilized to find the matrix of order 36 of the highest entangling power was the investigation of vectors forming the matrices $M$, $M_R$, and $M_\Gamma$, defined in Table~\ref{tab:block_small}, that are in the same position in all three matrices.
Each block AB1, AB5, and AB9 admits two invariant vectors -- in total there are 6 of them: $a_1$, $a_5$, $a_9$, $b_4$, $b_8$, and $b_{12}$.
Using this structure for reducing the number of parameters of the search, we obtained the value 0.985 as the best entangling power.

Finally, we conclude that the simplification given by Lemma~\ref{lemma:block_AME} provides a new perspective to the search of the multiunitary matrix of order 36 and the corresponding AME(4,6) state. Nonetheless, using the techniques described above we were not able to obtain an analytical multiunitary matrix.
Subsequently, we shall move on to the optimization by local unitary rotations, which yielded the final answer to our search in the form of the multiunitary matrix $U_{36}$, leading to the golden AME(4,6) state.

\section{Golden AME(4,6) state}\label{sec:analytical_AME}
Research conducted by Adam Burchardt and Wojciech Bruzda concerning rotations on the AME(4,6) state obtained by the numerical algorithm allowed them to find the analytical form of the multiunitary matrix.
This line of reasoning stemmed from the invariance of entanglement under LOCC processes, in particular under local unitary operations.
To formalize this remark, two matrices $M$ and $N$ of order 36 have the same entangling power if they are equivalent with respect to local unitaries,
\begin{equation}
    M = (U_1 \otimes U_2)\; N\; (U_3 \otimes U_4),
\end{equation}
with $U_i$ representing local unitary matrices of order 6.
By a careful manipulation of the matrix, it was possible to obtain an analytical form of a  multiunitary matrix of size 36, with only 4 or 2 non-zero elements in each row.
This matrix is depicted in Fig.~\ref{fig:AME_table}, with the values 
\begin{align}
a=&\big(\sqrt{2}(\omega+\overline{\omega})\big)^{-1}=\big(5+\sqrt{5}\big)^{-1/2}, \nonumber \\
b=&\big(\sqrt{2}(\omega^3+\overline{\omega}^{3})\big)^{-1}=\big((5+\sqrt{5})/20\big)^{1/2},   \\
c=&1/\sqrt{2},\nonumber
\end{align}
where the golden ratio is maintained by the values $b/a=\phi=(1+\sqrt{5})/2$; thus, motivating the name of the \emph{golden} AME(4,6) state.
The phase of any element is a multiplicity of a root of unity of order 20, written $\omega = e^{i\pi/10}$.

\begin{figure}[H]
    \includegraphics[scale=0.92]{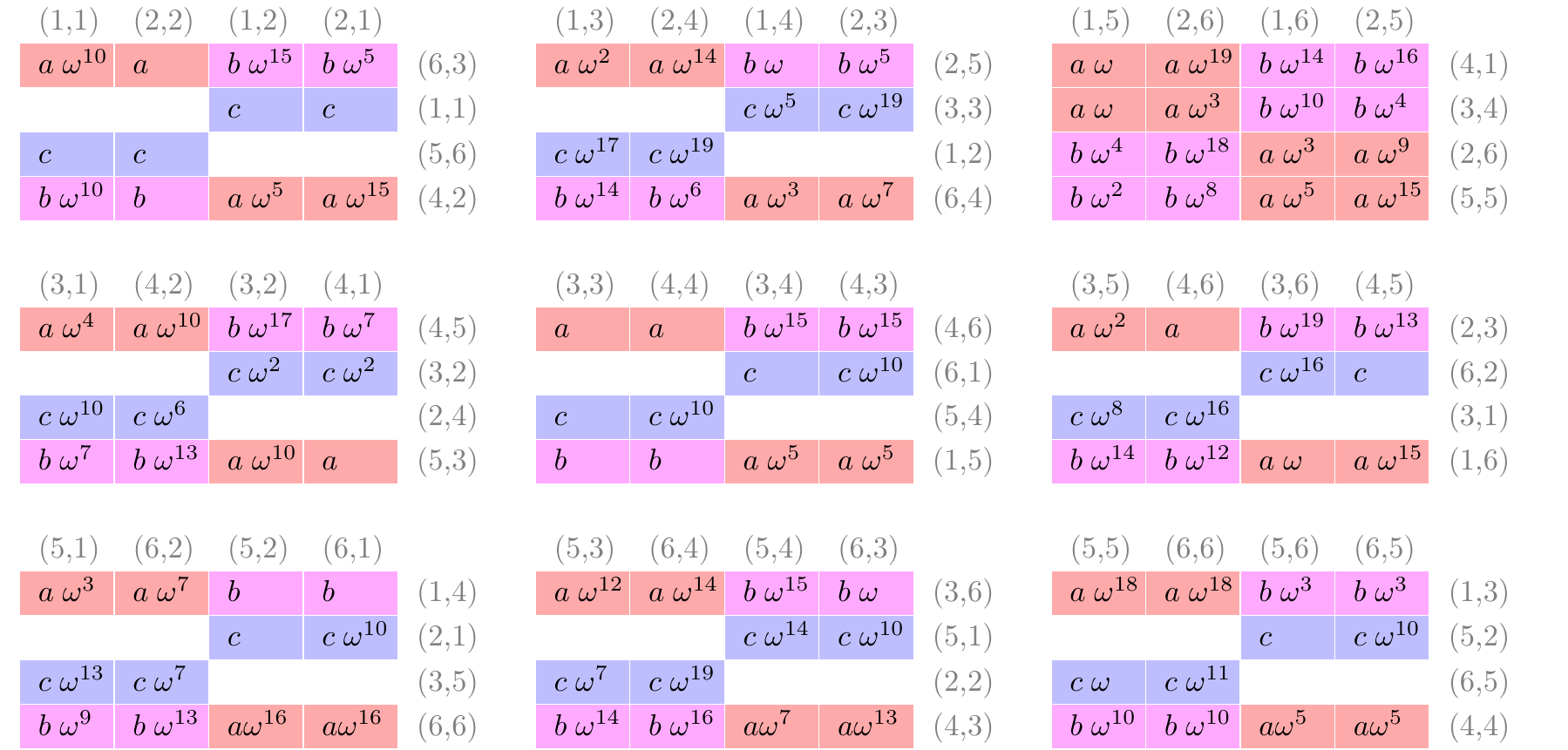}
    \caption{Analytical golden AME(4,6) state with each element depicted together with its indices for row $(i,j)$ and columns $(k,l)$.
    The indices refer to the position of a given row/column in the number system of base six, e.g.\ $(1,1)$ refers to a position 1.
    In order to observe the state either in a full matrix form or in a form of AME state, all the necessary files are available at~\cite{AME_files}. Figure reproduced from the joint paper~\cite{Rajchel_AME}.}
    \label{fig:AME_table}
\end{figure}

To exemplify the connection of the golden multiunitary matrix presented in Fig.~\ref{fig:AME_table} to an OQLS in a similar way in which $P_{36}$ is associated with a pair of Latin squares, let us present the state $\ket{\mathrm{AME}}$ in Fig.~\ref{fig:AME_cards} by using entangled cards from a quantum deck.
This implies non-trivial correlations between their entries -- non-zero elements belong to the same column.
Despite having elements in the same columns, the rows are orthogonal due to the unitarity of the matrix -- a similar trait cannot be achieved by the usage of permutation matrices in the case of permutation matrices and classical orthogonal Latin squares.
 
\begin{figure}[H]
    \includegraphics[scale=0.92]{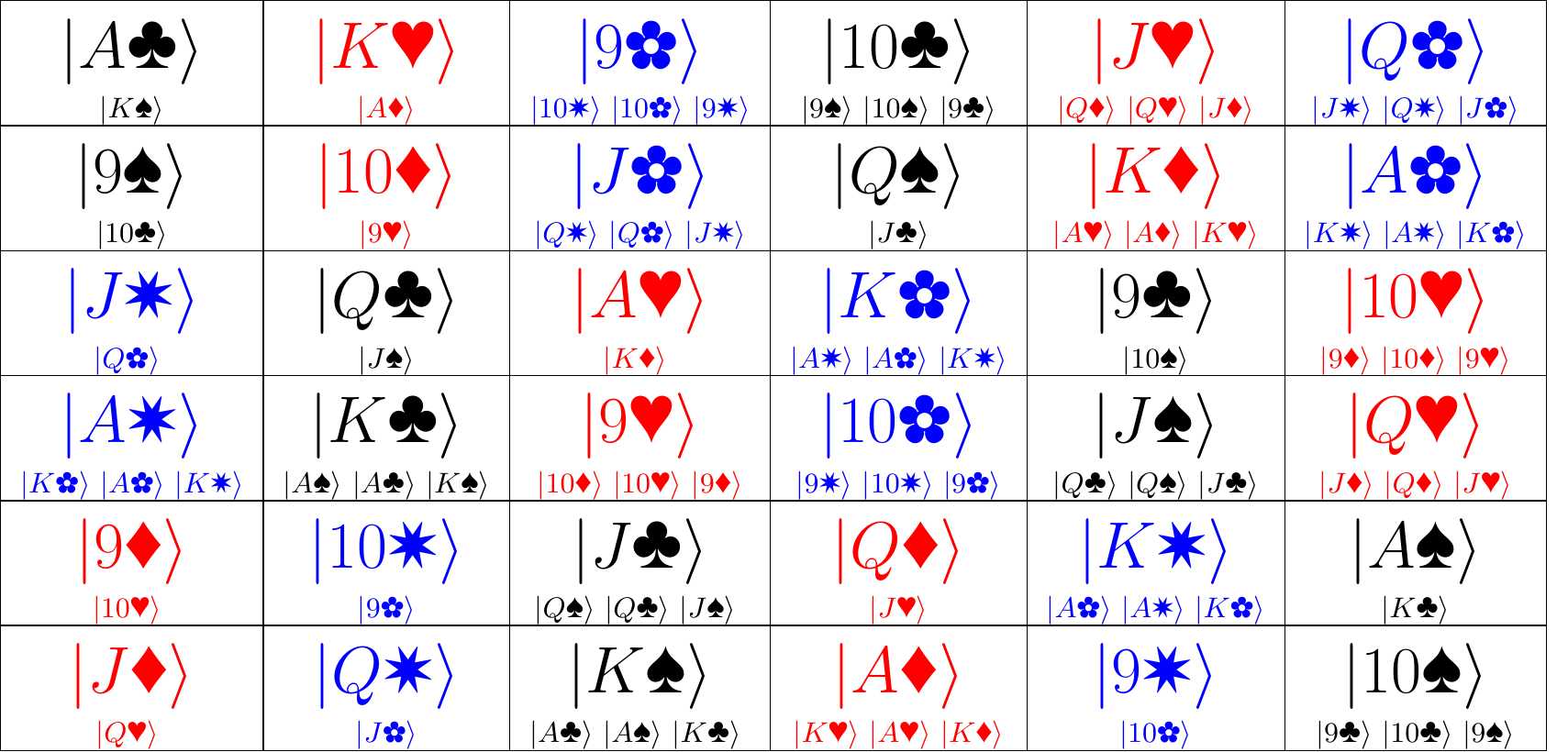}
    \caption{The golden AME(4,6) state depicted as a quantum orthogonal Latin square.
    Every element of a table is composed of 2 or 4 quantum states, which correspond to the positions of non-zero elements of the unitary matrix $U_{36}$ indicated in consecutive rows of Fig.~\ref{fig:AME_table}.
    Block structure is observed since e.g.\ aces and kings are not coupled to other figures and a similar property holds for suits.
    The explanation of the card symbols is contained in the caption of Fig.~\ref{fig:almost_OLS_6}.
    To reconstruct the entire multiunitary matrix one needs concrete numbers that appear in Fig.~\ref{fig:AME_table}.}
    \label{fig:AME_cards}
\end{figure}

Fig.~\ref{fig:AME_cards} illustrates the connection to the standard form of quantum states, i.e.\ in the bra-ket notation.
Using the notation from Lemma~\ref{lemma:OQLS_AME} we write down 36 states $\ket{\psi_{ij}}$ forming the OQLS and leading to the full AME(4,6) state, $\ket{\mathrm{AME}} = \frac{1}{N}\sum_{ij} \ket{ij}\ket{\psi_{ij}}$,
\begin{eqnarray}
 |\psi_{11}\rangle &=& c \ket{\text{A\club}} + c \ket{\text{K\spade}}, \nonumber \\
 |\psi_{12}\rangle  &=&  c \omega^{19} \ket{{\color{red}\text{K\heart}}} + c \omega^{17} \ket{{\color{red}\text{A\diamond}}} ,   \nonumber \\
  & & ....  
  \nonumber \\
 |\psi_{66}\rangle  &=&  
 b\omega^9 \ket{\text{10\spade}}+b\omega^{12}\ket{\text{9\club}}+a\omega^{16}\ket{\text{10\club}}+a\omega^{16}\ket{\text{9\spade}}, \nonumber
\end{eqnarray}
which combines the notation of two figures, Fig.~\ref{fig:AME_table} and \ref{fig:AME_cards}, into the golden AME(4,6) state.
As a final statement concluding the search, let us summarize the results analogously to Theorem~\ref{theorem:OLS}.
\begin{theorem}[Existence of OQLS]\label{theorem:QOLS}
    There exist OQLS of size 6.
    Therefore, orthogonal quantum Latin squares and AME(4,$N$) states exist in all dimensions for $N\geq 3$.
\end{theorem}

\section{Conclusions}
The search for a quantum equivalent of the famous Euler problem was tackled for many years in quantum information due to various connections and applications, e.g.\ to multipartite entanglement, and perfect tensors.
In this chapter, we report the positive solution to the problem in the form of the golden AME(4,6) state, as well as several developments that might prove interesting to solve other problems. 
Furthermore, some of the matrices discovered while exploring the highly entangled multipartite state are conjectured to be local extrema of functions such as the entangling power.
Therefore, they provide an interesting and fairly simple approximation to AME(4,6) state.
The main problem was solved; however, there are still many interesting questions that might be investigated by future researches.
We shall discuss them in Chapter~\ref{Summary}. 

Our solution opens a path to find other applications to quantum designs in creating highly entangled quantum states.
The author hopes that similar discoveries will be made in the near future, since quantum designs broaden our understanding of the quantum world.
The golden AME state implies the existence of a new quhex quantum error correction code, which may have potential applications in future quantum technologies.
\clearpage
\chapter{Genuinely quantum SudoQ}
\label{chapter_7}
\vspace{-1cm}
\rule[0.5ex]{1.0\columnwidth}{1pt} \\[0.2\baselineskip]

\section{Introduction}

The emergence of the theory of quantum mechanics prompted researchers to search for its applications in different areas of science.
Most notably, several branches of mathematics, such as functional analysis, gained a perspective of physical motivation.
Another example of a mathematical field that can be altered by the introduction of quantum mechanics is game theory, a branch of mathematics younger than quantum mechanics.
A cornerstone of the quantum subfield of game theory was laid by the 1999 paper of Eisert et al., in which the authors studied the famous prisoner's dilemma while extending classical strategies to their quantum superpositions \cite{Eisert_1999}.
The generalization helped the authors to find new solutions, optimal under altered assumptions.
What is more, in the case of asymmetry, the player having access to a superposition of classical strategies has an advantage over the classical player~\cite{Dariano_2002},\cite{Piotrowski_2003}.
This example shows that the incorporation of quantum rules into classical games significantly enriches the possibilities of strategies.

In 2020, this concept was extended to the popular Sudoku game, introduced briefly in Section~\ref{sec:Sudoku}.
Similar to the quantum generalization of Latin squares introduced in 2016 by Musto and Vicary~\cite{Musto_2016}, a quantum version of Sudoku was put forward by Nechita and Pillet in 2020~\cite{Nechita_2020}.
Quantum Sudoku is a notion that brings together classical and quantum combinatorial designs to study them for their deep connections to envisaging new schemes for quantum measurements.

We expand the idea of Nechita and Pillet~\cite{Nechita_2020}, introducing the notion of \emph{genuinely} quantum designs, which are those that cannot be brought to the classical form by means of unitary rotations.
The present chapter shall focus on certain properties of those non-classical quantum Sudoku.
A summary of some parts of this chapter, as well as an extension of the others, in which the author's involvement was less substantial, can be found in a joint paper~\cite{Rajchel_SudoQ}.
If not specified differently, the author's contribution to the work covered by this chapter was significant.

\begin{figure}[H]
    \includegraphics[width=1\columnwidth]{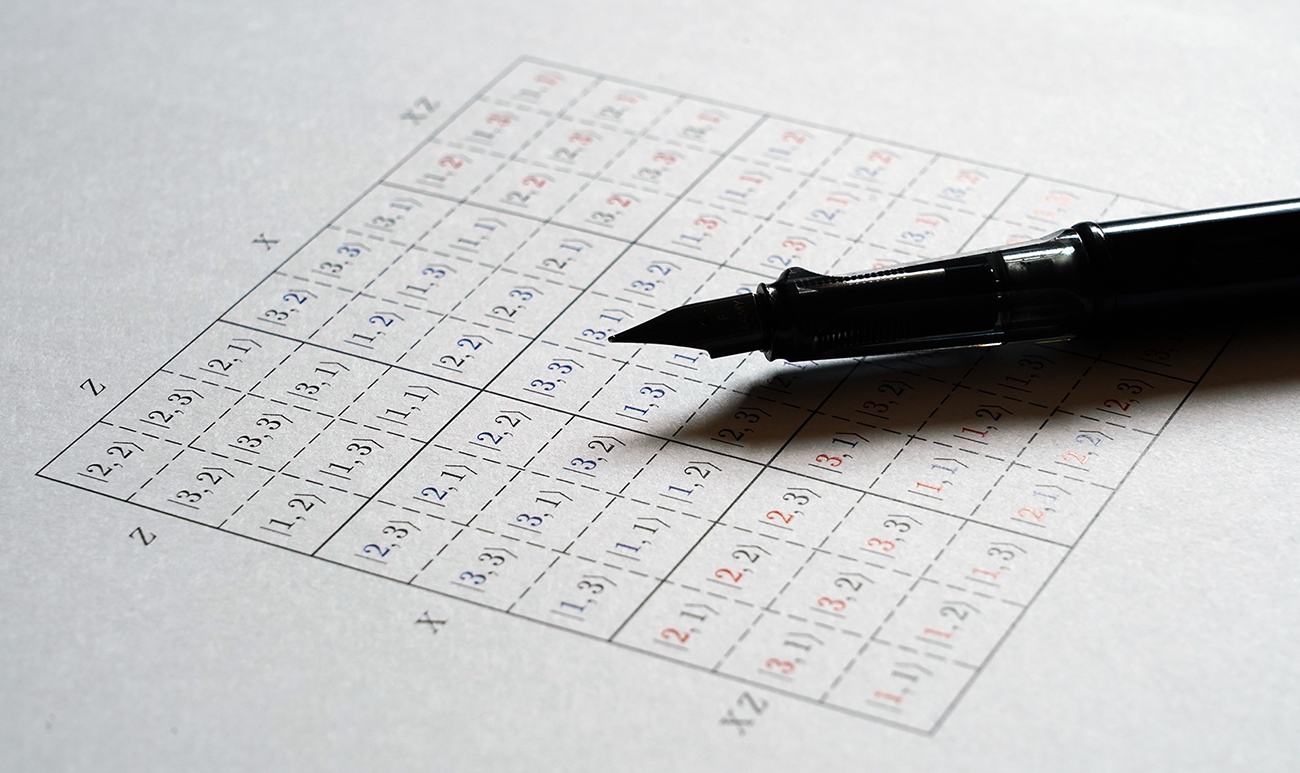}
    \caption{Genuinely quantum SudoQ can be constructed and solved using similar methods to the construction of mutually unbiased bases.
    The name of ``Sudoku'' is the Japanese abbreviation of \begin{CJK}{UTF8}{}\begin{Japanese}\mbox{数字は独身に限る}\end{Japanese}\end{CJK} (``Suuji wa dokushin ni kagiru''), which means ``numbers must occur only once'' or ``numbers better be single!''~\cite{Jana_2015}.
    This coincides with our motivation for conducting research of SudoQ from the perspective of quantum designs since in the genuinely quantum SudoQ each state repeats exactly once.
    The author is grateful to his wife for this artistic depiction of quantum Sudoku.}
    \label{fig:sudoQ_pen}
\end{figure}

\section{Classical Sudoku}\label{sec:Sudoku}
The research of quantum Sudoku is based on its classical predecessor.
\begin{definition}[Sudoku]
    An $N^2 \times N^2$ combinatorial arrangement of $N^2$ symbols from the set $\{1,...,N^2\}$, such that in every row, column, and $N\times N$ block all elements are distinct is called Sudoku design.
\end{definition}

As an example, we shall present the simplest Sudoku design possible in dimension $N^2= 2^2 = 4$,

\begin{equation}\label{eq:Sudoku}
 \begin{tabular}{!{\vrule width 1pt}c|c!{\vrule width 1pt}c|c!{\vrule width 1pt}}
    \noalign{\hrule height 1pt}
    $1$ & $2$ & $3$ & $4$ \tabularnewline
    \hline
    $3$ & $4$ & $1$ & $2$ \tabularnewline
    \noalign{\hrule height 1pt}
    $4$ & $1$ & $2$ & $3$ \tabularnewline
    \hline
    $2$ & $3$ & $4$ & $1$ \tabularnewline
    \noalign{\hrule height 1pt}
\end{tabular}.
\end{equation}
Any Sudoku design belongs to the class of Latin squares since it satisfies the row/columns conditions.
Apart from being a popular puzzle, Sudoku designs form a part of combinatorial designs.
As such, they encompass deep mathematical questions, including numbers of different arrangements~\cite{Felgenhauer_2005}, the minimal number of clues required for a unique solution~\cite{Mcguire_2014}, and its complexity (proved to be NP-complete~\cite{Yato_2003}).

\section{Quantum Sudoku}\label{sec:quantum_sudoku}
Based on the popular classical game of Sudoku, we recall its quantum implementation, first introduced by Nechita and Pillet~\cite{Nechita_2020}.
\begin{definition}[SudoQ]
    An $N^2\times N^2$ arrangement of $N^4$ quantum states from a Hilbert space $\mathcal{H}_{N^2}$, such that each row, each column, and each $N\times N$ block forms a basis of $\mathcal{H}_{N^2}$, is called quantum Sudoku (SudoQ).
\end{definition}

Any classical Sudoku can be \emph{quantized} in the sense that each element is mapped to a different computational basis $\mathcal{B}$ vector -- $\{1,...,4\}$ into $\mathcal{B} = \{\ket{1},...,\ket{4}\}$.
This transformation yields a SudoQ, as exemplified below by quantizing the Sudoku from arrangement~(\ref{eq:Sudoku})

\begin{equation}\label{eq:SudoQ_cardinality4}
     \begin{tabular}{!{\vrule width 1pt}c|c!{\vrule width 1pt}c|c!{\vrule width 1pt}}
    \noalign{\hrule height 1pt}
    $\ket{1}$ & $\ket{2}$ & $\ket{3}$ & $\ket{4}$ \tabularnewline
    \hline
    $\ket{3}$ & $\ket{4}$ & $\ket{1}$ & $\ket{2}$ \tabularnewline
    \noalign{\hrule height 1pt}
    $\ket{4}$ & $\ket{1}$ & $\ket{2}$ & $\ket{3}$ \tabularnewline
    \hline
    $\ket{2}$ & $\ket{3}$ & $\ket{4}$ & $\ket{1}$ \tabularnewline
    \noalign{\hrule height 1pt}
\end{tabular}.
\end{equation}

The above example of SudoQ is the smallest dimensional, $N=2$. 
However, using the simple quantization rule, it is possible to create such a design for an arbitrary $N$.
Any such design composed of quantum states is also trivially a quantum Latin square, a notion already studied in Chapter~\ref{chapter_6}.
A quantum design admitting orthogonality rules is connected with the design of experiments since it admits a particular construction of a projective measurement, see Section~\ref{sec:basics_QI}.
In order to construct a projective measurement, it suffices to have a basis of the appropriate Hilbert space, which is granted by rows, columns, or blocks.

Furthermore, such a design grants more experimental setups than a simple list of $N^4$ vectors.
This list, if assuming only orthogonality in rows, corresponds to $N^2$ different projective measurements.
However, due to additional constraints, available due to the column and block structure, any quantum Sudoku yields a collection of $3N^2$ orthogonal measurements.
Nonetheless, if all the elements of SudoQ are taken from the computational basis $\mathcal{B}$, then all of these measurements are the same.
As we shall see in the next section, this obstacle shall be removed by a deeper investigation of the number of different elements in the design.

\section{Cardinality}
The relation between classical Sudoku designs and their quantum counterparts is not reflexive, i.e. not all SudoQ designs can be transformed into Sudoku.
We distinguish the subset of quantum Sudoku that have their classical equivalent by a name of \emph{classical} SudoQ.

\begin{definition}[Classical QLS]
    A QLS design is called classical if it composes only of quantum states belonging to the computational basis $\mathcal{B}$.
\end{definition}

As an example of non-classical SudoQ, consider the simple unitary transformation $U$ acting on the design~(\ref{eq:SudoQ_cardinality4}), such that $U\ket{1} = \ket{1}$, $U\ket{2} = \ket{2}$, $U\ket{3} = \frac{1}{\sqrt{2}}(\ket{3}+\ket{4})$, and $U\ket{4} = \frac{1}{\sqrt{2}}(\ket{3}-\ket{4})$.
It yields the following SudoQ with $N=2$

\begin{equation}\label{eq:sudoku4x4_card_6}
         \begin{tabular}{!{\vrule width 1pt}c|c!{\vrule width 1pt}c|c!{\vrule width 1pt}}
    \noalign{\hrule height 1pt}
    $\ket{1}$ & $\ket{2}$ & $\ket{3}+\ket{4}$ & $\ket{3}-\ket{4}$ \tabularnewline
    \hline
    $\ket{3}+\ket{4}$ & $\ket{3}-\ket{4}$ & $\ket{1}$ & $\ket{2}$ \tabularnewline
    \noalign{\hrule height 1pt}
    $\ket{3}-\ket{4}$ & $\ket{1}$ & $\ket{2}$ & $\ket{3}+\ket{4}$ \tabularnewline
    \hline
    $\ket{2}$ & $\ket{3}+\ket{4}$ & $\ket{3}-\ket{4}$ & $\ket{1}$ \tabularnewline
    \noalign{\hrule height 1pt}
\end{tabular},
\end{equation}
where, for brevity, we neglected the normalization coefficients.
Note that all of the orthogonality conditions are preserved; however, the above example cannot be created directly out of a classical one.
Nonetheless, by a simple action of inverse unitary operator $U^\dagger$ it is possible to revert the design to its classical equivalent, as the number of different vectors in the pattern equals $N^2$.
On the other hand, this is not always the case.
If the number of elements of a given SudoQ of size $N^2$ is greater than $N^2$, then no unitary operation can change its elements into the computational basis.
This prompted us to propose a new notion concerning SudoQ designs and, more generally, quantum Latin squares.

\begin{definition}[Cardinality]\label{def:cardinality}
    The cardinality \emph{c} of a quantum Latin square is the number of its distinct elements.
\end{definition}
SudoQ designs form a subset of quantum Latin squares; therefore, the above definition is useful for their classification.
The importance of the above notion stems from the fact that it allows for a classification of quantum Latin squares by their degree of \emph{quantumness}.

\begin{definition}[Apparently quantum QLS]
    A QLS design of size $N^2$ is called apparently quantum if it composes only of $N^2$ quantum states, which do not belong to the computational basis $\mathcal{B}$.
\end{definition}

Ensuing the above remarks we notice that all \emph{apparently} quantum SudoQ, like the one shown in (\ref{eq:sudoku4x4_card_6}), can be unitarily converted into a classical SudoQ.

\begin{definition}[Genuinely quantum QLS]\label{def:genuinely_quantum_QLS}
    A QLS design of size $N^2$ is called genuinely quantum if it composes of more than $N^2$ quantum states.
\end{definition}
As an example of the above notion, consider \emph{genuinely} quantum SudoQ design of cardinality $c=6$
\begin{equation}\label{eq:SudoQ_cardinality6}
     \begin{tabular}{!{\vrule width 1pt}c|c!{\vrule width 1pt}c|c!{\vrule width 1pt}}
    \noalign{\hrule height 1pt}
    $\ket{1}$ & $\ket{2}$ & $\ket{3}$ & $\ket{4}$ \tabularnewline
    \hline
    $\ket{3}$ & $\ket{4}$ & $\ket{1}$ & $\ket{2}$ \tabularnewline
    \noalign{\hrule height 1pt}
    $\ket{2}-\ket{4}$ & $\ket{1}$ & $\ket{2}+\ket{4}$ & $\ket{3}$ \tabularnewline
    \hline
    $\ket{2}+\ket{4}$ & $\ket{3}$ & $\ket{2}-\ket{4}$ & $\ket{1}$ \tabularnewline
    \noalign{\hrule height 1pt}
\end{tabular}.
\end{equation}

A genuinely quantum Sudoku cannot be converted into a classical SudoQ by a means of unitary transformations; therefore, explaining its name and usefulness for quantum information protocols.
As a side note, let us mention that the golden AME(4,6) state, mentioned in Section~\ref{sec:analytical_AME}, forms also a pair of genuinely quantum OQLS.
In dimension 6 there is no pair of classical orthogonal Latin squares.
Therefore, golden AME state cannot be unitarily equivalent to such an arrangement.
Cardinality of the golden AME(4,6) state in the OQLS form (Fig.~\ref{fig:AME_cards}) equals to the maximal value of $c=36$.
In the subsequent sections, we shall focus on the characterization of the cardinality.
This study shall lead us to the demonstration of the attainability of the maximal cardinality equaling the number of all elements of a given design.

\section{Cardinality of small quantum Latin squares}
Let us inspect the smallest dimensional quantum Latin squares.
\begin{remark}
There are no genuinely quantum QLS of dimensions 2 and 3.
\end{remark}
\begin{proof}
    In the case of dimension 2 every vector in the pattern is orthogonal to all but one other vectors; therefore, the cardinality must equal 2.
    Without loss of generality, as the first row in an QLS of dimension 3 we set vectors from the computational basis $\ket{1}$, $\ket{2}$, and $\ket{3}$.
    Then, every element in the second row must also be from the computational basis.
    Otherwise, at least one of these vectors would not satisfy the orthogonality conditions.
    Analogous reasoning for the third row shows that there are no genuinely quantum QLS of dimension 3. 
\end{proof}

The situation becomes more interesting in the case of $4\times 4$ quantum Sudoku designs, as we have seen already with SudoQ of non-minimal cardinality $c=6$, presented by arrangement~(\ref{eq:SudoQ_cardinality6}).
The most prominent example of a genuinely quantum SudoQ in this dimension is the symmetric arrangement possessing the maximal cardinality $c=16$,
\begin{equation}\label{eq:SudoQ_cardinality16}
     \begin{tabular}{!{\vrule width 1pt}c|c!{\vrule width 1pt}c|c!{\vrule width 1pt}}
    \noalign{\hrule height 1pt}
    $\ket{1}$ & $\ket{2}$ & $\ket{3}+\ket{4}$ & $\ket{3}-\ket{4}$ \tabularnewline
    \hline
    $\ket{3}$ & $\ket{4}$ & $\ket{1}-\ket{2}$ & $\ket{1}+\ket{2}$ \tabularnewline
    \noalign{\hrule height 1pt}
    $\ket{2}+\ket{4}$ & $\ket{1}-\ket{3}$ & $\ket{1}+\ket{2}+\ket{3}-\ket{4}$ & $\ket{1}-\ket{2}+\ket{3}+\ket{4}$ \tabularnewline
    \hline
    $\ket{2}-\ket{4}$ & $\ket{1}+\ket{3}$ & $\ket{1}+\ket{2}-\ket{3}+\ket{4}$ & $\ket{1}-\ket{2}-\ket{3}-\ket{4}$ \tabularnewline
    \noalign{\hrule height 1pt}
\end{tabular}.
\end{equation}

The SudoQ design presented above is of the maximal entropy in the sense that the sum of the Shannon entropies of coefficients of its elements, written in the computational basis, is maximal (see Appendix C in the full paper~\cite{Rajchel_SudoQ}).
Jerzy Paczos and Marcin Wierzbi{\'n}ski provided a full parametrization of $4\times 4$ SudoQ designs, together with the characterization of all possible cardinalities.
We summarize these results by the theorem for which proof is available in the complete paper~\cite{Rajchel_SudoQ}.

\begin{theorem}\label{thm:SudoQ_cardinalities4x4}
    A SudoQ design of size $4\times 4$ admits only $c=4,6,8$, and $16$ as possible cardinalities.
\end{theorem}

Furthermore, we showed that there exists a similar result like the one for classical Sudoku design that states that there are uniquely solvable grids using only 4 clues -- non-empty entries in the initial array.

\begin{proposition}
    Any incomplete quantum design of the maximal cardinality $c=16$ consisting of 4 clues, placed as
    \begin{equation}
    \begin{tabular}{!{\vrule width 1pt}c|c!{\vrule width 1pt}c|c!{\vrule width 1pt}}
    \noalign{\hrule height 1pt}
    $e_1$ &   &   &   \tabularnewline
    \hline
      &   & $f_3$ &   \tabularnewline
    \noalign{\hrule height 1pt}
      & $v_2$ &   &   \tabularnewline
    \hline
      &   &   & $u_4$ \tabularnewline
    \noalign{\hrule height 1pt}
    \end{tabular},
\end{equation}
admits a unique solution.        
\end{proposition}

Consequently, we pose a conjecture that in analogy to the classical case, it is not possible to construct a uniquely solvable grid using fewer clues.
\begin{conjecture}
    Uniquely solvable grids of $4 \times 4$ SudoQ admit at least 4 clues.
\end{conjecture}

These results and remarks show that the understanding of low-dimensional SudoQ and quantum Latin squares is fairly comprehensive.
In the next section, we shall consider the case of general dimensions.

\section{Construction of SudoQ of the maximal cardinality}\label{sec:sudoQ_maximal_cardinality_general}
The most interesting problem from the point of view of the newly introduced notion of cardinality is whether it is useful in describing SudoQ designs or quantum Latin squares.
Therefore, we focused our research on answering the question of the existence of SudoQ designs of the maximal cardinality for every dimension.
This section shall be devoted to describing the path leading to the general construction.
Let us start by providing a useful notation.
Instead of treating SudoQ as an array of quantum states, one can understand this as a collection of vectors satisfying certain orthogonality conditions.

\begin{remark}\label{remark:SudoQ_notation}
    SudoQ of size $N^2$ is a set of $N^4$ vectors $\{\ket{v_{ijkl}\}}$, where $i,j,k,l \in \{1,...,N\}$, such that for a fixed pair of indices i,j or i,k or j,l the resulting $N^2$ vectors form a basis of the Hilbert space $\mathcal{H}_{N^2}$.
\end{remark}

Using the above notation, it is natural to construct a classical SudoQ via cyclic permutations of rows and columns.
This can be done in such a way that the element with indices $(i,j,k,l)$ will denote the entry $\ket{v_{ijkl}} = \ket{i\oplus j}\otimes \ket{k\oplus j}$, where $\oplus$ denotes addition modulo $N$.
The tensor product structure is trivial -- a computational basis of $\mathcal{H}_{N^2}$ is obtained by combining computational bases from two Hilbert spaces $\mathcal{H}_N$.
The construction presented above allows one to obtain a Sudoku design or, equivalently, a classical SudoQ design.
However, using the notation of sets of vectors it is possible to construct SudoQ designs of the maximal cardinality, employing families of unitary matrices of size $N$.

\begin{proposition}\label{prop:maximal_cardinality1}
    Consider $\{U_i\}^N_{i=1}$ and $\{V_i\}^N_{i=1}$, which are two families of unitary matrices of size $N$.
    Let the $k$-th rows of $U_i$ and $V_j$ be denoted by $\ket{u^{(i)}_{k}}$ and $\ket{w^{(j)}_k}$, respectively.
    Then, the following set of vectors forms a proper SudoQ design of size $N^2$,
    \begin{equation}\label{eq:proof_families_sudoku}
        \ket{v_{ijkl}} \coloneqq \ket{u^{(i)}_{j\oplus k}} \otimes \ket{w^{(j)}_{i\oplus l}},
    \end{equation}
    where $\oplus$ denotes addition modulo $N$.
\end{proposition}
\begin{proof}
The proof is based on the demonstration that the appropriate sets form a basis, as per Remark~\ref{remark:SudoQ_notation}.
First of all, for any fixed parameters $i$ and $j$, the vectors $\{\ket{u^{(i)}_{j\oplus k}}\}_{k=1}^N$ form consecutive rows of a unitary matrix $U_i$; therefore, form an orthonormal basis of $\mathcal{H}_N$.
Applying similar reasoning to $\{\ket{w^{(j)}_{i\oplus l}}\}_{l=1}^N$, we obtain that $\{\ket{v_{ijkl}}\}$ form a basis of $\mathcal{H}_{N^2}$ for any fixed $i$ and $j$.

Further, for any fixed $i$ and $k$ the vectors $\{\ket{u^{(i)}_{j\oplus k}}\}_{j=1}^N$ form a basis of $\mathcal{H}_N$.
Combining it with the previously considered fact that for fixed $i$ and $j$ also $\{\ket{w^{(j)}_{i\oplus l}}\}_{l=1}^N$ form a basis, we conclude that $\{\ket{v_{ijkl}}\}$ form a basis of $\mathcal{H}_{N^2}$ for any fixed $i$ and $k$.
Analogous reasoning can be conducted to the final possible pair of fixed indices $j$ and $l$, and, as a result, the set of vectors is proved to be a valid SudoQ design.
\end{proof}

Finally, we shall present the main result of this section in the form of two statements that use the previous notation.
Whenever we refer to vectors from the family $U_i$ or $V_i$, we shall mean the $N$ bases formed by its rows, so in total $N^2$ vectors.
Two vectors are said to be equal if they differ only by a complex phase.

\begin{proposition}\label{prop:maximal_cardinality2}
    If $c_U$ and $c_V$ stand for the number of distinct vectors in the families of bases $\{U_i\}$ and $\{V_i\}$, then the cardinality of the SudoQ design constructed in Eq.~(\ref{eq:proof_families_sudoku}) reads $c=c_U c_V$.
\end{proposition}
\begin{proof}
First, let us note that by a rearrangement of indices it is easy to prove the equality of the following two sets
\begin{equation}
    \{\ket{v_{ijkl}}\}_{ijkl} = \{\ket{u^{(i)}_{k}}\otimes \ket{w^{(j)}_{l}}\}_{ijkl}.
\end{equation}

Therefore, using the properties of the tensor product, we conclude that if the numbers of distinct vectors in $\{U_i\}$ and $\{V_{i}\}$ families are $c_U$ and $c_V$ then the cardinality of the resulting SudoQ design equals to $c_U c_V$.
\end{proof}

In each dimension it is straightforward to create a valid SudoQ design of the maximal cardinality, as it is given by any two families of random unitary matrices with probability 1.

\begin{theorem}\label{thm:sudoQ_maximal_cardinality}
    SudoQ designs of the maximal cardinality exist for all dimensions $N^2$.    
\end{theorem}
\begin{proof}
    The probability that rows of two random unitary matrices from the circular unitary ensemble are identical is equal to zero.
    Therefore, to prove the theorem, it suffices to take two families of random matrices, which do not share any rows with probability 1.
\end{proof}

As a result, we conclude that the cardinality of a SudoQ design is a non-trivial notion.
For every $N$ there exist designs of minimal cardinalities (classical SudoQ, $c=N^2$), as well as the maximal cardinalities ($c=N^4$).
The latter part of the chapter shall be devoted to the study of a different perspective on constructing SudoQ designs and its relation to mutually unbiased bases.

\section{SudoQ and mutually unbiased bases}
Another approach to the creation of the solutions admitting the maximal cardinality was initiated by Adam Burchardt.
He observed the similarities of the construction of SudoQ to the construction of another important quantum mechanical notion of mutually unbiased bases (MUBs).

\begin{definition}[MUBs]
    The bases $\mathcal{B}_1$,...,$\mathcal{B}_d$ of the Hilbert space $\mathcal{H}_N$ are called mutually unbiased if any two vectors $\ket{\phi}$ and $\ket{\psi}$ belonging to different bases $\ket{\phi} \in \mathcal{B}_i$ and $\ket{\psi} \in \mathcal{B}_j$ have the same product, $\braket{\phi|\psi} = \frac{1}{N}$.
\end{definition}

The name of mutually unbiased stems from the fact that if one possess multiple copies of a given quantum state, then performing the measurement in the particular basis does not give \emph{any} information about the projective measurements in the other bases.
The number of MUBs in a given dimension $N$ cannot exceed the value of $N+1$.
The existence and applications of MUBs have been recently widely studied; $N+1$ MUBs can be constructed in all prime dimensions $N$ as well as in all dimensions that are power of primes, $N^q$~\cite{Bandyopadhyay_2001}, e.g.\ using the Heisenberg-Weyl method~\cite{Brierley_2010,Hiesmayr_2021}.
Furthermore, applications of MUBs span various subfields of quantum information, e.g.\ vectors from the full set of MUBs form a projective 2-design~\cite{Welch_1974,Klappenecker_2005}, which is yet another connection to the field of quantum designs initiated by Zauner~\cite{Zauner_German}.
What is more, they can be used in quantum state tomography~\cite{Wooters_1987}, as well as in various cryptographic protocols~\cite{Cerf_2002}.

The relation of MUBs to the problem of constructing a valid SudoQ design can be seen by analyzing the Heisenberg-Weyl MUBs.
In order to create them, it is beneficial to introduce the generalized Pauli matrices of size $N$, acting as operators on the Hilbert space $\mathcal{H}_N$
\begin{equation}
    X = \sum^N_{j=1} \ket{j}\bra{j\oplus 1}, \quad Z = \sum^N_{j=1} \omega^j \ket{j}\bra{j},
\end{equation}
where $\omega$ stands for the $N$-th root of unity, $\omega = e^{2\pi i/N}$.
As before, addition inside the bra-ket notation is understood modulo $N$.
Then, $N(N+1)$ eigenvectors of the following $(N+1)$ operators
\begin{equation}
    Z, X, XZ, XZ^2, ..., XZ^{N-1}
\end{equation}
form the full set of $(N+1)$ mutually unbiased bases in the prime dimension $N$.

The eigenvectors of Weyl-Heisenberg operators can also be utilized to construct a SudoQ design of size $N^2$ of the maximal cardinality and of a particular symmetry, provided that the number $N$ is prime (for more details, see the joint paper~\cite{Rajchel_SudoQ}).
As an example, in Fig.~\ref{fig:SudoQ_9x9} we present a genuinely quantum $3^2\times 3^2 = 9\times 9$ SudoQ design of the maximal cardinality $c=81$, using eigenvectors of $Z$ (denoted by $\ket{i_1}$), of $X$ ($\ket{i_2}$), and of $XZ$ ($\ket{i_3}$)

\begin{equation}
\begin{split}
&\text{$Z$ eigenvectors:}\quad \ket{1_1}=\ket{1}, \quad
\ket{2_1}=\ket{2}, \quad
\ket{3_1}=\ket{3} ,\\
&\text{$X$:}\quad {\color{blue}\ket{1_2}}=\frac{1}{\sqrt{3}} (\ket{1}+\ket{2} +\ket{3}), \quad
{\color{blue}\ket{2_2}}=\frac{1}{\sqrt{3}} (\ket{1}+ \omega \ket{2} + \omega^2\ket{3}), \quad
{\color{blue}\ket{3_2}}=\frac{1}{\sqrt{3}} (\ket{1}+ \omega^2 \ket{2} + \omega\ket{3}) ,\\
&\text{$XZ$:}\quad {\color{red}\ket{1_3}}=\frac{1}{\sqrt{3}} (\ket{1}+\omega\ket{2} +\omega\ket{3}), \quad
{\color{red}\ket{2_3}}=\frac{1}{\sqrt{3}} (\ket{1}+ \omega^2 \ket{2} + \ket{3}), \quad
{\color{red}\ket{3_3}}=\frac{1}{\sqrt{3}} (\ket{1}+ \ket{2} + \omega^2 \ket{3}),
\end{split}
\end{equation}
where $\omega$ denotes the root of unity, $\omega= e^{2\pi i/3}$.
The above 9 vectors form eigenvectors of the Weyl-Heisenberg operators. Furthermore, they allow one to build identity matrix and two complex Hadamard matrices
\begin{equation}
\begin{bmatrix}
    1 & 0 & 0 \\
    0 & 1& 0 \\
    0 & 0 & 1
    \end{bmatrix},
    \quad\quad
\frac{1}{\sqrt{3}}\begin{bmatrix}
    1 & 1 & 1 \\
    1 & \omega & \omega^2 \\
    1 & \omega^2 & \omega
    \end{bmatrix},
\quad \quad
\frac{1}{\sqrt{3}}\begin{bmatrix}
    1 & 1 & 1 \\
    \omega & \omega^2 & 1 \\
    \omega & 1 & \omega^2
    \end{bmatrix}.
\end{equation} 

\begin{figure}[H]
    \includegraphics[scale=1.12]{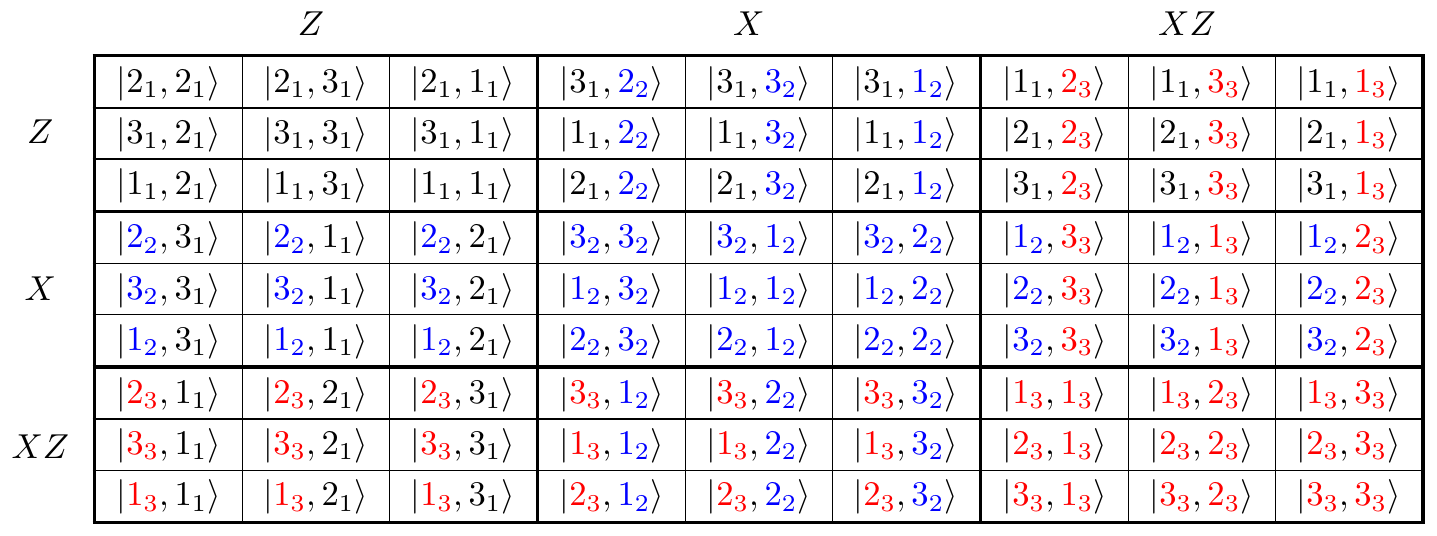}
    \caption{A $9\times 9$ genuinely quantum SudoQ of the maximal cardinality $c=81$, constructed using Weyl-Heisenberg operators, where the vectors forming the tensor products form three triples of mutually unbiased bases of size 9.
    This more scrupulous version of the photo from the beginning of the chapter (Fig.~\ref{fig:sudoQ_pen}) describes $3N^2 = 27$ orthogonal measurements.
    }
    \label{fig:SudoQ_9x9}
\end{figure}

Finally, let us note that the example of a SudoQ design of size 4 of the maximal cardinality, presented in design~(\ref{eq:SudoQ_cardinality16}), can also be constructed by means of Pauli operators, giving rise to two pairs of mutually unbiased bases of size 4.

\section{Conclusions}
The main achievement of our work was the introduction and characterization of genuinely quantum SudoQ, as well as its crucial characteristics, namely cardinality.
Furthermore, the connection to another type of quantum designs, i.e. mutually unbiased bases was discovered, together with the new application for the construction of MUBs.
Similarly as in the case of the classical Sudoku, investigation of quantum Sudoku designs posed many questions that have not yet been answered.
We discuss these questions in detail in Chapter~\ref{Summary}.
The author hopes that the SudoQ designs will be developed from the theoretical perspective in the future, while experimental features of quantum designs will find some applications in the laboratory.


\clearpage
\chapter{Concluding remarks}
\label{Summary}
\vspace{-1cm}
\rule[0.5ex]{1.0\columnwidth}{1pt} \\[0.2\baselineskip]

\section{Summary of the thesis}
The main goals of the research described in this thesis can be summarized in five points.
\begin{enumerate}
\renewcommand{\theenumi}{\Alph{enumi}} 
    \item Application of quantum combinatorics for finding designs in the form of pure states -- under this general name we encompass the whole area of quantum designs studied in this paper, e.g.\ the golden AME(4,6) state from Section~\ref{sec:analytical_AME} and the quantum Sudoku of the non-classical cardinality, given by Eq.~(\ref{eq:SudoQ_cardinality16}).
    This study established several novel ideas that might be useful for the research in other quantum setups (see Sections~\ref{sec:A_family}\,-\,\ref{sec:search_block_structure}).
    \item Introduction of the notion of \emph{genuinely quantum designs} -- the most important design is the absolutely maximally entangled state of four quhexes.
    Due to the non-existence of two orthogonal Latin squares of size 6 (see Theorem~\ref{theorem:OLS}) in this setup there are no classical counterparts, making it impossible to succeed in the search for AME(4,6) state without the new notion.    
    Furthermore, genuinely quantum designs can be helpful in devising quantum measurements of new properties. 
    This matter is discussed more thoroughly in Section~\ref{sec:quantum_sudoku}.
    \item Characterization of tripartite orthogonal gates by evaluating analytically their average entangling power, see Section~\ref{sec:average_e_p_tripartite_orthogonal}.
    \item A general numerical solution to the problem of unistochasticity in dimension 4 (Section~\ref{sec:Uffe_algorithm}).
    In addition, the proof that the analytical bracelet conditions for the circulant bistochastic matrices of the same size are sufficient for their unistochasticity, as presented in Section~\ref{sec:circulant_matrices}.
    \item Application of the newly found family of robust Hadamard matrices (Section~\ref{sec:robust_Hadamard_matrices}) that are relevant for the problem of unistochasticity in general dimensions.
    This also solves certain tasks in quantum information, e.g.\ proves the existence of equi-entangled bases discussed in Section~\ref{sec:equi-entangled}.
\end{enumerate}

Furthermore, the author notes that even with the broad areas of work covered by this thesis, it does not answer nor explain all open questions in the field of quantum mappings and designs, as shall be discussed in Section~\ref{sec:open_problems}.


\section{Open problems}\label{sec:open_problems}
The author believes that the following open questions remain the most interesting from the point of view of prospect research.
\begin{enumerate}
    \item Study of the unistochasticity using similar tools as introduced in Chapter~\ref{chapter_3}.
    As an example, robust Hadamard matrices might be applied to solve other problems in quantum information that rely on unistochasticity.
    In fact, a reverse problem might be interesting -- can we deduce the existence of robust Hadamard matrices given that a ray of the Birkhoff polytope is unistochastic?
    A positive answer to this question might mean that using the Birkhoff polytope, one can try to demonstrate the validity of the Hadamard conjecture, concerning the existence of Hadamard matrices in any dimension $N$ divisible by 4.
    \item Extension of the algorithm provided for the unistochasticity problem (see Section~\ref{sec:Uffe_algorithm}) to dimensions higher than 4.
    It is possible that a similar algorithm, employing the internal block structure exists -- though its principle might be based on less general assumptions than the particular instance of 4 blocks of size $2\times 2$.
    \item A more general version of Theorem~\ref{thm:circulant_size4}, concerning circulant matrices $C$ of arbitrary size.
    Employing a similar idea, i.e.\ the similarity between different elements of the product $CC^\dagger$, one may try to prove that also in the case of $5\times 5$ matrices the bracelet conditions are sufficient for unistochasticity.
    However, in this case, the task is much harder due to the existence of three distinct off-diagonal terms instead of two as in the $4\times 4$ case.
    \item Application of the setup for measuring entangling power (Chapter~\ref{chapter_4}) to other measures of multipartite entanglement.
    Furthermore, it is possible that using a similar rationale one can find also higher moments of statistical distributions of the entangling power $e_p$.
    \item From the experimental physics perspective, one is tempted to try to devise an experimental setup that is capable of measuring the average entangling power in a similar way as done by the Google team~\cite{Arute_2019}.
    It is possible that by going further in this direction it will be possible to verify quantum advantage in the domain of entanglement.
    \item The search for AME(4,6) state enabled us to find many intermediate solutions that might find their applications to study other quantum setups. 
    An obvious question arises of whether similar reasoning might be used to prove/disprove the existence of absolutely maximally entangled states for other numbers of parties.
    \item What is more, on the way to AME(4,6) state we have found a number of interesting families of unitary matrices and the corresponding states that possess a high degree of symmetry ($A$, $G$, and $W$ families).
    An open question remains, whether they can be of direct use in a laboratory due to the small number of non-zero elements -- the best classical setup with relatively low entangling power has 36 of them.
    These new families exhibit 40, 42, and 46 non-zero elements respectively, while the golden AME state has many more -- 112 non-zero terms.
    Depending on the application, it is possible that researchers may need to optimize the quality of an approximation versus the complexity degree of the resource under investigation, rendering novel families useful.
    A similar problem concerns the extent to which the $e_p/g_t$ plane is covered by unitary matrices -- in the case of smaller-sized matrices there are gaps near the AME state, compare with \cite{Jonnadula_2020}.
    \item The block-like structure and its application to study other multiunitary matrices.
    This structure, introduced in Section~\ref{sec:block-like_AME}, might prove relevant to other setups in which the existence of a highly symmetric multiunitary matrix is not decided.
    Since multiunitary matrices are important not only to the study of AME states but they lead to perfect tensors, this path could be tempting to future researchers. 
    \item The existence of the real AME(4,6) state.
    The golden AME(4,6) state found in this thesis is determined by a unitary matrix with complex phases.
    However, the intermediate states linked to the orthogonal matrices $A$, $G$, and $W$ were described by real coefficients.
    It is interesting to verify whether such a state exists also with only real coefficients since this might lead to easier experimental manipulation and would give more insight into the nature of genuinely quantum designs. 
    \item The application of the golden AME state.
    The new AME state provides simultaneously positive answers to several interesting questions from the perspective of quantum information, e.g.\ in encoding the system described by 6 internal classical states.
    Nonetheless, due to the novel nature of the genuinely quantum AME(4,6) state, we believe that it might be useful in other experimental tasks.
    \item The existence of nonequivalent AME(4,6) states.
    The numerical algorithm presented in Section~\ref{sec:Suhail_alg} allowed us to find various examples of AME(4,6) states.
    However, it is not clear whether they are equivalent, i.e.\ whether there exist LU operations transforming them into each other. 
    \item The applications of genuinely quantum designs in other contexts of quantum information.
    In this thesis, we studied their role mostly restricting to the case of AME states and the envisagement of experiments involving few distinct orthogonal states.
    Nonetheless, due to their versatile structure, it is possible that, similar to the case of classical orthogonal arrays, they shall play role in various quantum mechanical setups.
    \item The minimal number of clues in a quantum Sudoku.
    In 2014, it was proved by McGuire et al.~\cite{Mcguire_2014} that the minimal number of clues for the Sudoku of dimension 9 is 17.
    Nonetheless, a similar problem in the quantum domain, posed by Nechita and Pillet~\cite{Nechita_2020}, is still unanswered.
    \item The full characterization of cardinality for quantum Latin squares.
    In this thesis, we presented the result showing that in any dimension $N^2$ a SudoQ of the maximal cardinality exists (Theorem~\ref{thm:sudoQ_maximal_cardinality}).
    Nonetheless, apart from the $4\times 4$ case (Theorem~\ref{thm:SudoQ_cardinalities4x4}), a full description of admissible cardinalities in the general case of quantum Latin square of an arbitrary size $N$ is missing. 
\end{enumerate}

To summarize, the contribution presented in the thesis aimed not to find new physics but rather to widen our understanding of the already existing theories, especially quantum information with distinguishable particles.
Additionally, the research explored the branch of quantum combinatorics which is currently under development, while having a lot of theoretical as well as experimental prospects.
We hope that genuinely quantum designs introduced in this thesis shall prove useful in constructing new quantum protocols, fully utilizing \emph{quantumness}.


\clearpage

\printbibliography


\begin{appendices}

\chapter{Vector of derivatives of the average singular entropy}\label{app:hessian_s_e}
In Section~\ref{sec:sum_of_entropies_AME} we argued that matrices under consideration ($W$, $X_i$, $X_d$, and $A$) are local maxima of the average singular entropy.
Mathematical backgrounds for this conjecture are given by the present appendix, which includes evaluation of the derivative vector for an arbitrary matrix $X$ of size 36.
This analytical result reinforced by numerical calculations in Mathematica programming language, lead us to the conclusion that matrices of interest have vanishing derivatives. 

Recall the definition of the average singular entropy $s_e$ of a matrix $X$ 
\begin{equation}
    s_e(X) = \frac{1}{3} \bigg(E_S(X) + E_S(X^R) + E_S(X^\Gamma)\bigg),
\end{equation}
where $E_S(X)$ stands for the singular entropy of $X$, defined in Eq.~(\ref{eq:definition_singular_entropy}).

In order to evaluate the vector of derivatives $\nabla s_e(X)$ it will be beneficial to focus on the components of the function, namely singular entropies.
Thus, we shall start by investigating the vector of derivatives for the singular entropy.
Using notation similar to the one in Section~\ref{sec:local_maxima}, we need to evaluate the term
\begin{equation*}
    \frac{\mathrm{Tr}\big(X_\varepsilon X^\dagger_\varepsilon  X_\varepsilon X^\dagger_\varepsilon \big)}{\mathrm{Tr}^2\big(X_\varepsilon X^\dagger_\varepsilon \big)},
\end{equation*}
as well as the corresponding terms for reshuffling and partial transposition.
Note that the non-trivial denominator makes the calculations somewhat more complicated.
The singular entropy of the above term, up to the linear order, reads
\begin{equation}\label{eq:delta_singular_entropy}
    \begin{split}
    E_S(X_\varepsilon) 
    = \frac{N}{N-1} \bigg(1-\frac{\text{Tr}\big(XX^\dagger X X^\dagger+ \delta X X^\dagger X X^\dagger + X\delta X^\dagger X X^\dagger+ XX^\dagger\delta X X^\dagger+ XX^\dagger X \delta X^\dagger\big)}{\text{Tr}^2\big(XX^\dagger+\delta X X^\dagger+ X \delta X^\dagger \big)} \bigg).
    \end{split}
\end{equation}

Then, we simplify the denominator by the expansion correct to the linear order $\frac{1}{1+\epsilon} \approx 1-\epsilon$, leading to the formula valid while neglecting higher than linear terms in $\delta X$
\begin{equation}
    \frac{1}{\text{Tr}^2\big(XX^\dagger+\delta X X^\dagger+ X \delta X^\dagger \big)} 
    \approx \frac{1}{\text{Tr}^2\big(XX^\dagger\big)}\bigg(1-\frac{2\;\text{Tr}\big(\delta X X^\dagger+ X \delta X^\dagger \big)}{\text{Tr}\big(XX^\dagger\big)}\bigg).
\end{equation}

Combining this result with Eq.~(\ref{eq:delta_singular_entropy}), we arrive at
\begin{equation}
    \begin{split}
        E_S(X_\varepsilon) = E_S(X) + \frac{4N}{N-1}\; \mathrm{Re}\;\mathrm{Tr}\bigg( \frac{\delta X X^\dagger \big(\xi \mathbb{I}-XX^\dagger\big)}{\mathrm{Tr}^2\big(XX^\dagger \big)}\bigg),
    \end{split}
\end{equation}
where $\xi$ stands for $\frac{\text{Tr}\big(XX^\dagger X X^\dagger\big)}{\text{Tr}\big(XX^\dagger\big)}$

Utilizing any basis of the underlying set of complex matrices of size $N^2$, written as $\varepsilon_i M_i$ where $i \in \{1,...,2N^4\}$, the expression for the partial derivative of the singular entropy reads
\begin{equation}
    \nabla_i \; E_S(X) = \frac{4N}{N-1}\; \mathrm{Re}\;\mathrm{Tr}\bigg( \frac{M_i X^\dagger \big(\xi \mathbb{I}-XX^\dagger\big)}{\mathrm{Tr}^2\big(XX^\dagger \big)}\bigg), \,\,\, \mathrm{with }\,\,\, i = 1,...,2N^4.
\end{equation}

Extension of this formula to reshuffling and partial transposition is straightforward, yielding the ultimate formula for the partial derivatives of the average singular entropy
\begin{equation}
    \nabla_i\; s_e(X) = \frac{4N}{N-1}\; \mathrm{Re}\;\mathrm{Tr}\bigg( \frac{M_i X^\dagger \big(\xi \mathbb{I}-XX^\dagger\big)}{\mathrm{Tr}^2\big(XX^\dagger \big)} + \frac{M_i^R X^{R\dagger} \big(\xi_R \mathbb{I}-X^RX^{R\dagger}\big)}{\mathrm{Tr}^2\big(X^RX^{R\dagger} \big)} + \frac{M^\Gamma_i X^{\Gamma\dagger} \big(\xi_\Gamma \mathbb{I}-X^\Gamma X^{\Gamma\dagger}\big)}{\mathrm{Tr}^2\big(X^\Gamma X^{\Gamma\dagger} \big)}\bigg),
\end{equation}
where $\xi_R = \frac{\text{Tr}\big(X^R X^{R\dagger} X^R X^{R\dagger}\big)}{\text{Tr}\big(X^R X^{R\dagger}\big)} $ and $\xi_\Gamma = \frac{\text{Tr}\big(X^\Gamma X^{\Gamma\dagger} X^\Gamma X^{\Gamma\dagger}\big)}{\text{Tr}\big(X^\Gamma X^{\Gamma\dagger}\big)} $, which finalizes our search.
Now, we can move on to computing derivatives for the matrices of our interest.

In order to find partial derivatives, we must choose a particular basis $\{M_i\}_{i=1}^{2N^4}$.
Contrary to the case of unitary matrices, it is not possible to use the exponentiation of Hermitian matrices since this would lead to incomplete characterization of all directions.

The only assumption we made about the matrix $X$ of size $N^2$ was that it consists of complex numbers.
We should choose a basis that given real parameters $\varepsilon_i$ shall enable us to move in every direction in the set of complex matrices.
Therefore, the simplest possible choice involves $2N^4$ matrices of the form $M_j=\{\ket{k}\bra{l},\; i\ket{k}\bra{l}\}_{k,l=1}^{N^2}$, where $j\in \{1,...,N^4\}$ indexes the whole basis.
Nonetheless, the particular choice of the basis is irrelevant for analytical calculations while being necessary from the perspective of numerical computations.
The results of these computations, given in Table~\ref{table:4_categories}, show that this numerical technique allowed us to improve the quality of the approximation of the multiunitary matrix, starting with an exemplary matrix $G$.

\chapter{Block-like structure of the multiunitary matrix}\label{app:block_like}
Section~\ref{sec:Suhail_alg} introduced an iterative algorithm used to obtain a numerical multiunitary matrix, corresponding to an AME(4,6) state.
Then, in Section~\ref{sec:block-like_AME} it was argued that the particular form of this matrix allows reducing the number of parameters in the search for an analytic multiunitary matrix by a factor of 9.
This feat was accomplished using the remark that the non-zero vectors contained in one of the blocks defined in Table~\ref{tab:define_blocks} do not mix with the vectors from the other blocks.
By this we understand that each block of matrix $U$, its reshuffling $U^R$, or its partial transpose $U^\Gamma$ admits vectors only from a single set: AB1, AB5, AB9, CD1, CD5, CD9, EF1, EF5, or EF9.
Finally, it was explained that we determine only three blocks AB1, AB5, and AB9, and choose the other six to satisfy Eq.~(\ref{eq:block_equality}).
This procedure allows us to obtain a multiunitary matrix $U$, provided the three blocks AB satisfy the conditions given by Lemma~\ref{lemma:block_AME}.

Therefore, by optimizing only 24 vectors, instead of the full set of $24\times 3 = 72$ vectors, we might be able to obtain a simpler form of a multiunitary matrix.
The present appendix aims to visualize the proof of the Lemma~\ref{lemma:block_AME} by presenting the matrices $U^R$ and $U^\Gamma$; thereby, showing that the contents of blocks do not mix.
Tables~\ref{tab:U_Gamma} and \ref{tab:U_R} show these arrangements that, provided satisfying Eq.~(\ref{eq:block_equality}), prove the lemma.

\begin{table}[H]
$U^R=\begin{pmatrix}
\begin{tabular}{p{2cm}p{2cm}p{2cm}p{2cm}p{2cm}p{2cm}}
\cline{1-2}
\multicolumn{1}{|l|}{$a_1$} & \multicolumn{1}{l|}{$a_2$} &                       &                       &                       &                       \\ \cline{1-2}
\multicolumn{1}{|l|}{$b_1$} & \multicolumn{1}{l|}{$b_2$} &                       &                       &                       &                       \\ \cline{1-4}
                       & \multicolumn{1}{l|}{} & \multicolumn{1}{l|}{$c_1$} & \multicolumn{1}{l|}{$c_2$} &                       &                       \\ \cline{3-4}
                       & \multicolumn{1}{l|}{} & \multicolumn{1}{l|}{$d_1$} & \multicolumn{1}{l|}{$d_2$} &                       &                       \\ \cline{3-6} 
                       &                       &                       & \multicolumn{1}{l|}{} & \multicolumn{1}{l|}{$e_1$} & \multicolumn{1}{l|}{$e_2$} \\ \cline{5-6} 
                       &                       &                       & \multicolumn{1}{l|}{} & \multicolumn{1}{l|}{$f_1$} & \multicolumn{1}{l|}{$f_2$} \\ \cline{1-2} \cline{5-6} 
\multicolumn{1}{|l|}{$a_3$} & \multicolumn{1}{l|}{$a_4$} &                       &                       &                       &                       \\ \cline{1-2}
\multicolumn{1}{|l|}{$b_3$} & \multicolumn{1}{l|}{$b_4$} &                       &                       &                       &                       \\ \cline{1-4}
                       & \multicolumn{1}{l|}{} & \multicolumn{1}{l|}{$c_3$} & \multicolumn{1}{l|}{$c_4$} &                       &                       \\ \cline{3-4}
                       & \multicolumn{1}{l|}{} & \multicolumn{1}{l|}{$d_3$} & \multicolumn{1}{l|}{$d_4$} &                       &                       \\ \cline{3-6} 
                       &                       &                       & \multicolumn{1}{l|}{} & \multicolumn{1}{l|}{$e_3$} & \multicolumn{1}{l|}{$e_4$} \\ \cline{5-6} 
                       &                       &                       & \multicolumn{1}{l|}{} & \multicolumn{1}{l|}{$f_3$} & \multicolumn{1}{l|}{$f_4$} \\ \cline{3-6} 
                       & \multicolumn{1}{l|}{} & \multicolumn{1}{l|}{$a_5$} & \multicolumn{1}{l|}{$a_6$} &                       &                       \\ \cline{3-4}
                       & \multicolumn{1}{l|}{} & \multicolumn{1}{l|}{$b_5$} & \multicolumn{1}{l|}{$b_6$} &                       &                       \\ \cline{3-6} 
                       &                       &                       & \multicolumn{1}{l|}{} & \multicolumn{1}{l|}{$c_5$} & \multicolumn{1}{l|}{$c_6$} \\ \cline{5-6} 
                       &                       &                       & \multicolumn{1}{l|}{} & \multicolumn{1}{l|}{$d_5$} & \multicolumn{1}{l|}{$d_6$} \\ \cline{1-2} \cline{5-6} 
\multicolumn{1}{|l|}{$e_5$} & \multicolumn{1}{l|}{$e_6$} &                       &                       &                       &                       \\ \cline{1-2}
\multicolumn{1}{|l|}{$f_5$} & \multicolumn{1}{l|}{$f_6$} &                       &                       &                       &                       \\ \cline{1-4}
                       & \multicolumn{1}{l|}{} & \multicolumn{1}{l|}{$a_7$} & \multicolumn{1}{l|}{$a_8$} &                       &                       \\ \cline{3-4}
                       & \multicolumn{1}{l|}{} & \multicolumn{1}{l|}{$b_7$} & \multicolumn{1}{l|}{$b_8$} &                       &                       \\ \cline{3-6} 
                       &                       &                       & \multicolumn{1}{l|}{} & \multicolumn{1}{l|}{$c_7$} & \multicolumn{1}{l|}{$c_8$} \\ \cline{5-6} 
                       &                       &                       & \multicolumn{1}{l|}{} & \multicolumn{1}{l|}{$d_7$} & \multicolumn{1}{l|}{$d_8$} \\ \cline{1-2} \cline{5-6} 
\multicolumn{1}{|l|}{$e_7$} & \multicolumn{1}{l|}{$e_8$} &                       &                       &                       &                       \\ \cline{1-2}
\multicolumn{1}{|l|}{$f_7$} & \multicolumn{1}{l|}{$f_8$} &                       &                       &                       &                       \\ \cline{1-2} \cline{5-6} 
                       &                       &                       & \multicolumn{1}{l|}{} & \multicolumn{1}{l|}{$a_9$} & \multicolumn{1}{l|}{$a_{10}$} \\ \cline{5-6} 
                       &                       &                       & \multicolumn{1}{l|}{} & \multicolumn{1}{l|}{$b_9$} & \multicolumn{1}{l|}{$b_{10}$} \\ \cline{1-2} \cline{5-6} 
\multicolumn{1}{|l|}{$c_9$} & \multicolumn{1}{l|}{$c_{10}$} &                       &                       &                       &                       \\ \cline{1-2}
\multicolumn{1}{|l|}{$d_9$} & \multicolumn{1}{l|}{$d_{10}$} &                       &                       &                       &                       \\ \cline{1-4}
                       & \multicolumn{1}{l|}{} & \multicolumn{1}{l|}{$e_9$} & \multicolumn{1}{l|}{$e_{10}$} &                       &                       \\ \cline{3-4}
                       & \multicolumn{1}{l|}{} & \multicolumn{1}{l|}{$f_{9}$} & \multicolumn{1}{l|}{$f_{10}$} &                       &                       \\ \cline{3-6} 
                       &                       &                       & \multicolumn{1}{l|}{} & \multicolumn{1}{l|}{$a_{11}$} & \multicolumn{1}{l|}{$a_{12}$} \\ \cline{5-6} 
                       &                       &                       & \multicolumn{1}{l|}{} & \multicolumn{1}{l|}{$b_{11}$} & \multicolumn{1}{l|}{$b_{12}$} \\ \cline{1-2} \cline{5-6} 
\multicolumn{1}{|l|}{$c_{11}$} & \multicolumn{1}{l|}{$c_{12}$} &                       &                       &                       &                       \\ \cline{1-2}
\multicolumn{1}{|l|}{$d_{11}$} & \multicolumn{1}{l|}{$d_{12}$} &                       &                       &                       &                       \\ \cline{1-4}
                       & \multicolumn{1}{l|}{} & \multicolumn{1}{l|}{$e_{11}$} & \multicolumn{1}{l|}{$e_{12}$} &                       &                       \\ \cline{3-4}
                       & \multicolumn{1}{l|}{} & \multicolumn{1}{l|}{$f_{11}$} & \multicolumn{1}{l|}{$f_{12}$} &                       &                       \\ \cline{3-4}
\end{tabular}
\end{pmatrix}$
\caption{
A general form of $U^R$, which is the reshuffling of the numerical multiunitary matrix $U$ of size 36, see Table~\ref{tab:U_block_vectors}.
Vectors $a_1,...,f_{12}$ of length 6 are non-zero, every other blank vector consists of entries equal to 0.
}\label{tab:U_R}
\end{table}

\begin{table}[H]
$U^\Gamma=\begin{pmatrix}
\begin{tabular}{p{2cm}p{2cm}p{2cm}p{2cm}p{2cm}p{2cm}}
\cline{1-2}
\multicolumn{1}{|l|}{$a_1$} & \multicolumn{1}{l|}{$a_3$} &                       &                       &                       &                       \\ \cline{1-2}
\multicolumn{1}{|l|}{$a_2$} & \multicolumn{1}{l|}{$a_4$} &                       &                       &                       &                       \\ \cline{1-4}
                       & \multicolumn{1}{l|}{} & \multicolumn{1}{l|}{$a_5$} & \multicolumn{1}{l|}{$a_7$} &                       &                       \\ \cline{3-4}
                       & \multicolumn{1}{l|}{} & \multicolumn{1}{l|}{$a_6$} & \multicolumn{1}{l|}{$a_8$} &                       &                       \\ \cline{3-6} 
                       &                       &                       & \multicolumn{1}{l|}{} & \multicolumn{1}{l|}{$a_9$} & \multicolumn{1}{l|}{$a_{11}$} \\ \cline{5-6} 
                       &                       &                       & \multicolumn{1}{l|}{} & \multicolumn{1}{l|}{$a_{10}$} & \multicolumn{1}{l|}{$a_{12}$} \\ \cline{1-2} \cline{5-6} 
\multicolumn{1}{|l|}{$b_1$} & \multicolumn{1}{l|}{$b_3$} &                       &                       &                       &                       \\ \cline{1-2}
\multicolumn{1}{|l|}{$b_2$} & \multicolumn{1}{l|}{$b_4$} &                       &                       &                       &                       \\ \cline{1-4}
                       & \multicolumn{1}{l|}{} & \multicolumn{1}{l|}{$b_5$} & \multicolumn{1}{l|}{$b_7$} &                       &                       \\ \cline{3-4}
                       & \multicolumn{1}{l|}{} & \multicolumn{1}{l|}{$b_6$} & \multicolumn{1}{l|}{$b_8$} &                       &                       \\ \cline{3-6} 
                       &                       &                       & \multicolumn{1}{l|}{} & \multicolumn{1}{l|}{$b_9$} & \multicolumn{1}{l|}{$b_{11}$} \\ \cline{5-6} 
                       &                       &                       & \multicolumn{1}{l|}{} & \multicolumn{1}{l|}{$b_{10}$} & \multicolumn{1}{l|}{$b_{12}$} \\ \cline{5-6} 
                       &                       &                       & \multicolumn{1}{l|}{} & \multicolumn{1}{l|}{$c_9$} & \multicolumn{1}{l|}{$c_{11}$} \\ \cline{5-6} 
                       &                       &                       & \multicolumn{1}{l|}{} & \multicolumn{1}{l|}{$c_{10}$} & \multicolumn{1}{l|}{$c_{12}$} \\ \cline{1-2} \cline{5-6} 
\multicolumn{1}{|l|}{$c_1$} & \multicolumn{1}{l|}{$c_3$} &                       &                       &                       &                       \\ \cline{1-2}
\multicolumn{1}{|l|}{$c_2$} & \multicolumn{1}{l|}{$c_4$} &                       &                       &                       &                       \\ \cline{1-4}
                       & \multicolumn{1}{l|}{} & \multicolumn{1}{l|}{$c_5$} & \multicolumn{1}{l|}{$c_7$} &                       &                       \\ \cline{3-4}
                       & \multicolumn{1}{l|}{} & \multicolumn{1}{l|}{$c_6$} & \multicolumn{1}{l|}{$c_8$} &                       &                       \\ \cline{3-6} 
                       &                       &                       & \multicolumn{1}{l|}{} & \multicolumn{1}{l|}{$d_9$} & \multicolumn{1}{l|}{$d_{11}$} \\ \cline{5-6} 
                       &                       &                       & \multicolumn{1}{l|}{} & \multicolumn{1}{l|}{$d_{10}$} & \multicolumn{1}{l|}{$d_{12}$} \\ \cline{1-2} \cline{5-6} 
\multicolumn{1}{|l|}{$d_1$} & \multicolumn{1}{l|}{$d_3$} &                       &                       &                       &                       \\ \cline{1-2}
\multicolumn{1}{|l|}{$d_2$} & \multicolumn{1}{l|}{$d_4$} &                       &                       &                       &                       \\ \cline{1-4}
                       & \multicolumn{1}{l|}{} & \multicolumn{1}{l|}{$d_5$} & \multicolumn{1}{l|}{$d_7$} &                       &                       \\ \cline{3-4}
                       & \multicolumn{1}{l|}{} & \multicolumn{1}{l|}{$d_6$} & \multicolumn{1}{l|}{$d_8$} &                       &                       \\ \cline{3-4}
                       & \multicolumn{1}{l|}{} & \multicolumn{1}{l|}{$e_5$} & \multicolumn{1}{l|}{$e_7$} &                       &                       \\ \cline{3-4}
                       & \multicolumn{1}{l|}{} & \multicolumn{1}{l|}{$e_{6}$} & \multicolumn{1}{l|}{$e_{8}$} &                       &                       \\ \cline{3-6} 
                       &                       &                       & \multicolumn{1}{l|}{} & \multicolumn{1}{l|}{$e_{9}$} & \multicolumn{1}{l|}{$e_{11}$} \\ \cline{5-6} 
                       &                       &                       & \multicolumn{1}{l|}{} & \multicolumn{1}{l|}{$e_{10}$} & \multicolumn{1}{l|}{$e_{12}$} \\ \cline{1-2} \cline{5-6} 
\multicolumn{1}{|l|}{$e_1$} & \multicolumn{1}{l|}{$e_{3}$} &                       &                       &                       &                       \\ \cline{1-2}
\multicolumn{1}{|l|}{$e_{2}$} & \multicolumn{1}{l|}{$e_{4}$} &                       &                       &                       &                       \\ \cline{1-4}
                       & \multicolumn{1}{l|}{} & \multicolumn{1}{l|}{$f_{5}$} & \multicolumn{1}{l|}{$f_{7}$} &                       &                       \\ \cline{3-4}
                       & \multicolumn{1}{l|}{} & \multicolumn{1}{l|}{$f_{6}$} & \multicolumn{1}{l|}{$f_{8}$} &                       &                       \\ \cline{3-6} 
                       &                       &                       & \multicolumn{1}{l|}{} & \multicolumn{1}{l|}{$f_{9}$} & \multicolumn{1}{l|}{$f_{11}$} \\ \cline{5-6} 
                       &                       &                       & \multicolumn{1}{l|}{} & \multicolumn{1}{l|}{$f_{10}$} & \multicolumn{1}{l|}{$f_{12}$} \\ \cline{1-2} \cline{5-6} 
\multicolumn{1}{|l|}{$f_{1}$} & \multicolumn{1}{l|}{$f_{3}$} &                       &                       &                       &                       \\ \cline{1-2}
\multicolumn{1}{|l|}{$f_{2}$} & \multicolumn{1}{l|}{$f_{4}$} &                       &                       &                       &                       \\ \cline{1-2}
\end{tabular}
\end{pmatrix}$
\caption{
A general form of $U^\Gamma$, which is the partial transpose of the numerical multiunitary matrix $U$ of size 36, see Table~\ref{tab:U_block_vectors}.
Vectors $a_1,...,f_{12}$ of length 6 are non-zero, every other blank vector consists of entries equal to 0.
Here, we applied an additional transposition to keep the vectors in the form of rows.
}\label{tab:U_Gamma}
\end{table}

\end{appendices}




\end{document}